\documentclass[draft,noinfoline,nolinenumbers,notitlepage]{latex_common/ectaart}

\usepackage{amsfonts}
\usepackage[dvips]{graphicx}
\usepackage[round]{natbib}
\usepackage{anysize}
\usepackage{bm}
\usepackage{placeins}
\usepackage{amsmath, amsthm, amssymb}
\usepackage[bookmarksopen=true, 
            bookmarksopenlevel=1, 
                        colorlinks=true, linkcolor=blue, citecolor=blue, urlcolor=blue, 
]{hyperref}
\usepackage{tocloft}
\usepackage{booktabs}
\usepackage{multirow}

\allowdisplaybreaks 

\usepackage{subcaption}[=v1]

\graphicspath{{Figures_arxiv/}}

\newtheorem{lemma}{Lemma}

\newtheorem{theorem}{Theorem}

\papersize  {27.94cm} {21.59cm}
\marginsize {1in} {1in} {1in} {1in}

\usepackage[ruled,vlined]{algorithm2e}
\usepackage{setspace}

\SetKwInput{KwInput}{Input} 
\SetKwInput{KwOutput}{Output} 

\begin{document}

\title{Semiparametric Bayesian Estimation of Dynamic Discrete Choice Models}
\author{Andriy Norets \thanks{Department of Economics, Brown University; andriy$\_$norets@brown.edu} and Kenichi Shimizu \thanks{Department of Economics, University of Alberta; kenichi.shimizu@ualberta.ca } \thanks{This version: \today}}
\maketitle

\begin{abstract}
We propose a tractable semiparametric estimation method for structural dynamic discrete choice models.
The distribution of additive utility shocks in the proposed framework is modeled by location-scale mixtures of extreme value distributions with  varying numbers of mixture components. Our approach exploits the analytical tractability of extreme value distributions in the multinomial choice settings and the flexibility of the location-scale mixtures. 
We implement the Bayesian approach to inference using Hamiltonian Monte Carlo 
and an approximately optimal reversible jump algorithm.
In our simulation experiments, we show that the standard dynamic logit model can deliver misleading results, especially about counterfactuals, when 
the shocks are not extreme value distributed.
Our semiparametric approach delivers reliable inference in these settings.
We develop theoretical results on approximations by location-scale mixtures in an appropriate distance and 
posterior concentration of the set identified utility parameters and the distribution of shocks in the model.

\end{abstract}




\section{Introduction}
A dynamic discrete choice model is a dynamic program with discrete controls. These models have been used widely in various fields of economics, including labour economics, health economics, and industrial organization. See, for example, \citet{Rust_handbook:94} and \citet{Aguirregabiria_Mira:10} for literature surveys. In such models, a forward-looking decision-maker chooses an action from a finite set in each time period. The actions affect decision-maker's per-period payoff and the evolution of state variables. The decision-maker maximizes the expected sum of current and discounted future per-period payoffs. 

Some state variables in these models are usually assumed to be unobserved by the econometrician (see, for example, page 1008 of \citet{Rust:87} for further discussion). Most of the previous work on estimation of dynamic discrete choice models imposes specific parametric assumptions on the distribution of the unobserved states or utility shocks. 
The most commonly used parametric assumption is that the unobserved states are extreme value independently identically distributed (i.i.d.).
As shown in \cite{Rust:87} and \citet{Rust_handbook:94}, under this assumption, the integrals over the unobserved states in the likelihood and the Bellman equations have closed form expressions, which considerably 
alleviates the computational burden of the model solution and estimation.
At the same time, it is well known in the literature that imposing parametric distributional assumptions can be problematic, see, for example, \citet{Manski_book:99}. 
Thus, it is desirable to relax these assumptions if possible. 

There are several previous papers that treat the unobserved state distribution nonparametrically for the binary choice case. \citet{Aguirregabiria:10} shows the nonparametric identification of the shock distribution under particular assumptions on the per-period payoffs.  \citet{NoretsTang2013} show that under an unknown distribution of the unobserved state and discrete observed states, the utility parameters and the unobserved state distribution are only set-identified. They also show how to compute the identified sets. \cite{BuchholzShumXu:20}  provide identification results for the per-period payoffs when the observed state is continuous.
The framework of \cite{ChristensenConnault2023} for structural models expressed through a finite number of moment equalities and inequalities can be used to check the sensitivity of counterfactuals to the variation of the utility shocks distribution within a neighborhood of extreme value distribution in dynamic discrete choice models with a finite observed state space; Rust's binary choice model of bus engine replacement is used in that paper for illustration.

For the multinomial choice case, \citet{chen_2017} uses  exclusion restrictions (a subset of the state variables affects only current utility, but not state transition probabilities) to obtain identification and estimation results. In settings without exclusion restrictions, \citet{Norets_ddc_mult:11} shows that it is in principle possible to extend the method from \citet{NoretsTang2013} to compute the identified set in  multinomial case, but it is computationally very difficult.

In this paper, we propose a tractable semiparametric estimation method applicable to the general multinomial choice case. It is based on modeling the unknown distribution of shocks by a finite mixture of extreme value distributions with a varying number of mixture components.
Our approach exploits the analytical tractability of extreme value distributions and the flexibility of the location-scale mixtures.
The unobserved utility shocks can be integrated out analytically in the likelihood function and the expected value functions, similarly to the case with extreme value distributed shocks. At the same time, we show that the location-scale mixtures can approximate densities from a large nonparametric class in an appropriate distance 
and that for any given distribution of utility shocks, a finite mixture of extreme value distributions can deliver exactly the same conditional choice probabilities.
Posterior concentration on the identified sets of utility parameters and the distribution of shocks is an implication of these results.
We implement the Bayesian approach to inference for the model using Hamiltonian Monte Carlo 
and an approximately optimal reversible jump Markov chain Monte Carlo (MCMC) algorithm from \cite{Norets2017mcmc}. 
Similarly to \citet{NoretsTang2013}, frequentist confidence sets for identified sets can also be computed from the MCMC output.

We apply our framework to binary and multinomial choice models.
For the binary dynamic choice model from \cite{Rust:87}, our approach delivers estimation results that are consistent with the previous literature on 
semiparametric estimation (\citet{NoretsTang2013}). 
For the multinomial choice model of medical care use and work absence from \citet{Gilleskie_Eca:98}, 
we demonstrate how uncertainty about model parameters and counterfactuals increases when the distributional assumptions on the shocks are relaxed.
Moreover, we show that the standard dynamic logit model can deliver misleading results, especially about counterfactuals, when the shocks are not extreme value distributed.
Our semiparametric approach delivers reliable inference in these settings.

Even when the distribution of the utility shocks is assumed to be known, parameters and counterfactuals in dynamic discrete choice models could still be only set identified
under a variety of scenarios such as 
very flexible specifications of utility functions, lack of exclusion restrictions or variation in transitions for the observed state variables, and unknown time discount factors;
see, for example, \cite{Rust_handbook:94}, \cite{MagnacThesmar:02},
\citet{NoretsTang2013}, \cite{AbbringDaljord2020}, and \cite{KalouptsidiScottSouzaRodrigues2021}.
Our estimation framework 
does not require any special adjustments to accommodate such scenarios 
since parameters and counterfactuals are already set identified under 
nonparametric specification of the distribution of shocks.

The rest of the paper is organized as follows. Section \ref{section:setup} describes the general model setup. In Section \ref{section:semi}, we introduce our semiparametric framework. In Section \ref{section:bayesian}, we describe the Bayesian estimation method. Section \ref{section:theory} presents theoretical results. Sections \ref{section:rust} and \ref{section:gilleskie} contain the applications. 
Possible framework generalizations are discussed in Section \ref{sec:extensions}.
Derivations, proofs, and implementation details are given in appendices.

\section{General Model Setup}
\label{section:setup}

In the infinite-horizon version of the model, the decision maker maximizes the expected discounted sum of the per-period payoffs
\begin{align}
\label{eq:sequence_formulation}
\max_{d_t,d_{t+1},\ldots} E_t \left(\sum_{j=0}^\infty \beta^j  u\left( x_{t+j},d_{t+j},\epsilon_{t+j}  \right)  \right),
\end{align}
where $d_t \in \{0,1,\ldots,J \}$ is the control variable, $x_t \in X=\{1,\ldots,K\}$ is the state variable observed by the econometrician, $\epsilon_t = (\epsilon_{t0},\epsilon_{t1},\ldots,\epsilon_{tJ})^T \in R^{J+1}$ is the state variable unobserved by the econometrician, $\beta$ is the time discount factor, and $u(x_t,d_t,\epsilon_t)$ is the per-period payoff.  The decision-maker observes both $x_t$ and $\epsilon_t$ at time $t$ before making the decision.

Following \cite{Rust:87} and the subsequent literature, we assume that (i) the per-period payoffs are additively separable in $\epsilon_t$,
$u(x_t,d_t,\epsilon_t)= u(x_t,d_t)  + \epsilon_{td_t}$; (ii) $\epsilon_t$'s are independent of other variables and independently identically distributed (i.i.d.) over time according to a distribution $F$ with zero mean; 
(iii) the observed states evolve according to a controlled Markov chain 
$G$ with transition probabilities $G_{x}^j =\{Pr(x_{t+1}|x_t=x,d_t=j), \, x_{t+1} \in X\}$ and 
an initial distribution $\{Pr(x_{1}), \, x_1 \in X\}$.
The utility functions are assumed to depend on a vector of unknown parameters, $\theta \in \mathbb{R}^{d_\theta}$, that are estimated.
Below, we often omit $\theta$ in $u(x_t,d_t; \theta)$ and related objects such as value functions for notation brevity.
As in \cite{Rust:87}, \citet{Gilleskie_Eca:98}, and most of the literature, the time discount factor is assumed to be known; it is usually calibrated to imply a reasonable value of a risk free annual interest rate. 
In most applications, the observed state transition probabilities are estimated in a first stage prior to estimation of the preference parameters $\theta$ since it can be done directly from the observed state transitions with a relatively high precision and without solving the dynamic program.  In line with that and for notation brevity we treat $G$ as known and fixed.
Extensions of our results and methodology to continuous $X$ and unknown $G$ are discussed in Section \ref{sec:extensions}.

Under mild regularity conditions (\cite{BhattacharyaMajumdar:89}), the decision problem in \eqref{eq:sequence_formulation} admits the following Bellman representation
\begin{align}
\label{eq:emax_def}
Q(x) = \int \max_{j=0,1,\ldots,J} \bigg[ u(x,j) + \beta G_x^j Q + \epsilon_j \bigg] dF(\epsilon),
\end{align}
where $Q$ is called the \textit{Emax} function and $G_x^j Q$ denotes $E(Q(x_{t+1})|x_t=x,d_t=j)$.
The conditional choice probability (CCP) can be expressed as 
\begin{align}
\label{eq:ccp_def}
p(d | x) = \int 1\bigg\{ u(x,d) + \beta G_x^d Q + \epsilon_d \geq  u(x,j) + \beta G_x^j Q + \epsilon_j, \forall j  \bigg\} dF(\epsilon).
\end{align}
For a panel of observations, $D^n=\{x_{it}, d_{it},\, i=1,\ldots,n,\, t=1,\ldots,T\}$, of $n$ decision makers over $T$ time periods, 
the partial likelihood function (with the fixed $G_x^j$ pre-estimated from the observed transitions as is commonly done in practice) can be expressed as 
\begin{align}
\label{eq:likelihood}
\log L(D^n) = \sum_{i,t} \log p(d_{it} | x_{it}).
\end{align}
\cite{Rust:87} proposed to solve the dynamic optimization problem by first
iterating on the Bellman equation \eqref{eq:emax_def} to get close to the fixed point $Q$ 
and then using a Newton  method that quickly converges to the fixed point from a close starting point.
With $Q$ at hand, one can compute the CCPs in \eqref{eq:ccp_def} and evaluate the likelihood function at a given $\left(u;\beta;G \right)$.
Alternatively, \cite{JuddSu:08} proposed to use constrained optimization to maximize the likelihood function subject to \eqref{eq:emax_def}.
In either scenario, assuming that $\epsilon_{tj}$'s are i.i.d. Gumbel (or extreme value type I) delivers analytical expressions for the integrals in \eqref{eq:emax_def} and \eqref{eq:ccp_def}, 
\begin{align}
\label{eq:dyn_logit}
	p(d| x) = \frac{e^{u(x,d) + \beta G_x^d Q}}{\sum_{j=0}^J e^{u(x,j) + \beta G_x^j Q}}, \; 
	Q(x)= \log \sum_{d=0}^J e^{u(x,d) + \beta G_x^d Q}.
\end{align}
In the resulting dynamic logit specification, the computational burden of the model solution and estimation
is considerably alleviated.  Hence, the dynamic logit is predominantly used in applications of the estimable dynamic discrete choice models.
At the same time, the econometrics literature suggests that the distributional assumptions could be problematic in general, see, for example, \citet{Manski_book:99}.
In the following section, we specify a non-parametric model for the distribution of shocks for the general multinomial choice case that provides analytical simplifications comparable to those of the dynamic logit.


\section{Semiparametric Model}
\label{section:semi}

Rather than making a particular parametric assumption, we model the distribution of unobserved states using a flexible mixture specification.
In order to reduce the number of parameters, we use an innocuous normalization $\epsilon_{t0} = 0$ (the agent's decisions and value functions do not change if $\epsilon_{t0}$ is subtracted from the per-period payoff $u(x_t,d_t,\epsilon_t)$ for all $d_t \in \{0,1,\ldots,J\}$). 

For $\mu \in \mathbb{R}^J$ and $\sigma>0$, let us define a multivariate Gumbel density by 
\begin{equation}
\label{eq:phi_gumbel}
\phi(z; \mu,\sigma) = \prod_{j=1}^J \frac{1}{\sigma} \phi\left( \frac{z_j-\mu_j}{\sigma} \right),\, \mbox{where } \phi(z_j) = e^{-z_j-\gamma -e^{-z_j-\gamma}} 
\end{equation}
is the univariate Gumbel density
and $\gamma$ is the Euler-Mascheroni constant.  Some relevant properties of the Gumbel distribution are outlined in Appendix \ref{EV}.

For $\mu_k \in \mathbb{R}^{J}$, $\sigma_k \in \mathbb{R}_+$, $\omega_k \in [0,1]$, $k=1,\ldots, m$, and $\sum_{k=1}^m \omega_k =1$,
we model the unknown density by a location-scale mixture of Gumbel densities
\begin{align}
\label{eq:mixture}
\epsilon_t \sim  \sum_{k=1}^m \omega_k \phi(\cdot  ; \mu_k,\sigma_k ),
\end{align}
with a variable number of mixture components $m$ for which a prior distribution on the set of positive integers is specified.
Mixture models are extensively used in econometrics and statistics literature, see monographs by 
 \cite{McLachlanPeel:00} and \cite{FruhwirthSchnatter:06} for references.
It is well known that location-scale mixtures with a variable or infinite number of components 
can approximate any continuous or smooth density arbitrarily well.
For example, Bayesian models based on normal mixtures deliver optimal up to a log factor posterior contraction rates in adaptive estimation of smooth densities 
(\cite{Rousseau:10}, \cite{ShenTokdarGhosal2013}, and \cite{NoretsPelenis_Eca_2022}).
To develop intuition for this type of results note that the standard nonparametric density estimator based on kernel $\phi$ is a special case of \eqref{eq:mixture},
or, alternatively and more in line with the actual proofs,
the expectation of the standard kernel density estimator is a continuous mixture that can be discretized into a special case of \eqref{eq:mixture}.
Thus, it is reasonable to expect that the specification \eqref{eq:mixture} is very flexible. Indeed, in Section  \ref{section:theory},
we show that it can approximate smooth multivariate densities arbitrarily well in an appropriate distance so that the 
conditional choice probabilities and the \textit{Emax} function implied by the model with \eqref{eq:mixture} approximate those from the model with an arbitrary smooth density for $\epsilon_t$.

The model specification with \eqref{eq:mixture} also possesses attractive analytical properties.  
If a normalization $\epsilon_{t0} = 0$ is not imposed and $(J+1)$-dimensional version of \eqref{eq:mixture} is used, then
$Q$ and $p$ could be expressed as mixtures of the appropriately recentered and rescaled expressions from the dynamic logit model \eqref{eq:dyn_logit}.
However, even if the normalization $\epsilon_{t0} = 0$ is imposed, 
which is preferred as it reduces the dimension of the distribution we model nonparametrically,
closed form expressions for $Q$ and $p$ are still available.  They are presented in the following lemma.

%


\begin{lemma}\label{lm:ccp_emax}
Suppose $\epsilon \sim  \sum_{k=1}^m \omega_k \phi(\cdot  ; \mu_k,\sigma_k )$. Then, 
\begin{equation}\label{eq:ccp}
p(d|x)=
	\begin{cases}
      		\sum_{k=1}^m \omega_k \exp\left[-e^{-a_{kx}} \right], & \text{if }\ d=0 \\
      		\sum_{k=1}^m \omega_k \exp\left[\frac{u(x,d) + \beta G_x^d Q +\mu_{dk}}{\sigma_k} - A_{kx} \right] \left\{1- \exp\left[-e^{-a_{kx}} \right]  \right\}   
      		, & \text{if }\ d=1,\ldots,J;
    	\end{cases}
\end{equation}
\begin{equation}\label{eq:emax}
Q(x) = \sum_{k=1}^m \omega_k \sigma_k \left[A_{kx} + E1( e^{-a_{kx}} ) \right], \mbox{ where } E1(z) = \int_{z}^{\infty}  e^{-t} / t dt,
\end{equation}
\begin{align*}
A_{kx} &= \log \sum_{j=1}^J \exp \left( \frac{u(x,j) + \beta G_x^j Q +\mu_{jk}}{\sigma_k} \right) \mbox{ and }
a_{kx} =  \frac{u(x,0) + \beta G_x^0 Q}{\sigma_k} + \gamma - A_{kx}.
\end{align*}
\end{lemma}

The derivations of \eqref{eq:ccp} and \eqref{eq:emax} can be found in Appendix \ref{sec:ccp_emax}.
The derivatives of \eqref{eq:ccp} and \eqref{eq:emax} that are useful for the model solution and estimation are given in 
Appendix \ref{sec:derivatives}.
Similarly to \cite{Rust:87}, we obtain the solution of the Bellman equation \eqref{eq:emax} by a Newton-Kantorovich method 
described in Appendix \ref{sec:NK}.

\section{Inference}
\label{section:bayesian}

\subsection{Motivation and Overview of the Bayesian Approach}

In estimation of models based on location-scale mixtures with a variable number of components, the econometrician  faces several problems.
First, the scale parameters need to be bounded away from zero;  otherwise, the likelihood function is unbounded.
Second, the likelihood function is a rather complex function of parameters with multiple modes. 
Third, the number of mixture components needs to be selected in the estimation procedure.
Finally, there is usually considerable uncertainty about the estimated parameter values and it should be taken into account in 
model predictions and counterfactual analysis.

The Bayesian approach to inference and the associated simulation methods are well suited for solving these problems.
Prior distributions can provide soft constraints for the scale parameters and an appropriate penalization for the number of mixture components or model complexity.
MCMC methods can successfully explore very complex posterior or likelihood surfaces.
Posterior predictive distributions for objects of interest automatically incorporate the uncertainty about parameter values including the number of mixture components.

\subsection{Normalizations}
\label{subsection:normalizations}

In addition to the normalization $\epsilon_{t0} = 0$, the scale of $\epsilon_t$ can be innocuously normalized.  
Instead, to simplify the MCMC algorithm we impose a location and scale normalization on the parameters of the per-period payoffs and keep the location and scale of 
\eqref{eq:mixture} unrestricted.  Specifically, consider a linear in parameters utility specification
\[
u(x_t,j,\epsilon_t;\theta)=\theta_j + z_j(x_t)^\prime \theta_{J+1:d_\theta} + \epsilon_{tj},
\]
where $z_j(x_t)$ are known functions of the observed state variables.
In this specification, the intercepts $\theta_j$, $j=0,\ldots,J$, can be fixed to arbitrary values as long as 
the locations of $\epsilon_{tj}$, $j=1,\ldots,J$, are unrestricted.  

In applications, we set $\theta_0=0$ and 
$\theta_j$, $j=1,\ldots,J$ to $\hat{\theta}_j^{dl}$, the estimates obtained under the dynamic logit specification.
To normalize the scale of $\epsilon_t$, we assume that the sign of one of the coefficients, 
say $\theta_{J+1}$, is known and we keep this coefficient fixed (to the corresponding dynamic logit estimate, $\hat{\theta}_{J+1}^{dl}$).
Thus, the MCMC algorithm produces draws of $\theta_{J+2:d_\theta}$ and the mixture parameters $(\psi_{1m},m)$ in \eqref{eq:mixture}.
For comparisons of the estimation results with the dynamic logit estimates and the identified sets in 
\cite{NoretsTang2013},
the parameter draws are renormalized for reporting as follows 
\begin{equation}
\label{eq:renorm_params}
s \cdot \bigg ( \hat{\theta}_{1:J}^{dl} + \sum_{k=1}^m \omega_k \mu_k, \hat{\theta}_{J+1}^{dl}, \theta_{J+2:d_\theta}
\bigg),
\end{equation}
where the addition of $\sum_{k=1}^m \omega_k \mu_k$ to the intercepts corresponds to the 
zero mean for shocks and the scale factor is defined by the mixture parameters 
\[
s = \log 2\big/ E[\tilde{\epsilon}_{t1} 1(\tilde{\epsilon}_{t1}\geq M_{\tilde{\epsilon}_{t1}})],
\]
where
$\tilde{\epsilon}_{t1} = \epsilon_{t1} - \sum_{k=1}^m \omega_k \mu_{1k}$,
$M_{\tilde{\epsilon}_{t1}}$ denotes the median of $\tilde{\epsilon}_{t1}$,
and $\epsilon_{t1} \sim \sum_{k=1}^m \omega_k \phi(\cdot  ; \mu_{1k},\sigma_k )$.
There are many possible scale normalizations. The particular scale normalization we use here reduces to the one introduced by \cite{NoretsTang2013} for the binary choice case.
Let us emphasize that the normalizations discussed above are innocuous for estimation and counterfactual analysis as long as the assumed sign of $\theta_{J+1}$ is correct.


\subsection{Priors}

Let us introduce the prior distributions for the parameters of the mixture in \eqref{eq:mixture}.
We use the following prior distributions on the number of mixture components and the mixing weights, 
\begin{align}
\label{eq:m_prior}
\Pi(m) & \propto e^{-\underbar{A}_m m (\log m)^\tau}, \\
\label{eq:omega_prior}
\Pi(\omega_1,\ldots,\omega_m|m) & = \mbox{Dirichlet}(\underbar{a}/m,\ldots,\underbar{a}/m),
\end{align}
where the hyperparameters $\underbar{a}$, $\underbar{A}_m$, and $\tau$ are specified in the applications below.
For the theoretical results obtained in the present paper, we only need $\Pi(m)>0, \forall m$ and full support 
on the simplex for $\Pi(\omega_1,\ldots,\omega_m|m)$.  Nevertheless, the functional forms in \eqref{eq:m_prior} and \eqref{eq:omega_prior} perform well in applications and deliver optimal posterior contraction rates in nonparametric multivariate density estimation by mixtures of normal distributions, see, for example, \cite{ShenTokdarGhosal2013} and \cite{NoretsPelenis_Eca_2022}.  
We allow the scale parameter $\sigma_k$ to have a multiplicative part $\sigma$ that is common across the mixture components: $\sigma_k=\tilde{\sigma}_k \cdot \sigma$. This multiplicative specification performs well in a variety of applications of location-scale mixture models (see, for example, \cite{Geweke:05})
and is also important for the aforementioned optimal posterior concentration results for mixtures of normals.
In the applications, we use finite mixtures of normals as flexible priors for $\log \sigma$, $\log \tilde{\sigma}_k$ and the location parameters $\mu_{kj}$.

\subsection{MCMC Algorithm}
\label{sec:mcmc_details}

Our MCMC algorithm for simulating from the model posterior distribution combines Hamiltonian Monte Carlo (HMC) for simulating parameters conditional on the number of mixture components and an approximately optimal reversible jump algorithm from \citet{Norets2017mcmc} for simulating the number of mixture components.
HMC is a very popular and efficient MCMC algorithm; see, for example, \citet{neal2012hmc} for an introduction.
HMC requires only evaluation of the likelihood and the prior and their derivatives.
The proposals in HMC are obtained following the Hamiltonian dynamics on the parameter space that describe the movement of a puck on a friction-less surface with some 
initial random momentum. 
For implementing the HMC step of the algorithm we utilize the HMC sampler from the Matlab Statistics and Machine Learning toolbox.
The package can choose HMC's parameters, such as a step size, automatically, and we perform this automatic initialization 
once for each value of $m$ that we encounter in the MCMC run.
The package works only with unbounded parameters.  Hence, we transform the bounded parameters, such as mixing weights and scales, for the HMC step.
The form of the prior for the transformed parameters and the derivatives of the likelihood used in the algorithm are reported in Appendix \ref{sec:auxiliary_res_det}.




For the reversible jump algorithm, we need to transform the mixing weights into unnormalized weights $\gamma_k$, $k=1,\ldots,m$,
so that they have interpretation under different values of $m$.  Specifically,
conditional on $m$,
$\omega_k = \gamma_k/\sum_{l=1}^m \gamma_l$ and the Dirichlet prior on
$(\omega_1,\ldots,\omega_m)$ corresponds to a gamma prior for the unnormalized weights: 
$\gamma_k|m \sim Gamma(\underline{a}/m,1)$, $k=1,\ldots,m$.  
Let $\psi_k = (\mu_k,\tilde{\sigma}_k,\gamma_k)$ and $\psi_{1m}=(\theta, \sigma,\psi_1,\ldots,\psi_m)$,
where $\theta$ includes model parameters such as coefficients in the utility functions.
With this notation, the likelihood function is denoted by $p(D^n|m,\psi_{1m})$.

The following short description of the reversible jump algorithm is adapted from \cite{NoretsPelenis_Eca_2022}, see 
\cite{Norets2017mcmc} for more details.
 Denote a proposal distribution for the parameter of a new mixture component $m+1$ by $\tilde{\pi}_{m+1}(\psi_{m+1}|D^n,\psi_{1m})$.
The algorithm works as follows.  Simulate proposal $m^\ast$ from
$Pr(m^\ast=m+1|m)=Pr(m^\ast=m-1|m)=1/2$.  If $m^\ast=m+1$, then also simulate $\psi_{m+1} \sim \tilde{\pi}_{m+1}(\psi_{m+1}|D^n,\psi_{1m})$.
Accept the proposal with probability $\min\{1, \alpha(m^\ast,m)\}$, where
\begin{align}
\alpha(m^\ast,m) &= \frac{p(D^n|m^\ast,\psi_{1m^\ast}) \Pi(\psi_{1m^\ast}|m^\ast)\Pi(m^\ast)}
{p(D^n|m,\psi_{1m}) \Pi(\psi_{1m}|m)\Pi(m) } \nonumber
\\
& 
\label{eq:alpha_ar}
\cdot
\left(
\frac{1\{m^\ast=m+1\}}{\tilde{\pi}_m(\psi_{m+1}|\psi_{1m},Y)} + 1\{m^\ast=m-1\}\tilde{\pi}_{m-1}(\psi_{m}|\psi_{1m-1},D^n)
\right).
\end{align}
Innocuous random relabeling of mixture components increases the acceptance probability for attempts to delete $m$-th mixture component ($m^\ast=m-1$).  
\cite{Norets2017mcmc} shows that an optimal choice of the proposal distribution $\tilde{\pi}_m$ is 
the conditional posterior $p(\psi_{m+1}|D^n,m+1,\psi_{1m})$.  
The conditional posterior can be evaluated up to a normalization constant;  however, it 
seems hard to directly simulate from it and compute the required normalization constant.
 Hence, we use a Gaussian approximation to $p(\psi_{m+1}|D^n,m+1,\psi_{1m})$ as the proposal (with the mean equal to the conditional posterior mode, obtained by a Newton method, and the variance equal to the inverse of the negative of the Hessian evaluated at the mode).

The algorithm pseudo code is presented below.


\begin{algorithm}[h]
\setstretch{1.00}
\caption{Estimation Algorithm}
Step 1: Pre-estimate/fix the observed state transition probabilities $G$.\\
Step 2: Tune HMC hyperparameters for each number of mixture components $m=1,...,\bar{m}$, where $\bar{m}$ is a pre-specified positive integer.\\
Step 3: Initialize the parameters $m^{(0)}$ and $\psi_{1m^{(0)}}^{(0)}$.\\
Step 4: Run MCMC by iterating the following two steps.\\
\KwOutput{The posterior draws: $\left\{m^{(\ell)}, \psi_{1m^{(\ell)}}^{(\ell)} \right\}$, $\ell=1,...,L$.}
\For{$\ell \in \{0,1,...,L-1\}$}{

    1) Simulate $m^{(\ell+1)}|\ldots$ using optimal reversible jump.
    \begin{itemize}
    \item Propose $m^\ast$ with $Pr(m^\ast=m^{(\ell)}+1|m^{(\ell)})=Pr(m^\ast=m^{(\ell)}-1|m^{(\ell)})=1/2$. 
    \item If $m^\ast=m^{(\ell)}+1$, then also simulate $\psi_{m^{(\ell)}+1}^{(\ell)}$ from 
$\tilde{\pi}_{m^{(\ell)}+1}(\cdot|D^n,\psi_{1m^{(\ell)}}^{(\ell)})$. 

\item If $m^\ast=m^{(\ell)}-1$, choose $k$ randomly from $\{1,\ldots,m^{(\ell)}\}$ and exchange values of $\psi_k^{(\ell)}$ and $\psi_{m^{(\ell)}}^{(\ell)}$.

\item With probability  $\min\{1, \alpha(m^\ast,m^{(\ell)})\}$ accept $m^{(\ell+1)}= m^\ast$; otherwise, $m^{(\ell+1)}=m^{(\ell)}$.
\end{itemize}

    2) Simulate $\psi_{1m^{(\ell+1)}}^{(\ell+1)}|\ldots$ using HMC. 
    \begin{itemize}
    \item Initialize HMC algorithm by the current value $\psi_{1m^{(\ell+1)}}^{(\ell)}$
    and the hyperparameters specific to $m=m^{(\ell+1)}$.
    If HMC hyper parameters have not yet been tuned/obtained for $m=m^{(\ell+1)}$, then obtain and store them. 
    \item Perform one or several iterations of the HMC algorithm to get $\psi_{1m^{(\ell+1)}}^{(\ell+1)}$.
    \end{itemize}

}
Step 5: Use MCMC draws $\left\{m^{(\ell)}, \psi_{1m^{(\ell)}}^{(\ell)} \right\}$ to conduct inference on parameters or functions of interest.
\label{alg:mcmc}
\end{algorithm}
\FloatBarrier

The Matlab code for the MCMC algorithm and replication instructions for the estimation results in the applications in Sections \ref{section:rust} and \ref{section:gilleskie} are publicly available.\footnote{\url{https://anorets.github.io/papers/mix_ddcm_code.zip}}


\pagebreak

\section{Approximation Results and Asymptotics}
\label{section:theory}

In this section, we show that location-scale mixtures of Gumbel densities can arbitrarily well approximate densities from a large nonparametric class.
These approximation results combined with the \cite{Schwartz:65}'s theorem imply a posterior consistency result for the set identified model parameters.
We also show that a model with a finite mixture of Gumbels can exactly match the CCPs  
from a model with an arbitrary distribution of shocks.

\subsection{Approximation Results}
\label{sec:approx_results}

Let us first define a distance for distributions of utility shocks: 
for $F_i$ with density $f_i$, $i=1,2$, 
\[
\rho(F_1,F_2) = \int (1+ \sum_{j=0}^J |\epsilon_j| ) | f_1(\epsilon) -f_2(\epsilon) | d\epsilon.
\]
This distance is appropriate for our purposes as the \textit{Emax} function and the conditional choice probabilities are continuous in that distance as shown in the following lemma.

\begin{lemma}
\label{lm:LipschitzCont_Q_p}
Suppose (i)  $|u(x,j)|\leq \bar{u}<\infty$ for all $x \in X$ and $j=0,1,\ldots,J$;
(ii) under $F$, the density for $\epsilon_j -\epsilon_d$ is bounded for all $j \ne d$;
(iii) under $F$, $E(|\epsilon_j|)$ is finite for all $j$.
Then, the \textit{Emax} function and the conditional choice probabilities are locally Lipschitz continuous in $F$,
\[
\sup_x |Q(x;F)-Q(x;\tilde{F})| \leq C \cdot \rho(F,\tilde{F}),
\]
\[
\sup_{d,x} |p(d|x;F)-p(d|x;\tilde{F})| \leq C^\prime \cdot \rho(F,\tilde{F}),
\]
where constants $C$ and $C^\prime$ depend on $\beta$, $\bar{u}$ and the bounds on the densities and moments in conditions (ii)-(iii).
\end{lemma}

The lemma holds irrespective of whether the innocuous normalization $\epsilon_{t0}=0$ is imposed.
Its proof is given in Appendix \ref{sec:proofs}.

The following lemma shows that densities satisfying smoothness and finite moment conditions can be approximated by mixtures of Gumbels in distance $\rho$.

\begin{lemma}
\label{lm:approx_by_mgumbles}
Let $f$ be a density on $\mathbb{R}^I$ satisfying a moment existence condition 
\[\int ||\mu||_2 f(\mu) d\mu < \infty,
\]
and a smoothness condition
\begin{equation}
\label{eq:f_smooth}
 |f(z+h) - f(z) | \leq ||h||_2 L_f(z) e^{\tau ||h||_2},
\end{equation}
 for some $\tau>0$ and an envelope function $L_f(\cdot)$ such that
\begin{equation}
\label{eq:L_inf_smooth}
 \int (1+ |z_i| ) L_f(z) dz < \infty , i=1,\ldots,I.
\end{equation}
Then, for any $\delta>0$, there exist $(m,\omega,\mu,\sigma)$ 
where 
$m \in \mathbb{Z}^+$,
$\omega_j \in [0,1]$ with $\sum_{j=1}^m \omega_j=1$,
$\mu_j \in \mathbb{R}^I$, $j=1,\ldots,m$, and 
$\sigma>0$
such that 
\[
\rho\left(f(  \cdot  ), \sum_{j=1}^m \omega_j \phi(\ \cdot \  ; \mu_j, \sigma)  \right) < \delta.
\]

\end{lemma}

We conjecture that the smoothness and tail conditions on $f$ in the lemma can be weakened at the expense of the proof simplicity.
The lemma is proved in Appendix \ref{sec:proofs}.  The proof uses only smoothness and tail conditions on $\phi$ that are shown to hold for Gumbel densities
in Lemmas \ref{lm:lemma_bounded_integral} and \ref{lm:Lipschitz_Gumbel_location} in Appendix \ref{sec:proofs}.
Thus, Lemma \ref{lm:approx_by_mgumbles} holds for more general location-scale mixtures.  
These generalizations do not seem essential and we do not elaborate on them here for brevity.

The final intermediate result that we need for establishing posterior consistency 
is the continuity of finite Gumbel mixtures in parameters in distance $\rho$, which we present in the following lemma.

\begin{lemma}
\label{lm:Fcont_params}
Let $F^1$ and $F^2$ denote two mixtures of Gumbel densities on $\mathbb{R}^{J}$ with densities
$
f^i(\epsilon) = \sum_{k=1}^m \omega_k^i \phi\left( \epsilon; \mu_k^i, \sigma^i  \right)
$. Then, for a given $\delta>0$ and $F^1$, there exists $\tilde{\delta}>0$ such that
for any $F^2$ with parameters satisfying:
$|\sigma^1-\sigma^2|< \tilde{\delta}$, $|\omega_k^1-\omega_k^2|< \tilde{\delta}$, and 
$|\mu_k^1-\mu_k^2|< \tilde{\delta}$, $k=1,\ldots,m$, we have 
$\rho (  F^1, F^2 ) < \delta$.
\end{lemma}

\subsection{Posterior Consistency}
\label{sec:post_cons}

Let us denote the short panel dataset by $D^n = \{d_{it},x_{it}, t=1,\ldots,T, \, i=1,\ldots,n\}$;
the observations are assumed to be independently identically distributed over $i$, with a small $T$ and a large $n$.
The utility function is parameterized by a vector $\theta \in \mathbb{R}^{d_\theta}$, $u(x,d; \theta)$.  
Let $P(\theta,F)=\{p(d|x;\theta,F), \, x \in X,\, d=0,\dots,J\}$ denote
the collection of the CCPs for the distribution of shocks $F$ and parameters $\theta$.

\begin{theorem}
\label{th:post_cons}
Let $(\theta_0,F_0)$ be the data generating values of parameters.
Suppose 
(i) The observed state space is finite, $X=\{1,\ldots,K\}$;
(ii) $G^d$, $d=0,\ldots,J$ and the distribution of the initial observed state $x_{i1}$ are known and fixed;
(iii) $\forall x \in X$, $\exists t \in \{1,\ldots,T\}$, such that $Pr(x_{it}=x)>0$;
(iv) $u(x,d; \theta)$ is continuous in $\theta$;
(v) $F_0$ satisfies the conditions of Lemma \ref{lm:approx_by_mgumbles};
(vi) For any $\delta>0$, $\Pi(B_\delta(\theta_0))>0$, where $B_\delta(\theta_0)$ is a ball with radius $\delta$ and center $\theta_0$;
(vii) For any $\delta>0$, positive integer $m$, $\mu_k \in \mathbb{R}^J$, $\sigma_k>0$, $w_k \geq 0$, $k=1,\ldots,m$, $\sum_{k=1}^m w_k=1$,
$\Pi(B_\delta(\mu_1,\sigma_1, \ldots, \mu_m,\sigma_m, w_1, \ldots,w_{k-1})|m)>0$.
Then, for any $\delta>0$,
\[
\Pi \big( \theta,F: ||P(\theta_0,F_0)-P(\theta,F) ||> \delta \big  |D^n \big) \rightarrow 0 \mbox{ almost surely.}
\]

\end{theorem}

The theorem shows that the posterior concentrates on the set of parameters and distributions of shocks $(\theta,F)$ such that their implied 
CCPs $P(\theta,F)$ are arbitrarily close to the data generating CCPs $P(\theta_0,F_0)$.
To prove this result we use \cite{Schwartz:65} posterior consistency theorem: if the prior puts positive mass on any Kullback-Leibler neighborhood 
of the data generating distribution then the posterior puts probability converging to 1 on any weak neighborhood of the data generating distribution.
Since $X$ is finite, the convergence in weak topology and Kullback-Leibler divergence for distributions on $\{d_{it},x_{it}, t=1,\ldots,T\}$ 
are equivalent to convergence for vectors
$\{p(d|x), \, x \in X,\, d=0,\dots,J\}$ in a euclidean metric when $G^d$ and the distribution of the initial $x_{i1}$ are fixed
and satisfy our theorem condition (iii).
Thus, to obtain the conclusion of the theorem we only need to establish that 
the prior puts positive probability on any euclidean neighborhood of $P(\theta_0,F_0)$.
First, note that when $u(x,d; \theta)$ is continuous in $\theta$, $P(\theta,F)$ is also continuous in $\theta$ in our settings, see, for example, \cite{Norets_ddcm_diff_cont:09}; and, thus, Lipschitz continuity of $P(\theta,F)$ in $F$ from Lemma \ref{lm:LipschitzCont_Q_p} delivers continuity of $P(\theta,F)$ in $(\theta,F)$.
The finite mixture approximation result in Lemma \ref{lm:approx_by_mgumbles}, 
the continuity of $P(\theta,F)$ in $(\theta,F)$, the continuity of 
finite mixtures in parameters in Lemma \ref{lm:Fcont_params}, 
and the theorem conditions (vi) and (vii) on the priors, imply a positive prior probability for any 
neighborhood of $P(\theta_0,F_0)$, and thus, the theorem conclusion.
Possible extensions of the theorem (and the lemmas above) to continuous $X$ and unknown $G^d$ are discussed in Section \ref{sec:extensions}.


Theorem \ref{th:post_cons} characterizes the support of the posterior in the limit but not its shape, which can also be of interest.
Note that the data depend on $(\theta,F)$ only through CCPs $P(\theta,F)$ and the posterior for CCPs concentrates at $P(\theta_0,F_0)$. 
Therefore, the posterior for $(\theta,F)$ converges to
the conditional prior $\Pi(\theta,F|P)$ at $P=P(\theta_0,F_0)$ under continuity conditions on $\Pi(\theta,F|P)$, see, for example, 
\cite{PlagborgMoller2019}.
As the distribution of shocks is an infinite dimensional object and the solution to the dynamic program does not have a simple explicit
form,
it appears difficult to characterize the conditional prior $\Pi(\theta,F|P)$, which is implied by the map $P(\theta,F)$ and the prior on $(\theta,F)$. 
Nevertheless, we can deduce from our approximation and continuity results that under the conditions of Theorem \ref{th:post_cons},
for $\delta > 0$ there exists $\tilde{\delta}>0$
such that 
$\theta \in B_{\tilde{\delta}}(\theta_0)$ and $F \in B_{\tilde{\delta}}(F_0)$ imply 
$P(\theta,F) \in B_{\delta}(P(\theta_0,F_0))$ and
\begin{align*}
&\Pi\bigg(\theta \in B_{\tilde{\delta}}(\theta_0), F \in B_{\tilde{\delta}}(F_0) \bigg| P \in B_{\delta}(P(\theta_0,F_0))\bigg) = 
\frac{\Pi(\theta \in B_{\tilde{\delta}}(\theta_0), F \in B_{\tilde{\delta}}(F_0))}{\Pi\big (P \in B_{\delta}(P(\theta_0,F_0))\big)}
\\
&\geq
\Pi\big(\theta \in B_{\tilde{\delta}}(\theta_0), F \in B_{\tilde{\delta}}(F_0)\big)>0,
\end{align*}
which suggests that the conditional prior would not rule out the data generating parameter values.

\subsection{Exact Matching of CCPs}
\label{sec:exact_match_ccps}

In this subsection, we show that for a finite observed state space, our model formulation based on finite mixtures can exactly match the CCPs  
from a model with an arbitrary distribution of shocks.

\begin{lemma}
\label{lm:exact_match_ccps}

Suppose 
(i) The observed state space is finite, $X=\{1,\ldots,K\}$;
(ii) $(\theta_0,F_0)$ are the data generating values of parameters; 
(iii) $F_0$ has finite first moments and a density that is positive on $\mathbb{R}^J$.
Then there exists a finite mixture of Gumbels $F$ such that 
$P(\theta_0,F_0)=P(\theta_0,F)$. An upper bound on the number of mixture components in $F$ depends only on $K$ and $J$.

\end{lemma}

The result in Lemma \ref{lm:exact_match_ccps} holds not only for mixtures of Gumbels but more generally for location-scale mixtures of 
distributions with finite first moments, which is evident from the proof presented in Appendix \ref{sec:proofs}. 
Lemma \ref{lm:exact_match_ccps} can be used to relax the smoothness assumptions on $F_0$ in 
the posterior consistency results of Theorem \ref{th:post_cons}. Specifically, 
conditions (v) in Theorem \ref{th:post_cons} can be replaced by conditions (iii) in Lemma \ref{lm:exact_match_ccps}; in the proof,
the approximation results in Lemma \ref{lm:approx_by_mgumbles} can be replaced 
by the exact CCPs matching results in Lemma \ref{lm:exact_match_ccps}.
Nevertheless, the approximation results in Lemma \ref{lm:approx_by_mgumbles} have independent value. 
First, they hold for infinite and continuous $X$. 
Furthermore, they imply that the prior on the distribution of shocks is flexible
in a sense that it puts positive probability on any metric $\rho$ neighborhood in a large nonparametric class of distributions,
which suggests that the conditional prior for the distribution of shocks and parameters given CCPs, $\Pi(\theta, F | P)$, is also flexible as discussed at the end of Section \ref{sec:post_cons}. 
Finally, while asymptotically the number of mixture components is bounded for the exact matching in Lemma \ref{lm:exact_match_ccps} and has to increase to infinity for the approximation results in Lemma \ref{lm:approx_by_mgumbles}, in practice,  a small number of mixture components delivers sufficiently good approximations and the exact matching requires a very large number of mixture components ($m=182$ for Rust's bus engine replacement model).

\subsection{Inference for Identified Sets}
\label{sec:inference_IS}

\citet{NoretsTang2013} and generalizations of their results to the multinomial case in \citet{Norets_ddc_mult:11} show that
the utility parameters $\theta$ and the distribution of shocks $F$, and, thus, functions of $(\theta,F)$ such as results of counterfactual experiments, are set identified in the present settings.
\cite{MoonSchorfheide:12} show that in contrast to the point identified regular settings, the Bayesian credible sets for set identified parameters do not have frequentist coverage properties and are too small from the classical perspective.  \citet{NoretsTang2013} point out in their Section 3.2 that Bayesian and classical inference results can be reconciled if inference is performed on the identified sets.
\cite{KlineTamer2016} further study this approach and \cite{Kitagawa:11} obtain related results under multiple priors for set identified parameters.
In this subsection, we describe how credible and confidence sets for identified sets can be defined and computed from the output of our MCMC algorithm.

Suppose the data generating values of parameters are $(\theta_0,F_0)$ and we are interested in $\eta_0=g(\theta_0,F_0)$.
The data generating values of CCPs, $P_0=P(\theta_0,F_0)$, can be consistently estimated from the observed data, and, thus, are considered known in the identification analysis.
The identified set for $\eta_0$ is defined by
\[
I_\eta(P_0) = \{\eta=g(\theta,F), \; \forall (\theta,F) \;s.t.\; P(\theta,F)=P_0\}.
\]
Following \citet{NoretsTang2013} and \cite{KlineTamer2016}, we can define a posterior distribution on the space of identified sets 
$I_\eta(P)$ using the marginal posterior distribution on the CCPs $P$.  Then, a $1-\alpha$-credible set for $I_\eta(P)$ can be defined from  a $1-\alpha$-credible set for CCPs, $B_{1-\alpha}^P$, by
\begin{equation}
\label{eq:B_I_eta}
B_{1-\alpha}^{I_\eta} = \bigcup_{P \in B_{1-\alpha}^P} I_\eta(P).
\end{equation}
When $B_{1-\alpha}^P$ is a $1-\alpha$ highest posterior density set and the Bernstein - von Mises theorem holds for the point identified CCPs $P$,  
$B_{1-\alpha}^P$ asymptotically has a $1-\alpha$ frequentist coverage probability for $P(\theta_0,F_0)$ and, thus, 
$B_{1-\alpha}^{I_\eta}$ has at least a $1-\alpha$ frequentist coverage probability for the identified set $I_\eta(P_0)$ and $\eta_0$.

The sets in \eqref{eq:B_I_eta} might be conservative; nevertheless, it would be prudent to report them in applications 
as in the limit they do not depend on the shape of the prior and possess both frequentist and Bayesian properties.

For a scalar $\eta_0$, an approximation to $B_{1-\alpha}^{I_\eta}$ can be computed from a sample of MCMC posterior draws
$\{\theta^{(l)},F^{(l)},P^{(l)}=P(\theta^{(l)},F^{(l)}), \; l=1,\ldots,L\}$ as follows.  
First, we obtain an approximation to $B_{1-\alpha}^P$ 
\[
\hat{B}_{1-\alpha}^P = \left\{P:\; (P-\bar{P})^\prime \hat{\Sigma}^{-1} (P-\bar{P}) \leq \chi^2_{1-\alpha}(JK)\right\},
\]
where $\bar{P}=\sum_{l=1}^L P^{(l)}/L$, $\hat{\Sigma} = \sum_{l=1}^L (P-\bar{P})(P-\bar{P})^\prime/L$, and $\chi^2_{1-\alpha}(JK)$ is the $1-\alpha$ quantile of the $\chi^2$ distribution with $JK$ degrees of freedom.
Then, $B_{1-\alpha}^{I_\eta}$ is approximated by
\[
\hat{B}_{1-\alpha}^{I_\eta}=\left[\min_{l:\: P^{(l)} \in \hat{B}_{1-\alpha}^P} \eta\left(\theta^{(l)},F^{(l)}\right), \max_{l:\: P^{(l)} \in \hat{B}_{1-\alpha}^P} \eta\left(\theta^{(l)},F^{(l)}\right)\right].
\]
We report these sets along with the standard HPD sets in the application in Section \ref{section:gilleskie}.


\section{Application I: Rust's Binary Choice Model}
\label{section:rust}

\citet{NoretsTang2013} propose a method for computing identified sets for parameters in dynamic binary choice models and apply their method to the \citet{Rust:87}'s model. In this section, we show that our semiparametric model can also recover the identified set for that model.

\subsection{\citet{Rust:87}'s Optimal Bus Engine Replacement Problem}

In each time period $t$, the agent decides whether to replace the bus engine ($d_t=1$) or not ($d_t=0$)  given the current mileage $x_t$ of the bus. Replacing an engine costs $\theta_0$. If $d_t=0$, then the agent conducts a regular maintenance which costs $-\theta_1x$. The utility function of the agent is
$u(x,0) = \theta_0 + \theta_1 x$ and $u(x,1) =\epsilon$.
\citet{Rust:87} assumes that $\epsilon$ follows the logistic distribution. 
The mileage  $x_t\in \{1, \ldots ,K=90 \}$ evolves over time following the transition probabilities: $Pr(x_{t+1}|x_t,d_t=0)=\pi_0$ for $x_{t+1}-x_t = 0$; 
$Pr(x_{t+1}|x_t,d_t=0)=\pi_1$ for $x_{t+1}-x_t = 1$; $Pr(x_{t+1}|x_t,d_t=0)=1-\pi_0-\pi_1$ for $x_{t+1}-x_t = 2$; and 
$Pr(x_{t+1}|x_t,d_t=0)=0$ otherwise.
When the engine is replaced ($d_t=1$), the mileage restarts at $x_t=1$.


As in \citet{NoretsTang2013}, we use the following data generating process: 
logistic distribution for $\epsilon$, 
$\theta_0 = 5.0727, \theta_1= -0.002293,\pi_0 = 0.3919, \pi_1 = 0.5953$ and the discount factor $\beta = 0.999$. 

At the data generating parameters, we solve the dynamic program to obtain the vector of CCPs,  $(p(d|1),\ldots,p(d|K))$ for $d=0,1$. 
Rather than using simulated observations in this exercise, we use the true CCPs as the sample frequencies and report the results for different sample sizes.
Specifically, for a given $N \in Z^+$, we set $n_{dx}$, the number of times $d$ was chosen at each 
state $x$ as follows, $n_{0x} = p(0|x) \times N$ and $n_{1x}=p(1|x) \times N$ for $x=1,\ldots,K$.  
In this way, we can check if our MCMC algorithm for the semiparametric model specification can recover 
the identified set computed by the algorithm from \cite{NoretsTang2013} for a given fixed vector of CCPs.
Estimation results for simulated data are presented for the multinational choice application in Section \ref{section:gilleskie}.

Since we do not impose a location and scale normalization on the mixture specification for the distribution of 
$\epsilon$ and the model has only two utility parameters that are defined by the location and the scale, 
the values of $(\theta_0,\theta_1)$ corresponding to the location and scale of the logistic distribution 
are computed from $(\sigma,\omega_k,\mu_k,\sigma_{k}, \, k=1,\ldots,m)$ by formula \eqref{eq:renorm_params}.
We use the following flexible prior distributions that are tuned to spread the prior probability over a large region for $(\theta_0,\theta_1)$ that includes the identified set.
\begin{align*}
& \pi(m=k) \propto e^{-\underbar{A}_m k(\log k)^\tau}, \, \underbar{A}_m=0.05, \, \tau=5, \\
& (\omega_1,\ldots,\omega_m) \sim \mbox{Dirichlet}(\underbar{a}/m,\ldots,\underbar{a}/m), \, \underbar{a}=10, \\
&\mu_{k} \sim  0.5N(2.5,1^2)+0.5N(-3,7^2),\\
& \log \sigma_{k} \sim  0.4N(0,1^2)+0.6N(-6,1^2), \; \log \sigma \sim N(0,0.01^2).
\end{align*}

Below, we report the draws of $(\theta_0,\theta_1)$ obtained from the draws of $(\sigma,\omega_k,\mu_k,\sigma_{k}, \, k=1,\ldots,m)$
for the prior and the posterior for $N \in \{ 3, 10 \}$ and compare them to the true identified set.
Panel (a) of Figure \ref{fig_Rust_id}  shows the prior draws of the utility parameters. 
First note that the prior draws are not  uniformly distributed on the utility parameter space. 
In practice, it is difficult to come up with a prior for the distribution parameters that implies a uniform prior in the $\theta$ space. Second, many prior draws are outside of the identified set. 
The other two panels in Figure \ref{fig_Rust_id} show posterior draws of utility parameters with different number of observations $N \in \{3, 10 \}$.
The posterior concentrates more on the identified set as $N$ increases.
\begin{figure}[ht] 
  \begin{subfigure}[b]{0.3\linewidth}
    \centering
    \includegraphics[width=\linewidth]{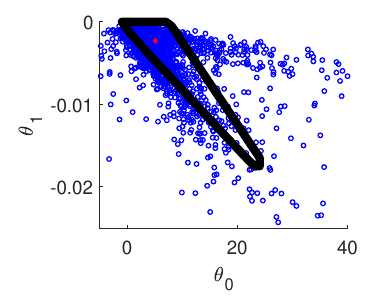} 
    \caption{prior} 
  \end{subfigure}
  \begin{subfigure}[b]{0.3\linewidth}
    \centering
    \includegraphics[width=\linewidth]{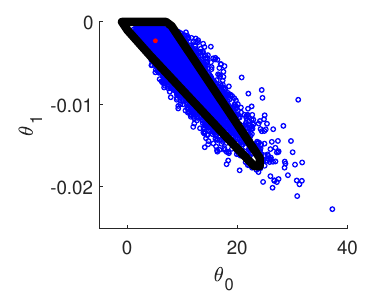} 
    \caption{$N=3$} 
  \end{subfigure}
  \begin{subfigure}[b]{0.3\linewidth}
    \centering
    \includegraphics[width=\linewidth]{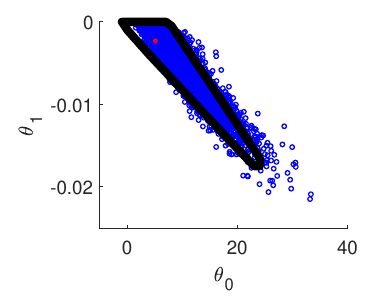} 
    \caption{$N=10$} 
  \end{subfigure}
\caption{\small{
Utility parameter draws. 
Prior (left). Posterior draws for $N=3$ (center) and $N=10$ (right). 
500,000 MCMC iterations.  The first 100,000 draws discarded as burn-in.
Every 10th draw is shown. 
The true identified set is shown by solid black lines.
The red dot corresponds to the point-identified utility parameter values under the logistic distribution for $\epsilon$.
}}
\label{fig_Rust_id}
\end{figure}
To assess the convergence of the MCMC algorithm, consider Figure \ref{fig_Rust_id3D} showing the draws of utility parameters for 500,000 MCMC iterations. As can be seen from the figure, the chain sweeps through the identified set repeatedly during the MCMC run.
Thus, we conclude that our approach can be used to recover the identified sets of utility parameters.
Figure \ref{fig_Rust_trace} in Appendix \ref{sec:rust_implement_details} shows additional evidence of MCMC convergence including trace plots of $m$, $\sum_{k=1}\omega_k\mu_k$, and other parameters.

\begin{figure}[ht] 
  \begin{subfigure}[b]{0.4\linewidth}
    \centering
    \includegraphics[width=\linewidth]{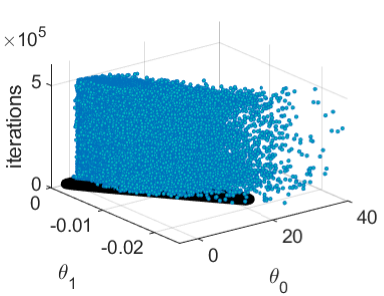} 
    \caption{$N=3$} 
  \end{subfigure}
  \begin{subfigure}[b]{0.4\linewidth}
    \centering
    \includegraphics[width=\linewidth]{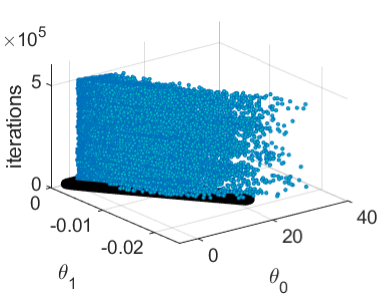} 
    \caption{$N=10$} 
  \end{subfigure}
\caption{\small{
Utility parameter draws. The z-axis represents the MCMC iterations. 
The true identified set is shown in black.
}}
\label{fig_Rust_id3D}
\end{figure}

\FloatBarrier



\section{Application II: Gilleskie's Multinomial Model of Medical Care Use and Work Absence}
\label{section:gilleskie}

In this section, we illustrate our methodology using a multinomial choice model of medical care use and work absence from \citet{Gilleskie_Eca:98}. 
An extension of \citet{NoretsTang2013} to the multinomial case seems computationally infeasible and, hence, in this application we provide comparisons of our method only with a dynamic logit specification.

\subsection{The model}\label{model}
In the model, individuals occupy one of $2$ distinct health states: well, $k=0$, or sick, $k=1$. 
An individual receives the utility associated with being well until contracting an illness of a specific type (we make the simplifying assumption that there is only one illness type, although \citet{Gilleskie_Eca:98} works with two illness types). An illness episode can last up to $T$ periods enumerated by $t=1,\ldots,T$; $t=0$ corresponds to the state of being well, $k=0$.

\subsubsection{Alternatives}
An individual who became sick makes decisions about doctor visits and and work absences. In each period $t$ of an illness, alternatives available to an employed individual who is sick are:
$d_t=0$ - work  and don't visit a doctor, 
$d_t=1$ - work and visit a doctor,  
$d_t=2$ - don't work and don't visit a doctor, and 
$d_t=3$ - don't work and visit a doctor. 
The utility of the agent depends on 
the elapsed length of the current illness $t$, the accumulated number of physician visits $v_t$, and the accumulated number of work absences $a_t$. 
The state variables observed by the econometrician and the agent at $t$ are $x_t=\left(t, v_t,a_t \right)$.
  Note that $k=1$ if and only if $t>0$, so $k$ does not appear in the definition of $x_t$.

\subsubsection{State variable transitions}

The state variables evolve in the following way. An individual always starts with the state of being well, $x_0=(0,0,0)$. The individual contracts an illness and moves to the state $x=(1,0,0)$ with probability 
$\pi^S(H)$,
where $H$ is a vector of exogenous indicators for health status  and  being between 45-64 years of age.  

The accumulated number of physician visits $v_t$ and the accumulated number of illness-related absence from work $a_t$ both take values in $ \{0,1,\ldots, T-1\} $. They start from $v_1=a_1=0$ and evolve in the following way:
$v_{t+1} = v_t +1(d_t=1 \text{ or } 3)$ and 
$a_{t+1} = a_t +1(d_t=2 \text{ or } 3)$.

In each illness period $t \in \{1,\ldots,T\}$, the individual recovers and returns to the state of being well with probability $\pi^W(x_t,d_t)$.  Gilleskie parameterizes and estimates 
$\pi^W(x_t,d_t)$ and $\pi^S(H)$ prior to estimating the preference parameters. We use those estimates in our application.

\subsubsection{Utility}
\label{sec:utility}
The per-period consumption is defined as $C(x_t,d_t) = Y -\big[    PC 1(d_t=1 \text{ or } 3)  +  Y\big(1-L\Phi(x_t,d_t) \big)1(d_t=2 \text{ or } 3)   \big]1(t>0)$,
where $Y$ is the per-period labor income, $PC$ is the cost of a doctor visit, and 
$\Phi(x_t,d_t)=\exp(\phi_1+\phi_2a'(x_t,d_t))/[1+\exp(\phi_1+\phi_2a'(x_t,d_t))]$ 
is the portion of income that the sick leave coverage replaces, where 
$a'(x_t,d_t)$ is the value of $a_{t+1}$ given $(x_t,d_t)$. 
$L\in (0,1)$ is the sick leave coverage rate.

The per-period utilities  can be expressed in the following form. 
\begin{alignat}{4}
u(d_t=1,x_t,\epsilon_t) 
&=  \theta_1  +&&\theta_4 1(t=0)  &&+\theta_6 C(x_t,1)1(t>0)   &&+ \epsilon_{t1}, \nonumber \\
u(d_t=2,x_t,\epsilon_t) 
&=  \theta_2   +&&\theta_4 1(t=0)  &&+\theta_6 C(x_t,2)1(t>0)   &&+\epsilon_{t2},\nonumber \\
u(d_t=3,x_t,\epsilon_t) 
&=  \theta_3   +&&\theta_4 1(t=0)    &&+\theta_6 C(x_t,3)1(t>0)   &&+\epsilon_{t3}, \nonumber \\
u(d_t=0,x_t,\epsilon_t) 
&=                &&\theta_5 1(t=0)                       &&+ \theta_6 C(x_t,0)1(t>0)   &&+\epsilon_{t0},          \nonumber          
\end{alignat}
where
$\theta_4=-\infty$ so that when $t=0$ the decision $d_t=0$ is always chosen.
Since we do not restrict the location and scale of $\epsilon_t$, the values 
$ \theta_p$, $p=1,2,3$ can be set to arbitrary values and $\theta_5$ can be set to an arbitrary positive value in our semiparametric estimation procedure. 
More details on per-period utilities are given in Appendix \ref{sec:utility_derivation}.

\subsection{Estimation}
\label{sec:gill_estimation}
For data generation we use parameter values based on estimates in  \citet{Gilleskie_Eca:98}
for Type 2 illness with some adjustments so that the expected number of doctor visits and work absences roughly match with Gilleskie's sample.
We let the data generating distribution of the utility shocks to be a two-component mixture of extreme value distributions: $
\sum_{k=1}^2
\omega_k
\phi \left( \epsilon; \mu_k,\sigma_k \right)
$. See Appendix \ref{sec:giil_data_gen} for the data-generating parameter values. 
The panel data $\{x_{it}, d_{it},\, i=1,\ldots,n,\, t=1,\ldots,T\}$ is sequentially simulated for $n=100$ individuals and $T=8$ periods. 

The priors are specified as follows, 
$\underbar{a}=10$, $A_m=0.05$, and $\tau=5$,
$\mu_{jk} \sim N(0, 2^2)$, 
$\log \sigma_{k}  \sim N(0, 1)$, 
$\log \sigma \sim N(0,0.01^2)$, and 
$\theta_6 \sim N(0,4^2)$.
This gives normal prior on $\mu_{jk}$'s and $\theta_6$ with large variances. 
The log-normal prior on the component specific scale parameters also implies sufficiently large prior probabilities for large values of $\sigma_k$'s.
Prior sensitivity checks presented in Appendix \ref{sensitivity} show that the obtained estimation results are not substantively affected by moderate changes in the prior.
Appendix \ref{gilleskie_logit} shows results when extreme value distribution of shocks is used for the data generation.

\subsubsection{Estimation Results and Counterfactuals}\label{results}
We use 20,000 MCMC iterations to explore the posterior distribution. 
Figure \ref{fig_Gilleskie_m} shows a trace plot and a p.m.f. of the number of mixture components $m$. 
The posterior has its peak at $m=2$ and the posterior probability of $m=1$ is small relative to its prior. 

\FloatBarrier

\begin{figure}[ht] 
  \begin{subfigure}[b]{0.4\linewidth}
    \centering
    \includegraphics[width=0.8\linewidth]{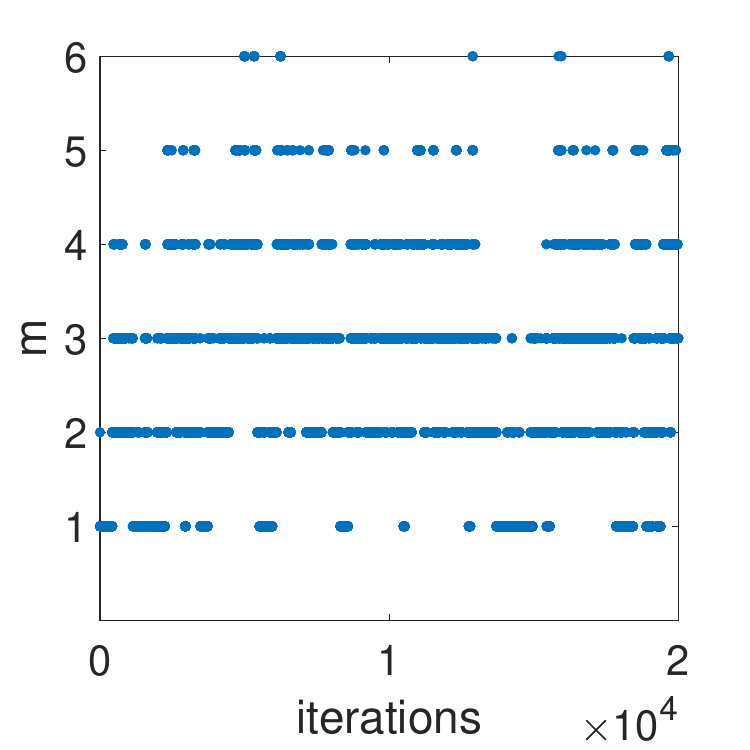} 
\caption{\small{Trace plot of $m$ }}
  \end{subfigure}
  \begin{subfigure}[b]{0.4\linewidth}
    \centering
    \includegraphics[width=0.8\linewidth]{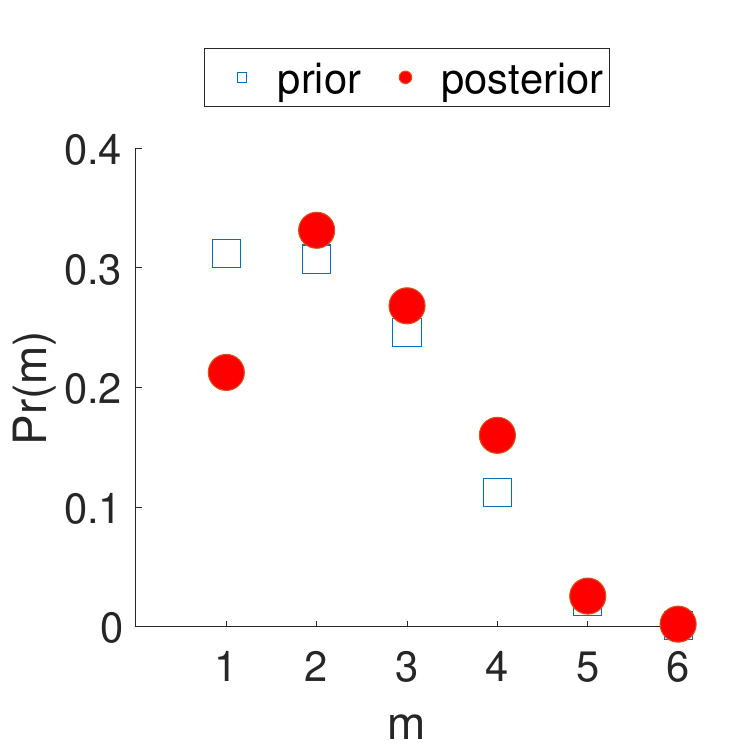} 
\caption{\small{p.m.f. of $m$ }}
  \end{subfigure}
\caption{\small{Trace plot and p.m.f. of $m$ }}
\label{fig_Gilleskie_m}
\end{figure}
\FloatBarrier

Figure \ref{fig_Gilleskie_posterior_theta} shows the posterior densities of the utility function parameters in the 
location and scale normalization corresponding to the original model in \cite{Gilleskie_Eca:98} described 
in Section \ref{subsection:normalizations}.  The corresponding trace plots presented in 
Figure \ref{fig_Gilleskie_trace} in Appendix \ref{appsec:traceplots} provide evidence that the MCMC algorithm converged.

\begin{figure}[ht] 
  \begin{subfigure}[b]{0.7\linewidth}
    \centering
    \includegraphics[width=\linewidth]{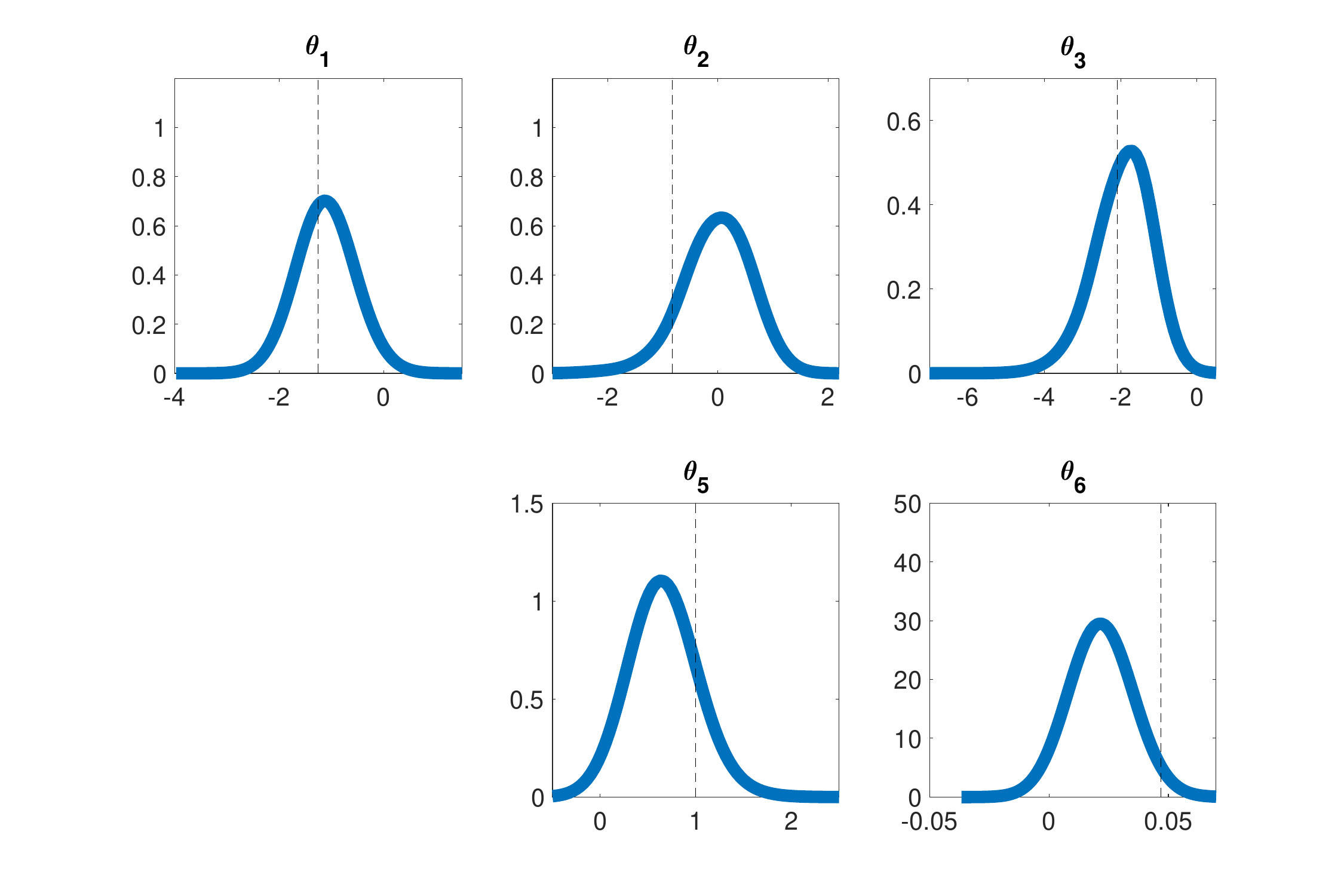} 
  \end{subfigure}
\caption{\small{Posterior densities of utility parameters (solid) and the data generating values (dashed). }}
\label{fig_Gilleskie_posterior_theta}
\end{figure}


The standard parametric approach that Gilleskie takes is to assume that the shocks are  extreme value i.i.d. and to estimate the model by the maximum likelihood method. 
In Figure \ref{fig_Gilleskie_EvEa} and Table \ref{table_estimation}, we compare estimation results from our semiparametric method to those from the MLE.  
For the comparison, we use the expected number of doctor visits $E(v)$  implied by the model.  It is a function of the model parameters and is of interest in the application.
One of the main advantages of structural estimation is that it provides an attractive framework for counterfactual experiments. 
We consider a counterfactual experiment presented in Section 6.2 of Gilleskie's paper (Experiment 1). 
In this experiment, we are interested in the behavior of individuals when the coinsurance rate paid out of pocket is set to zero. 
Thus, we examine the counterfactual model solution when $PC=0$ and the transition probabilities and $(\theta,F)$ are unchanged. 
The last row of Table \ref{table_estimation} and panel (b) in Figure \ref{fig_Gilleskie_EvEa} display the estimation results for $E(v)$ in the counterfactual environment.

In addition to the posteriors, the true values, the HPD credible intervals, and the MLE confidence intervals computed by the Delta method, Figure \ref{fig_Gilleskie_EvEa} and Table \ref{table_estimation} 
present credible intervals for the identified sets of $E(v)$, $\hat{B}_{0.95}^{I_{E(v)}}$, that are introduced in Section \ref{sec:inference_IS}. 
While the point estimates of the increase in $E(v)$ are similar for both approaches,
the 95\% Bayesian credible interval for the identified set of the counterfactual $E(v)$ in the semiparametric model is up to 6 times wider than the 95\% confidence interval for the MLE in the parametric setting. 
Importantly, the confidence interval by far misses the true counterfactual value of $E(v)$, while the credible interval for the identified set includes the true value and the HPD credible interval gets very close to it.


\begin{figure}[ht] 
  \begin{subfigure}[b]{0.5\linewidth}
    \centering
    \includegraphics[width=\linewidth]{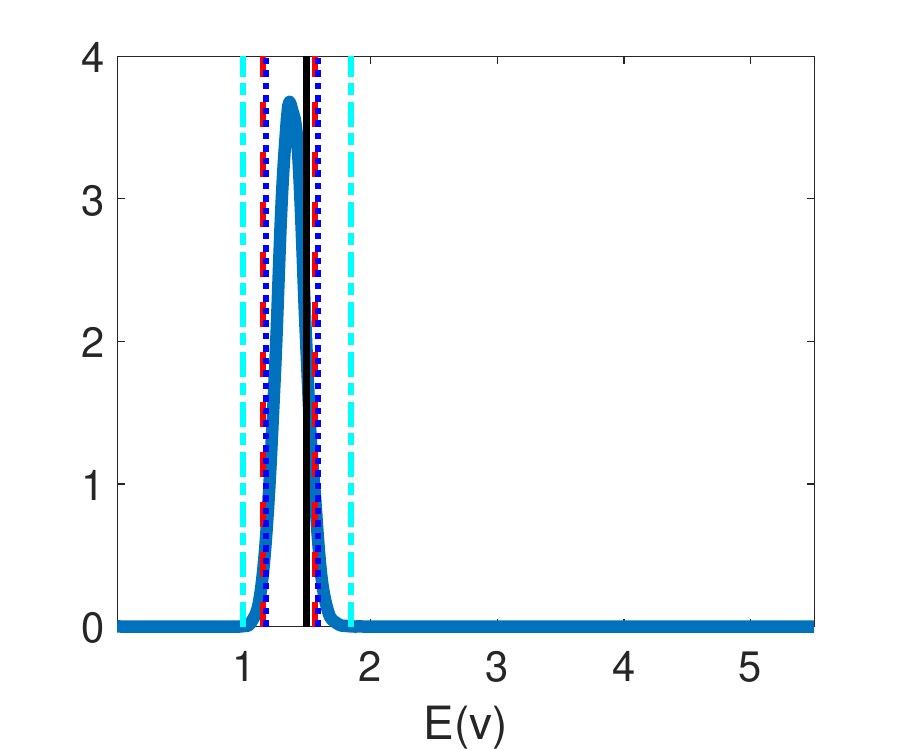} 
\caption{\small{Mean doctor visits }}
  \end{subfigure}
  \begin{subfigure}[b]{0.5\linewidth}
    \centering
    \includegraphics[width=\linewidth]{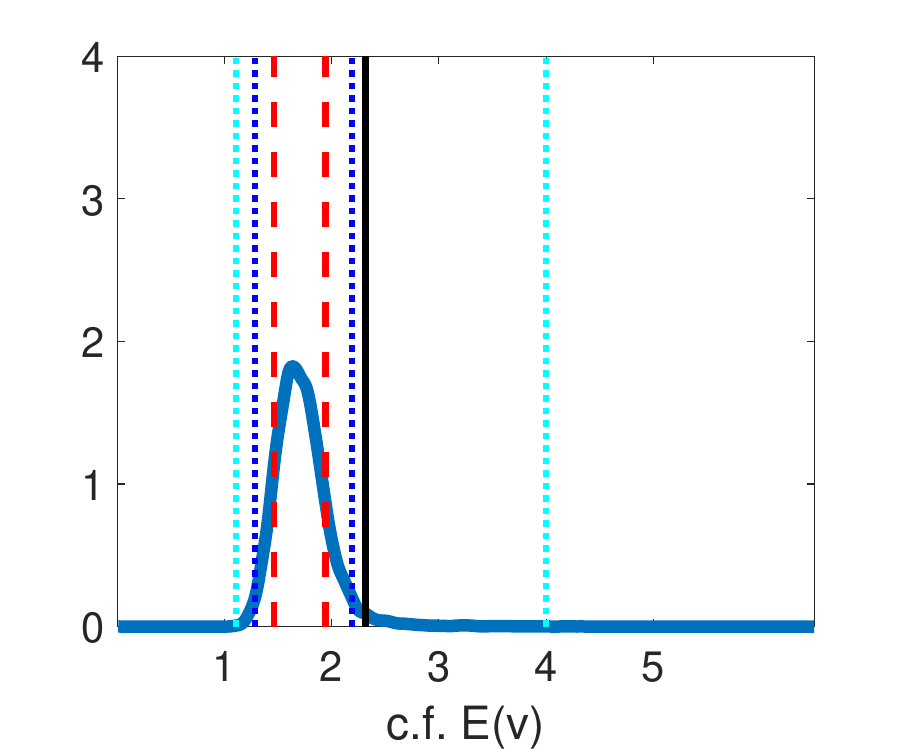} 
\caption{\small{ Mean c.f. doctor visits }}
  \end{subfigure}
\caption{\small{Actual and counterfactual (c.f.) expected number of doctor visits. 
Posterior density (blue solid),
95\% confidence interval (red dashed), 
95\% HPD interval (blue dotted),
$\hat{B}^{I_{E(v)}}_{0.95}$ (light blue dash-dotted),
and data generating value (black solid). }}
\label{fig_Gilleskie_EvEa}
\end{figure}
\FloatBarrier



\begin{table}[h]
  \centering
  \def\sym#1{\ifmmode^{#1}\else\(^{#1}\)\fi}%
  \begin{tabular}{l*{6}{c}}
    \toprule
    & \multicolumn{3}{c}{MLE under logit}  & \multicolumn{3}{c}{Semiparametric Bayes} \\
    \cmidrule(lr){3-4}\cmidrule(lr){5-7}
    & \multicolumn{1}{c}{True} 
    & \multicolumn{1}{c}{Est.} & \multicolumn{1}{c}{95\%CI (length)}  
    & \multicolumn{1}{c}{Est.} & \multicolumn{1}{c}{95\%HPD (length)}  
                                   & \multicolumn{1}{c}{$\hat{B}^{I_{E(v)}}_{0.95}$ (length)}      \\  
    \midrule
        \addlinespace
     $E(v)$  & 1.49  &  1.35    &  [1.15,    1.55]  (0.40)  & 1.37 &  [1.17,    1.58]    (0.41)   &  [0.99,    1.84]    (0.85)   \\
        \addlinespace
    $c.f.E(v)$ & 2.31  &  1.70      &  [1.46,    1.94]   (0.48) &  1.72   &  [1.28,    2.18]      (0.90)    &  [1.10,    3.99]    (2.88) \\
        \addlinespace        
    \bottomrule
  \end{tabular}
  
\caption{Estimated (Est.)  $E(v)$ and  $c.f.E(v)$: 
the MLE with its 95\% confidence interval and
the posterior mean with the 95\% HPD interval and $\hat{B}^{I_{E(v)}}_{0.95}$. }
\label{table_estimation}  
\end{table}

\FloatBarrier

These results illustrate the following general observations.  
First, the dynamic logit MLE can deliver misleading point and set estimates, especially for counterfactuals, when the shocks are not extreme value distributed.
Second, as is well known from \cite{MoonSchorfheide:12}, the HPD credible intervals for set identified parameters might be too short at least from the classical perspective.
Third, the credible intervals for the identified sets introduced in Section \ref{sec:inference_IS} are more conservative than the standard HPD intervals and might be preferred 
under set identification.
Finally, and most importantly, the whole posterior distribution (and not just point and set estimators) should be reported and taken into account in making decisions or policy recommendations.



\section{Extensions}
\label{sec:extensions}

In this section, we briefly discuss extensions of our results and methodology
to models with unknown $G$, continuous $X$, multiple agents, and finite horizon.

The assumption of known transition and initial probabilities for the observed states, $G$, can be relaxed with the following notation changes.  The expression for the likelihood function in \eqref{eq:likelihood} needs to be replaced by
\[
\sum_{i=1}^n \left( \sum_{t=1}^T \log p(d_{it} | x_{it}) + \sum_{t=2}^T \log Pr(x_{it}|x_{it-1},d_{it-1}) + \log Pr(x_{i1})\right),
\]
where $Pr(x_{it}|x_{it-1},d_{it-1})$ and $Pr(x_{i1})$ are the corresponding elements of $G$.
Assume that a prior distribution on the elements of $G$ has a positive and continuous density at the DGP value $G_0$.
To generalize the posterior consistency result in Theorem \ref{th:post_cons} and its proof it suffices 
to note that the collection of CCPs $P(\theta,F,G)$ (with the explicit dependence on $G$ now) is continuous in $G$ in the present settings, see, for example \cite{Norets_ddcm_diff_cont:09}.
The conclusion of the theorem would be 
$\Pi \big( \theta,F,G: ||(G_0,P(\theta_0,F_0,G_0))-(G,P(\theta,F,G)) ||> \delta \big  |D^n \big) \rightarrow 0$ almost surely.
The construction of the credible sets for the identified sets in Section \ref{sec:inference_IS} would use $G$ as point identified reduced form parameters along with the CCPs.

As demonstrated  by \cite{SrisumaLinton2012}, discretizations of continuous $X$ could nontrivially affect estimation results; hence, it would be desirable to extend our method to continuous $X$. 
The results in Lemma \ref{lm:ccp_emax} (analytical expressions for CCPs and \textit{Emax}), Lemma \ref{lm:LipschitzCont_Q_p} (continuity of CCPs and \textit{Emax} in $F$), Lemma \ref{lm:approx_by_mgumbles} (approximations by mixtures of Gumbels), and Lemma \ref{lm:Fcont_params} (continuity of mixtures in parameters) are proved for arbitrary/continuous $X$. A version of the posterior consistency result in Theorem \ref{th:post_cons} would also hold for a continuous and bounded $X$ under additional assumptions that the CCPs and the transition densities for the observed states are bounded away from zero. In this case the continuity of CCPs in $F$ in the sup norm established in Lemma \ref{lm:LipschitzCont_Q_p} would deliver positive prior probability for the  Kullback-Leibler neighbourhoods of the DGP, which is a sufficient condition for the Schwartz posterior consistency theorem.  The norm on the CCPs in the conclusion of the theorem would be one implied by the weak topology, and it would be harder to interpret.  
Unknown $G$ parameterized by a finite dimensional parameter can be handled similarly to the extension discussed in the previous paragraph; treating $G$ nonparametrically for continuous $X$ would be more involved.
The result on the exact matching of CCPs in Lemma \ref{lm:exact_match_ccps} would not hold for continuous $X$, but it is not essential for the proposed method.  
Solutions of the Bellman equations for continuous $X$ can be obtained using 
sieve approximations (\cite{KristensenMogensenMoonSchjerning2021}) or random grids (\cite{Rust:97}, \cite{ImaiJainChing:09}, \cite{Norets_eca:09}).
However, the main challenge for implementing our method for models with continuous $X$ appears to be a fast and accurate computation of the derivatives of the CCPs and the likelihood, which we leave to future work.

Single agent models with extreme value distributed utility shocks were  extended to dynamic discrete games with multiple agents in \cite{AguirregabiriaMiraGames:07} and \cite{PesendorferSchmidtDengler:08}, where in the first stage choice probabilities are estimated nonparametrically and then structural parameters are estimated from equilibrium conditions and the estimated choice probabilities.
A direct extension of our likelihood based method to dynamic games would have to handle the difficulties arising from possible multiple equilibria.  The likelihood function would not be defined in this case (same values of shocks can be consistent with different choice probabilities).  Flexible modelling of an equilibrium selection mechanism is a non-trivial problem, hence we do not pursue it in the present paper.  Nevertheless, one way to exploit the ideas from our paper in the estimation of dynamic discrete games is to use the method of \cite{PesendorferSchmidtDengler:08} with several fixed location-scale mixtures of extreme value distributions instead of just one extreme value distribution to check the sensitivity of the estimation results to the distributional assumptions on the shocks.

To apply our method to models with a finite horizon, one would just need to use a backward induction method instead of a Newton method for solving the Bellman equations.

\section{Conclusion}

In this paper, we propose and implement a semiparametric Bayesian estimation method for dynamic discrete choice models that uses flexible mixture specifications for modeling the distribution of unobserved state variables. 
We establish approximation and posterior consistency results that provide frequentist asymptotic guarantees for the method.
Our approach is shown to perform well in practice for binary and multinomial choice models.
The computational costs of solving the dynamic program with our mixture specification are comparable to those of the dynamic logit.
Even though the proposed MCMC algorithm requires many more iterations than a standard maximum likelihood method, the proposed framework is a robust and computationally tractable semiparametric alternative to the standard dynamic logit model; it provides more reliable inference, especially for counterfactuals.

\bibliographystyle{latex_common/ecta} 

\bibliography{latex_common/allreferences} 

\appendix

\section{Closed Form Expressions for CCPs and \textit{Emax}}

\label{sec:ccp_emax}
\begin{proof}[proof of Lemma \ref{lm:ccp_emax}]
Note that for $E1(z) = \int_{z}^{\infty}  e^{-t} / t dt$, $-E1(z) = Ei(-z) $ for positive $z>0$, where $Ei(z) =  -\int_{-z}^{\infty}  e^{-t} / t dt$ denotes the exponential integral function.

The mixture distribution in \eqref{eq:mixture} can be represented as 
$\epsilon|Z=k \sim \phi(\epsilon;\mu_k,\sigma_k)$, where the random variable 
$Z$ indicates a mixture component, $Pr(Z=k)=\omega_k$. 
Then, 
\[p(d)=\sum_{k=1}^m \omega_k p(d|Z=k) \mbox{ and }
E\left[\max_{j=0,1,\ldots,J} v_j + \epsilon_j \right] = \sum_{k=1}^m \omega_k E\left[\max_{j=0,1,\ldots,J} v_j + \epsilon_j | Z=k\right],
\]
where the dependence on the observed state $x$ is suppressed to simplify the notation.

For each $k=1,\ldots,m$, $Pr(d=0|Z=k)$ equals to 
\begin{align*}
& Pr(\epsilon_j \leq v_0 -v_j  \forall j=1,\ldots,J|Z=k) = \prod_{j=1}^J \exp \left[ -e^{ -\big( \frac{v_0 - v_j -\mu_{jk}}{\sigma_k} \big)} - e^{-\gamma}  \right]\\
& = \exp \left[ - e^{-\gamma} e^{-\frac{v_0}{ \sigma_k}} \sum_{j=1}^J e^{\frac{v_j+\mu_{jk}}{\sigma_k}} \right] = \exp \left[ - e^{-\gamma} e^{-\frac{v_0}{ \sigma_k }}e^{A_k}\right] = \exp [-e^{-a_{k}}].
\end{align*}
For $d \ne 0$, for each $k=1,\ldots,m$, note that $P(d|Z=k) = \int p(d|Z=k,\epsilon_d) f(\epsilon_d|Z=k) d\epsilon_d $ and
\begin{align*}
 &p(d|Z=k,\epsilon_d) = Pr\big( v_d +\epsilon_d \geq v_j +\epsilon_j \forall j \ne d|Z=k,\epsilon_d \big)\\
 &= Pr\big( 0 \leq \epsilon_d+v_d - v_0 \text{ and } \epsilon_j \leq \epsilon_d + v_d -v_j \forall j \ne 0, d  |Z=k,\epsilon_d \big)\\
 &= Pr\big(  \epsilon_j \leq \epsilon_d + v_d -v_j \forall j \ne 0, d  |Z=k,\epsilon_d\big)  1\left( 0 \leq \epsilon_d+v_d - v_0   \right)\\
 &= \prod_{j=1,j\ne d}^J \exp \left[  -e^{-\left(\frac{\epsilon_d + v_d -v_j -\mu_{jk}}{\sigma_k} \right)} e^{-\gamma}\right]1\left( 0 \leq \epsilon_d+v_d - v_0   \right)\\
 &= \prod_{j=1}^J \exp \left[  -e^{-\left(\frac{\epsilon_d + v_d -v_j -\mu_{jk}}{\sigma_k} \right)} e^{-\gamma}\right] \exp\left[ e^{-\left(\frac{\epsilon_d -\mu_{dk}}{\sigma_k}\right) }e^{-\gamma}\right]  1\left( 0 \leq \epsilon_d+v_d - v_0   \right).
 \end{align*}
Hence, $p(d|Z=k)$ equals to
\begin{align*}
&=\int \prod_{j=1}^J \exp \left[  -e^{-\left(\frac{s + v_d -v_j -\mu_{jk}}{\sigma_k} \right)} e^{-\gamma}\right] \exp\left[ e^{-\left(\frac{s -\mu_{dk}}{\sigma_k}\right) }e^{-\gamma}\right]  1\left( 0 \leq s+v_d - v_0   \right) \\
& \times \frac{1}{ \sigma_k } e^{-\left(\frac{s-\mu_{dk}}{ \sigma_k }\right)}e^{-\gamma}\exp \left[ -e^{-\left( \frac{s-\mu_{dk}}{  \sigma_k } \right)} e^{-\gamma} \right] ds\\
&=\int \prod_{j=1}^J \exp \left[  -e^{-\left(\frac{s + v_d -v_j -\mu_{jk}}{\sigma_k} \right)} e^{-\gamma}\right]  1\left( 0 \leq s+v_d - v_0   \right) \times \frac{1}{ \sigma_k } e^{-\left(\frac{s-\mu_{dk}}{ \sigma_k }\right)}e^{-\gamma} ds\\
&=\int_{s=-\infty}^{s=\infty} \prod_{j=1}^J \exp \left[  -e^{-\left(\frac{s -\mu_{dk}}{\sigma_k} \right)} e^{-\left(\frac{v_d+\mu_{dk}}{\sigma_k} \right)} e^{\left(\frac{u_j+\mu_{jk}}{\sigma_k} \right)} e^{-\gamma}\right]  1\left( 0 \leq s+v_d - v_0   \right) \\
& \times \frac{1}{ \sigma_k }  e^{-\left(\frac{s-\mu_{dk}}{ \sigma_k }\right)} e^{-\gamma} ds.
\end{align*}
Let $t= e^{\frac{s-\mu_{dk}}{ \sigma_k   }}$. Then $ds=- \sigma_k \frac{1}{t}dt$.
Note that $0 \leq s+v_d-v_0 \iff \frac{v_0 -v_d -\mu_{dk}}{  \sigma_k } \leq  \frac{s-\mu_{dk}}{ \sigma_k   } \iff e^{ - \left( \frac{v_0 -v_d -\mu_{dk}}{  \sigma_k } \right) }\geq t$ and that $s=\infty \iff t=0, s=-\infty \iff t=\infty$.
Now 
\begin{align*}
&p(d|Z=k) =
\int_{t=\infty}^{t=0} \prod_{j=1}^J \exp \left[  -te^{-\left(\frac{v_d+\mu_{dk}}{\sigma_k} \right)} e^{\left(\frac{v_j+\mu_{jk}}{\sigma_k} \right)} e^{-\gamma}\right]  1\left( t \leq e^{ - \left(\frac{v_0 -v_d -\mu_{dk}}{  \sigma_k } \right) }  \right)\\
&  \times \frac{1}{ \sigma_k } t e^{-\gamma} \left( 
- \sigma_k \frac{1}{t}dt  \right)\\
&=e^{-\gamma} \int_{t=0}^{t=  e^{ - \left(\frac{v_0 -v_d -\mu_{dk}}{  \sigma_k } \right)}} 
\exp\left[ -t e^{-\left(\frac{v_d+\mu_{dk}}{\sigma_k} \right)} e^{-\gamma} \sum_{j=1}^J  e^{\left(\frac{v_j+\mu_{jk}}{\sigma_k}\right)}\right] dt\\
&=e^{-\gamma} \frac{\exp\left[ -t e^{-\left(\frac{v_d+\mu_{dk}}{\sigma_k} \right)} e^{-\gamma} \sum_{j=1}^J  e^{\left(\frac{v_j+\mu_{jk}}{\sigma_k}\right)}\right]}{-e^{-\left(\frac{v_d+\mu_{dk}}{\sigma_k} \right)} e^{-\gamma} \sum_{j=1}^J  e^{\left(\frac{v_j+\mu_{jk}}{\sigma_k}\right)}}
\bigg\rvert_{t=0}^{t=  e^{ - \left(\frac{v_0 -v_d -\mu_{dk}}{  \sigma_k } \right)}}\\
&=\frac{-e^{\left(\frac{v_d+\mu_{dk}}{\sigma_k} \right)}}{\sum_{j=1}^J  e^{\left(\frac{v_j+\mu_{jk}}{\sigma_k}\right)}} 
 \Bigg\{ \exp\left[ -e^{ - \left(\frac{v_0 -v_d -\mu_{dk}}{  \sigma_k } \right)} e^{-\left(\frac{v_d+\mu_{jk}}{\sigma_k} \right)} e^{-\gamma} \sum_{j=1}^J  e^{\left(\frac{v_j+\mu_{jk}}{\sigma_k}\right)}\right]  -1 \Bigg\}\\
&=\frac{e^{\left(\frac{v_d+\mu_{dk}}{\sigma_k} \right)}}{\sum_{j=1}^J  e^{\left(\frac{v_j+\mu_{jk}}{\sigma_k}\right)}} 
 \Bigg\{1- \exp\left[ - e^{-\gamma} e^{ - \left(\frac{v_0 }{  \sigma_k } \right)} e^{-A_k} \right]   \Bigg\} 
=\frac{e^{\left(\frac{v_d+\mu_{dk}}{\sigma_k} \right)}}{\sum_{j=1}^J  e^{\left(\frac{v_j+\mu_{jk}}{\sigma_k}\right)}} \left[1-\exp[-e^{-a_k}] \right].
 \end{align*}
 We see that the probabilities sum to 1,
 \begin{align*}
\sum_{d=0}^{J}P(d|Z=k) =  
& \exp[-e^{-a_k}] 
+\sum_{d=1}^{J}\frac{e^{\left(\frac{v_d+\mu_{dk}}{\sigma_k} \right)}}{\sum_{j=1}^J  e^{\left(\frac{v_j+\mu_{jk}}{\sigma_k}\right)}} 
 \left[1-\exp[-e^{-a_k}] \right] =1.
 \end{align*}
This proves  \eqref{eq:ccp}.
To prove \eqref{eq:emax}, we have for each $k=1,\ldots,m$, 
\begin{align*}
&E\left[\max_{j=0,1,\ldots,J} v_j + \epsilon_j | Z=k\right] = \\
&E\left[\max_{j=0,1,\ldots,J} v_j + \epsilon_j | \max_{j=1,\ldots,J} v_j + \epsilon_j \leq v_0, Z=k \right] 
Pr\left[ \max_{j=1,\ldots,J} v_j + \epsilon_j \leq v_0 | Z=k\right] \\
+&E\left[\max_{j=0,1,\ldots,J} v_j + \epsilon_j | \max_{j=1,\ldots,J} v_j + \epsilon_j > v_0, Z=k \right] 
Pr\left[ \max_{j=1,\ldots,J} v_j + \epsilon_j > v_0 | Z=k\right].
\end{align*}
First, note that $Pr\left[ \max_{j=1,\ldots,J} v_j + \epsilon_j \leq v_0 | Z=k\right] =Pr\left[ v_j + \epsilon_j \leq v_0,  \forall j=1,\ldots,J| Z=k\right]$
\begin{align*}
&=\prod_{j=1}^J \exp\left[ -e^{ - \left( \frac{v_0 -v_j -\mu_{jk}}{  \sigma_k } \right) } e^{-\gamma}\right] 
=\exp \left[ -e^{-\gamma} \sum_{j=1}^J e^{ - \left( \frac{v_0 -v_j -\mu_{jk}}{  \sigma_k } \right) }\right] .
\end{align*}
Note that  $\sum_{j=1}^J e^{ - \left( \frac{v_0 -v_j -\mu_{jk}}{  \sigma_k } \right) } 
= e^{-\frac{v_0}{   \sigma_k  }}
\sum_{j=1}^J e^{-\frac{v_j+\mu_{jk}}{   \sigma_k  } } 
= e^{-\frac{v_0}{   \sigma_k  }} e^{A_k} = \exp\left[ -\frac{v_0}{   \sigma_k  } + A_k \right] 
$$\\= \exp\left[ -\frac{v_0-A_k \sigma_k}{   \sigma_k  }  \right] $.
Hence, 
$Pr\left[ \max_{j=1,\ldots,J} v_j + \epsilon_j \leq v_0 | Z=k\right] =\exp \left[ -e^{\left( -\frac{v_0-A_k \sigma_k}{   \sigma_k  }  \right) }e^{-\gamma}  \right] =\exp[-e^{-a_k}]$.
This means, $\max_{j=1,\ldots,J} v_j + \epsilon_j \leq v_0 | Z=k \sim \phi\left(\cdot;  A_k \sigma_k,  \sigma_k \right)$.
Next, we have 
\begin{align*}
&E\left[\max_{j=0,1,\ldots,J} v_j + \epsilon_j | \max_{j=1,\ldots,J} v_j + \epsilon_j > v_0, k \right] \\
&=E\left[\max_{j=1,\ldots,J} v_j + \epsilon_j | \max_{j=1,\ldots,J} v_j + \epsilon_j > v_0, k \right] 
= \int_{v_0}^{\infty} \frac{y f_{\max_{j=1,\ldots,J} v_j + \epsilon_j}  (y|Z=k) }{Pr\left[ \max_{j=1,\ldots,J} v_j + \epsilon_j > v_0  \right]} dy.
\end{align*}
Note that 
$\int_{v_0}^{\infty} y f_{\max_{j=1,\ldots,J} v_j + \epsilon_j}  (y|Z=k)  dy=
\int_{v_0}^{\infty} y  \frac{1}{   \sigma_k } e^{\left( -\frac{v_0-A_k \sigma_k}{   \sigma_k  }  \right) }e^{-\gamma}  \left[ -e^{\left( -\frac{v_0-A_k \sigma_k}{   \sigma_k  }  \right) }e^{-\gamma} \right]dy $.
Let $z = \gamma +\frac{y}{  \sigma_k   } -A_k$. Then $dz = \frac{1}{  \sigma_k   }dy$. 
Note that $y=v_0 \to z = \gamma + \frac{v_0}{  \sigma_k} -A_k = a_k$ and $y=\infty \to z = \infty$. Note that 
$y =  \sigma_k (z - \gamma +A_k)$.
We have 
\begin{align*}
&\int_{v_0}^{\infty} y f_{\max_{j=1,\ldots,J} v_j + \epsilon_j}  (y|Z=k)  dy =
\int_{a_k}^{\infty} \sigma_k (z - \gamma +A_k)\frac{1}{   \sigma_k } e^{-z}\exp\left[ -e^{-z}\right] ( \sigma_k dz ) \\
&=\sigma_k \int_{a_k}^{\infty}  z  e^{-z}\exp\left[ -e^{-z}\right] dz +\sigma_k (A_k-\gamma) \int_{a_k}^{\infty}   e^{-z}\exp\left[ -e^{-z}\right] dz \\
&=\sigma_k \left[ \gamma - a_k \exp\left( -e^{-a_k}\right) +E1(e^{-a_k}) \right]+\sigma_k (A_k-\gamma) \left[  1 - \exp\left( -e^{-a_k}\right) \right].
\end{align*}
Finally, $E\left[\max_{j=0,1,\ldots,J} v_j + \epsilon_j | Z=k\right]=v_0 \exp\left( -e^{-a_k}\right) + \int_{v_0}^{\infty} y f_{\max_{j=1,\ldots,J} v_j + \epsilon_j}  (y|Z=k)  dy $
\begin{align*}
&= v_0 \exp\left( -e^{-a_k}\right) +\\
&\sigma_k \left[  \gamma - a_k \exp\left( -e^{-a_k}\right) +E1(e^{-a_k}) +A_k -\gamma -A_k \exp\left( -e^{-a_k}\right) + \gamma \exp\left( -e^{-a_k}\right)       \right] \\
&=\sigma_k \left[ \exp\left( -e^{-a_k}\right)\left(  \frac{v_0}{ \sigma_k }  + \gamma -A_k -a_k \right) +E1(e^{-a_k})  + A_k    \right] \\
&=\sigma_k \left[ A_k +E1(e^{-a_k})  \right].
\end{align*}
This proves  \eqref{eq:emax}.
\end{proof}



\section{Proofs and Intermediate Results}
\label{sec:proofs}

\begin{proof}[proof of Lemma \ref{lm:LipschitzCont_Q_p}]
We first show Lipschitz continuity of the \textit{Emax} function. Suppose  there are two distributions $F_1$
and $F_2$. Define the corresponding \textit{Emax} functions $Q_i(x) = Q(x; F_i)$. 
\begin{align*}
Q_2(x) - Q_1(x) = 
& \bigg(
Q_2(x) - \int \max_d \bigg[ u(x,d) + \beta G_x^d(Q_1) + \epsilon_d \bigg] dF_2(\epsilon)
\bigg) \\
& +
\bigg(
 \int \max_d \bigg[ u(x,d) + \beta G_x^d(Q_1) + \epsilon_d \bigg] dF_2(\epsilon) -Q_1(x)
\bigg).
\end{align*}
The term in the first parentheses of the right hand side is bounded by
\begin{align*}
&\int \max_d \bigg[\left(  u(x,d) + \beta G_x^d(Q_2) + \epsilon_d \right) -\left(  u(x,d) + \beta G_x^d(Q_1) + \epsilon_d \right)  \bigg] dF_2(\epsilon)\\
&=
\beta \max_d \bigg[G_x^d(Q_2) - G_x^d(Q_1) \bigg]\\
&\leq 
\beta \max_d \sum_{x' \in X} | Q_2(x') - Q_1(x') | G_x^d(x') \\
&\leq \beta  || Q_2-Q_1||,
\end{align*}
where $||Q_2-Q_1|| = \sup_{x \in X} | Q_1(x) - Q_1(x)|$. 
The term in the second parentheses is 
\begin{align*}
& \int \max_d \left[u(x,d) + \beta G_x^d(Q_1) + \epsilon_d  \right]\left[f_2(\epsilon)-f_1(\epsilon)\right]d(\epsilon)\\
& 
\leq 
\int \left[ \bar{u}+ \beta \sup_x Q_1(x) + \sum_d |\epsilon_d| \right] \left|f_2(\epsilon)-f_1(\epsilon)\right|d(\epsilon)\\
&\leq 
c \int \left[ 1 + \sum_d |\epsilon_d| \right]|f_2(\epsilon)-f_1(\epsilon)|d(\epsilon)= c \rho(F_1,F_2)
\end{align*}
for $c=\max\{1, \bar{u}+ \beta (\bar{u} + \sum_d \int |\epsilon_d|dF_1(\epsilon))/(1-\beta)\}$. Thus, 
\begin{align*}
Q_2(x) - Q_1(x) \leq \beta ||Q_2-Q_1|| + c\rho(F_1,F_2), 
\end{align*}
for each $x \in X$. Finally, we have 
\begin{align*}
||Q_2-Q_1|| \leq \frac{c}{1-\beta} \rho(F_1,F_2) = C \rho(F_1,F_2).
\end{align*}
This proves the local Lipschitz continuity of the
\textit{Emax} function. 

Given some $d^* \in \{0,1,\ldots,J\}$ 
and $x$, define the set 
\[
S_1= \bigg\{\epsilon: u(x,d^*) + \beta G_x^{d^*} (Q_1) + \epsilon_{d^*} \geq  u(x,d) + \beta G_x^{d} (Q_1) + \epsilon_{d}, \forall d  \bigg\}, 
\]
on which $d^*$
is optimal under $F_1$. Similarly, define  $S_2$. With this notation, 
\[
p(d^* | x; F_i) = \int 1(S_i) dF_i(\epsilon), 
\]
for $i = 1, 2$, where $1(\cdot)$ is an indicator function. 
Now,
\[
p(d^* | x; F_1)- p(d^* | x; F_2) = \int [1(S_1)-1(S_2) ]  dF_1+ \int 1(S_2) d(F_1-F_2).
\]
The second integral is bounded by $\rho(F_1,F_2)$. 
Note that $1(S_1)-1(S_2)$ is bounded above by 
\begin{align*}
\sum_d 
&1( \{ \text{ $d^*$ is optimal under $F_1$ but $d$ is optimal under $F_2$} \} )\\
+
&1( \{ \text{ $d$ is optimal under $F_1$ but $d^*$ is optimal under $F_2$} \} ).
\end{align*}
If $d^*$ is optimal under $F_1$ and $d$ is optimal under $F_2$, then 
\[
u(x,d) - u(x,d^*) + \beta G_x^d(Q_1) - \beta G_x^{d^*} (Q_1) 
\leq 
\epsilon_{d^*} - \epsilon_d 
\leq 
u(x,d) - u(x,d^*) + \beta G_x^d(Q_2) - \beta G_x^{d^*} (Q_2). 
\]
Denote this interval by $A_1$.
Similarly, 
if $d$ is optimal under $F_1$ and $d^*$ is optimal under $F_2$, then 
\[
u(x,d^*) - u(x,d) + \beta G^{d^*}(Q_1) - \beta G_x^{d} (Q_1) 
\leq 
\epsilon_{d^*} - \epsilon_d 
\leq 
u(x,d^*) - u(x,d) + \beta G^{d^*} (Q_2) - \beta G_x^{d} (Q_2).
\]
Denote this interval by $A_2$.
Note that the length of the intervals $A_1$ and $A_2$ is bounded by $2 \beta || Q_2-Q_1 ||$. 
Also, by condition (ii) of the lemma, the density of the difference $\epsilon_{d^*} - \epsilon_d$ implied by $F_1$ is bounded by some positive constant $\bar{f}>0$.
Thus,
\begin{align*}
& \int [1(S_1)-1(S_2) ]  dF_1
\leq
 \sum_d\int 1(\epsilon_{d^*} - \epsilon_d \in A_1) dF_1(\epsilon) 
 +
  \sum_d\int 1(\epsilon_{d^*} - \epsilon_d \in A_2) dF_1(\epsilon)
	\\
	& \leq 4 (J+1)  \beta || Q_2-Q_1 || \bar{f} \leq 4 (J+1) \beta C \rho(F_1,F_2),
\end{align*}
where the last inequality follows by the proven Lipschitz continuity of the \textit{Emax} function.
In summary, we have, for each $d^*=0,1,\ldots,J$, 
\begin{align*}
| p(d^* | x; F_1)- p(d^* | x; F_2) | \leq  4 (J+1) \beta C \rho(F_1,F_2) + \rho(F_1,F_2) = C' \rho(F_1,F_2), 
\end{align*}
for all $x \in X$. 
\end{proof}

\begin{proof}[Proof of Lemma \ref{lm:approx_by_mgumbles}]
\begin{align}
&\rho\left( f(\cdot), \ \sum_{j=1}^m \omega_j \phi\left( \cdot;   q_j, \sigma   \right) \right)
= \int \left(1 + \sum_{i=1}^I |z_i|  \right)  \bigg| f(z) - \sum_{j=1}^m \omega_j \phi \left( z; q_j, \sigma \right)\bigg| dz\nonumber\\
&\leq 
\sum_{i=1}^I \int (1+|z_i|) \bigg| f(z) - \sum_{j=1}^m \omega_j  \phi \left(z;q_j,\sigma \right) \pm \int  \phi\left(z;\mu,\sigma\right) f(\mu) d\mu \bigg|  dz\nonumber\\
&\leq 
\label{eq:f_contm} 
\sum_{i=1}^I \int (1+|z_i|) \bigg| f(z)  - \int  \phi\left(z;\mu, \sigma \right) f(\mu) d\mu \bigg| dz  \\
&+
\int (1+|z_i|) \bigg|  \int  \phi\left(z; \mu, \sigma \right) f(\mu) d\mu -  \sum_{j=1}^m \omega_j \phi \left( z; q_j, \sigma \right) \bigg| dz. 
\label{eq:contm_discm}
\end{align}
With the change of variable of $\mu_i$ to $\theta_i =\frac{z_i-\mu_i}{\sigma}, i=1,\ldots,I$, \eqref{eq:f_contm} is bounded by 
\begin{align*}
&\sum_{i=1}^I \int \int (1+|z_i|) \bigg| f(z)  -f(z-\sigma \theta) \bigg| \phi(\theta) d\theta dz \\
&\leq \sum_{i=1}^I \int \int (1+|z_i|) L_f(z) e^{\tau \sigma  ||\theta||_2} \sigma  ||\theta||_2   \phi(\theta) d\theta dz \\
&=\sigma \bigg[\int ||\theta||_2 e^{\tau \sigma  ||\theta||_2}   \phi(\theta) d\theta \bigg] \sum_{i=1}^d \bigg[ \int (1+|z_i|)  L_f(z) dz \bigg],
\end{align*}
where the inequality follows from the smoothness assumption \eqref{eq:f_smooth}.
Lemma \ref{lm:lemma_bounded_integral} shows that the term in the first square brackets is bounded for sufficiently small $\sigma$.
The last term is finite by assumption \eqref{eq:L_inf_smooth}. Hence, we can choose  $\sigma$ small  enough to make the term above arbitrarily small.

To bound \eqref{eq:contm_discm}, let $A_j, j=0,1,\ldots,m$, be a partition of $\mathbb{R}^I$ 
consisting of 
adjacent hypercubes $A_1,\ldots,A_m$ with sides $h_m^{1/I}$ so that they are collectively centered at zero and   
the rest of the space is $A_0$. As $m$ increases, the fine part of the partition becomes finer, $h_m \to 0$ and $m \to \infty$. 
Also, it covers larger and larger parts of $\mathbb{R}^I$; 
that is, $m \cdot h_m \to \infty$ as  $m \to \infty$.
Define $\omega_j = \int_{A_j}  f(\mu) d\mu $.
 This implies, 
\begin{align*}
&\int \phi\left(z;\mu, \sigma \right) f(\mu) d\mu -  \sum_{j=1}^m \omega_j \phi \left( z;q_j, \sigma \right)\\
&= 
\sum_{j=0}^m \int_{A_j}   \phi\left(z; \mu, \sigma \right) f(\mu) d\mu -\sum_{j=1}^m \int_{A_j}  \phi\left(z;q_j, \sigma \right) f(\mu) d\mu\\
&=
\sum_{j=1}^m \int_{A_j}  \bigg[ \phi\left( z; \mu, \sigma \right)-  \phi\left(z; q_j, \sigma \right)\bigg] 
+ \int_{A_0} \phi\left(z; \mu, \sigma \right) f(\mu) d\mu.
\end{align*}
The expression in \eqref{eq:contm_discm} can be bounded as follows,
\begin{align*}
&\sum_{i=1}^I \int (1+|z_i|) \bigg|  \int   \phi\left(z;\mu, \sigma\right) f(\mu) d\mu  -\sum_{j=1}^m \omega_j  \phi\left(z;q_j, \sigma\right)  \bigg| dz \\
& \leq \sum_{i=1}^I \sum_{j=1}^m \int \int_{A_j}  (1+|z_i|) \bigg|  \phi\left(z;\mu, \sigma\right) - \phi\left(z;q_j, \sigma\right)  \bigg| f(\mu) d\mu dz\\
&+ \sum_{i=1}^I \int \int_{A_0}  (1+|z_i|) \phi\left(z;\mu, \sigma\right)  f(\mu) d\mu dz.
\end{align*}
Consider the change of variable of $z$ to $y=\frac{z-\mu}{\sigma}$ and define $\delta_j = \frac{q_j - \mu }{\sigma}$, $j=1,\ldots,m$ so that $\frac{z-q_j}{\sigma} =y - \delta_j$
and the above equals to 
\[
\sum_{i=1}^I \sum_{j=1}^m \int \int_{A_j} (1+|\mu_i + \sigma y_i | ) \bigg| \phi(y) - \phi(y-\delta_j) \bigg| f(\mu) d\mu dy
+
\int \int_{A_0}  (1+|\mu_i + \sigma y_i | ) \phi(y)  f(\mu) d\mu dy.
\]
As $\mu, q_j \in A_j$, we have $||\delta_j||_\infty \leq \big\Vert \frac{q_j - \mu }{\sigma} \big\Vert_1 \leq \frac{I\cdot h_m^{1/d}}{I \cdot \sigma} \equiv \bar{\delta}$, which will be made small. 
By Lemma \ref{lm:Lipschitz_Gumbel_location}, 
$\phi(\cdot)$ is Lipschitz continuous and $| \phi(y) - \phi(y-\delta_j) | \leq \bar{L}_\phi(y) || \delta_j||_\infty$, which implies that $| \phi(y) - \phi(y-\delta_j) | \leq \bar{L}_\phi(y)  \bar{\delta}$.
 
The first term above is bounded by $\bar{\delta}$ times 
\begin{align*}
\sum_{i=1}^I \left(  \sum_{j=1}^m c_{1j}  \int_{A_j} (1+|\mu_i| ) f(\mu) d\mu  +c_{2j} \sigma  \right),
\end{align*}
where $c_{1j} = \int \bar{L}_\phi(y) dy$ and $c_{2j}$ is such that $\int |y_i| \bar{L}_\phi(y) dy < c_{2j}$ for $i=1,\ldots,I$. 
The bound approaches to zero as $m \to \infty$. Since $\phi(\cdot)$ has bounded first moment, the second integral is bounded by 
\begin{align*}
\sum_{i=1}^I 
 \int_{A_0}  (1+|\mu_i| + \sigma c )  f(\mu) d\mu,  
\end{align*}
for some constant $c>0$. As $mh_m\to \infty$, this term goes to zero. 

\end{proof}

\begin{proof}[Proof of Lemma \ref{lm:Fcont_params}]
We have
\begin{align}
& \rho \left(  F^1, F^2 \right)
= \int ( 1 + \sum_j | \epsilon_j| ) \bigg| f^1(\epsilon)  -f^2(\epsilon)  \bigg| d \epsilon \notag \\
&=
\int  ( 1 + \sum_j | \epsilon_j| )
 \bigg|
\sum_{k=1}^m 
\omega_k^1  \phi\left( \epsilon; \mu_k^1, \sigma_k^1 \right)
 -
 \omega_k^2  \phi\left( \epsilon; \mu_k^2, \sigma_k^2 \right) 
 \pm 
  \omega_k^1  \phi\left( \epsilon; \mu_k^2, \sigma_k^2 \right) 
 \bigg| d \epsilon \notag \\
 &
\label{eq:S_w_phi_phi}
\leq
 \sum_{k=1}^m \omega_k^1
\int  
( 1 + \sum_j | \epsilon_j| )
 \bigg| \phi\left( \epsilon; \mu_k^1, \sigma_k^1 \right)
 -
\phi\left( \epsilon; \mu_k^2, \sigma_k^2 \right) \bigg|
d \epsilon\\
&
\label{eq:S_w_w_phi}
+
\sum_{k=1}^m |   \omega_k^1- \omega_k^2  |
\int
( 1 + \sum_j | \epsilon_j| )
\phi\left( \epsilon; \mu_k^2, \sigma_k^2 \right)  d \epsilon.
\end{align}
Expression \ref{eq:S_w_w_phi} is bounded by a positive constant times $\sum_{k=1}^m |   \omega_k^1- \omega_k^2  |$.

To bound \ref{eq:S_w_phi_phi} note that 
\begin{align*}
&\int  ( 1 + \sum_j | \epsilon_j| ) 
 \bigg| \phi\left( \epsilon; \mu_k^1, \sigma_k^1 \right)
 -
\phi\left( \epsilon; \mu_k^2, \sigma_k^2 \right) \bigg|
d \epsilon \nonumber \\
&=
\int  ( 1 + \sum_j | \epsilon_j| ) 
\bigg| 
\phi\left( \epsilon; \mu_k^1, \sigma_k^1 \right)
-\phi\left( \epsilon;\mu_k^2, \sigma_k^2\right) 
\pm\phi\left(\epsilon; \mu_k^2, \sigma_k^1\right)  
\bigg|d \epsilon \nonumber \\
&\leq 
\int  ( 1 + \sum_j | \epsilon_j| ) 
\bigg| 
\phi\left( \epsilon; \mu_k^1, \sigma_k^1\right)
- \phi\left( \epsilon;\mu_k^2, \sigma_k^1\right)  
\bigg|d \epsilon \\
&+
\int  ( 1 + \sum_j | \epsilon_j| ) 
\bigg| 
\phi\left( \epsilon; \mu_k^2, \sigma_k^1\right)
- \phi\left( \epsilon; \mu_k^2, \sigma_k^2\right)  
\bigg|d \epsilon. 
\end{align*}
With the change of variable $z=\frac{\epsilon-\mu_k^1}{\sigma_k^1}$ and by Lemma \ref{lm:Lipschitz_Gumbel_location}, the first part can be bounded as follows 
\begin{align*}
&\int \left(1 + \sum_j |\sigma_k^1 z_j +\mu_{jk}^1 | \right) \Bigg| \phi(z;0_{J+1},1) - \phi\left( z- \frac{\mu_k^2-\mu_k^1}{\sigma_k^1}; 0_{J+1},1 \right) \Bigg| dz\\
&\leq \bigg\Vert \frac{\mu_k^2-\mu_k^1}{\sigma_k^1} \bigg\Vert_\infty
\int \left(1 + \sum_j \bigg|\sigma_k^1 z_j +\mu_{jk}^1 \bigg| \right) \bar{L}_\phi(z)  dz,
\end{align*}
where the integral is bounded. By Lemma \ref{lm:Lipschitz_Gumbel_scale}, the second part can be bounded as follows 
\[
\int  ( 1 + \sum_j | \epsilon_j| ) 
\bigg| 
\phi\left( \epsilon; 0_{J+1}, \sigma_k^1 \right)
- 
\phi\left( \epsilon;0_{J+1},\sigma_k^2 \right)  
\bigg|d \epsilon 
\leq 
\bigg| \frac{1}{\sigma_k^1}-\frac{1}{\sigma_k^2} \bigg| 
\int  ( 1 + \sum_j | \epsilon_j| ) 
\bar{M}_\phi(\epsilon)
d \epsilon, 
\]
where the integral is bounded. 
To summarize, we have 
\begin{align*}
\rho \left(  F^1, F^2 \right) \leq \sum_{k=1}^m \omega_k^1
 \left[
 c_1\bigg\Vert \frac{\mu_k^2-\mu_k^1}{\sigma_k^1} \bigg\Vert_\infty
 +
  c_2\bigg| \frac{1}{\sigma_k^1}-\frac{1}{\sigma_k^2} \bigg|
 \right]
 +
 c_3 \sum_{k=1}^m |   \omega_k^1- \omega_k^2  |, 
\end{align*}
for some positive finite constants $c_i$'s, which implies the claimed result.
\end{proof}

\begin{proof}[Proof of Lemma \ref{lm:exact_match_ccps}]
Let us denote the set of shock values at which alternative $j$ is optimal at the observed state $k$ by
\[
E_{jk} = \{(\epsilon_1,\ldots,\epsilon_J): \;  
u(k,j;\theta_0) + \beta G_k^j Q + \epsilon_j \geq  u(k,d;\theta_0) + \beta G_k^d Q + \epsilon_d, \forall d=0,1,\ldots,J \},
\]
where the normalization $\epsilon_0=0$ is used.
Sets $\{E_{jk}, j=0,1,\ldots,J\}$ define a partition of $\mathbb{R}^J$ (up to the overlapping boundaries) for each $k \in \{1,\ldots,K\}$.
Consider a refinement of these $K$ partitions,
\[
\big \{A:\; A = \bigcap_{k=1}^{K} E_{j_k k},\; j_k \in \{0,\ldots,J\},\; \lambda(A)>0 \big\}
=\{A_1,\ldots,A_L \},
\]
where $\lambda$ is the Lebesgue measure.
Let us define
\[
q_l = \int_{A_l} dF_0(\epsilon_{1J}) \in (0,1),
\]
\[
r_l = \int_{A_l} \epsilon_{1J} dF_0(\epsilon_{1J}) \in \mathbb{R}^J,
\]
where $\epsilon_{1J}=(\epsilon_1,\ldots,\epsilon_J)^\prime$.
It follows from a characterization of the identified sets under unknown distribution of shocks in \cite{Norets_ddc_mult:11} that 
any distribution that implies the same $\{q_l, r_l, l=1,\ldots,L\}$ delivers the same CCPs. 
The key to this result is that the shocks enter the utility additively and the integrals over the shocks in the Bellman equations can be replaced with 
expressions depending on the distribution of shocks only through $\{q_l, r_l, l=1,\ldots,L\}$.
In what follows, we construct a finite mixture distribution that delivers the same $\{q_l, r_l, l=1,\ldots,L\}$ as $F_0$.
Specifically, consider a mixture with $L \cdot (J+1)$ components
\begin{align*}
\sum_{l=1}^L \sum_{j=0}^J \omega_{lj} \phi(\cdot; \mu_{lj}, \sigma),
\end{align*}
where we fix the locations as follows. 
Since $F_0$ is assumed to have a positive density, by Lemma 2 in \cite{Norets_ddc_mult:11}, $r_l/q_l \in int(A_l)$ for all $l$.
This is an implication of the supporting hyperplane theorem and the closedness and convexity of $E_{jk}$'s and thus $A_l$, $l=1,\ldots,L$
($r_l/q_l$ is the expectation of $\epsilon_{1J}$ conditional on $\epsilon_{1J} \in A_l$).
Therefore, $\exists \Delta > 0$ such that 
\begin{align*}
\mu_{lj} &= r_l/q_l +e_j \cdot \Delta \in int(A_l), \, j=1,\ldots,J, \\
\mu_{l0} &= r_l/q_l - (1,\ldots,1)^\prime \cdot \Delta \in int(A_l), \, l=1,\ldots,L,
\end{align*}
 where $e_j$ is a column vector of length $J$ with $1$ in the $j$th coordinate and $0$'s in the others. 

The mixing weights $\omega_{lj}$ are chosen as a solution to the following linear system of equations that matches $\{q_l, r_l, l=1,\ldots,L\}$,
\begin{align}
q_l & = \sum_{\tilde{l}=1}^{L} \sum_{j=0}^J \omega_{\tilde{l}j} \int_{A_l} \phi(\epsilon_{1J}; \mu_{\tilde{l}j}, \sigma) d\epsilon_{1J},
\nonumber \\ 
\label{eq:match_qlrl}
r_l& =  \sum_{\tilde{l}=1}^{L} \sum_{j=0}^J \omega_{\tilde{l}j} \int_{A_l} \epsilon_{1J} \phi(\epsilon_{1J}; \mu_{\tilde{l}j}, \sigma) d\epsilon_{1J}, 
\;\;\;\;\;
l \in \{1,\ldots,L\}.
\end{align}
Let us show that for sufficiently small $\sigma$, this linear system has a unique solution that belongs to the $(L(J+1)-1)$-simplex.
First, note that since $\mu_{lj} \in int(A_l)$,
\[
\lim_{\sigma\rightarrow 0} \int_{A_l} \phi(\epsilon_{1J}; \mu_{\tilde{l}j}, \sigma) d\epsilon_{1J} = 1\{l=\tilde{l}\}
\mbox{ and }
\lim_{\sigma\rightarrow 0} \int_{A_l} \epsilon_{1J} \phi(\epsilon_{1J}; \mu_{\tilde{l}j}, \sigma) d\epsilon_{1J} = \mu_{lj} 1\{l=\tilde{l}\}.
\]
Thus, the limiting system corresponding to $\sigma=0$ in \eqref{eq:match_qlrl} 
is
\begin{equation}
\label{eq:match_qlrl_limit}
q_l  = \sum_{j=0}^J \omega_{lj}, \;\;\;
r_l =  \sum_{j=0}^J \omega_{lj} \mu_{lj}, 
\;\;\;\
l \in \{1,\ldots,L\}.
\end{equation}
Plugging the definition of $\mu_{lj}$ into \eqref{eq:match_qlrl_limit}, we obtain $\omega_{lj}=\omega_{l0}$, $j=1,\ldots,J$ and, thus, the limiting system \eqref{eq:match_qlrl_limit} has
the unique solution 
\[
\omega_{lj}^\ast = q_l/(J+1), \; l=1,\ldots,L,\; j=0,\ldots,J.
\]
It follows that the matrix of the linear coefficients in the limiting system is invertible.  
Since matrix inversion is continuous and $\omega_{lj}^\ast>0$, the system \eqref{eq:match_qlrl} has a unique strictly positive solution 
$\{\omega_{lj}^{\ast\ast}, \;l=1,\ldots,L,\; j=0,\ldots,J\}$ for all sufficiently small $\sigma$.
Since $\sum_{l=1}^L q_l=1$, $\sum_{l,j} \omega_{lj}^{\ast\ast}=1$ as well.

\end{proof}

\begin{lemma}
\label{lm:lemma_bounded_integral}

The term
\begin{align*}
\int ||z||_2 e^{\tau \sigma \cdot ||z||_2}   \phi(z; 0_d,1) dz 
\end{align*}
is bounded for any $\tau>0$ and sufficiently small $\sigma>0$
where 
$0_d$ is a column vector of zeros with length $d$.
\end{lemma}
\begin{proof}
By Cauchy-Schwartz inequality, the term is bounded by
\begin{align*}
\left(\int ||z||_2^2    \phi(z;0_d,1) dz \right)^{1/2}  \left(\int  e^{2\tau \sigma \cdot ||z||_2}   \phi(z;0_d,1) dz \right)^{1/2}.
\end{align*}
 The first term in the product is bounded. 
As 2-norms are bounded by 1-norms, the term in the second parentheses is bounded by 
\begin{align*}
&\int  e^{2\tau \sigma \cdot ||z||_1}   \phi(z;0_d,1) dz 
=
 \prod_{i=1}^d \int_{-\infty}^\infty e^{2 \tau \sigma  |z_i| } \phi (z_i) dz_i \\
&
\leq
 \prod_{i=1}^d \int_{-\infty}^\infty e^{2 \tau \sigma  |z_i| } e^{-z_i} dz_i 
= 
2^d \prod_{i=1}^d \int_0^\infty e^{2 \tau \sigma z_i -z_i} dz_i, 
\end{align*}
which is bounded for small enough $\sigma$.

\end{proof}

\begin{lemma}
\label{lm:Lipschitz_Gumbel_location}
The density
$\phi(\ \cdot \ ;\mu,\sigma)$ is locally Lipschitz continuous in the location parameter
with envelope $\bar{L}_\phi$
in a sense that,  for an arbitrary $\bar{\delta}$, 
if $\delta \in \mathbb{R}^d$ is bounded by $\bar{\delta}$, then
 \begin{align*}
 | \phi(z; 0_d,1) - \phi(z-\delta; 0_d,1) | \leq \bar{L}_\phi(z) || \delta||_\infty, 
 \end{align*}
  where 
$0_d$ is a column vector of zeros with length $d$,
$ ||\delta||_\infty = \max_{i=1,\ldots,d} |\delta_i | $, and
 \begin{align*}
\bar{L}_\phi(z) = L_\phi(z_1) \prod_{i=2}^d \phi(z_i) +  L_\phi(z_2)\phi(z_1 -\delta_1) \prod_{i=3}^d \phi(z_i) + \cdots +   L_\phi(z_d)\prod_{i=1}^{d-1} \phi(z_i-\delta_i),  
 \end{align*}
\begin{align*}
L_\phi(z)=
\begin{cases}
   K_1e^{-(z+\gamma)} e^{-e^{-K_2 (z+\gamma)}},& \text{if } z > -\gamma\\
     K_3e^{-2(z+\gamma)} e^{-e^{-K_4  (z+\gamma)}},              & \text{otherwise},
\end{cases}
\end{align*}
for some positive $K_i$'s that could depend on $\bar{\delta}$.
\end{lemma}

\begin{proof}
First, we establish local  Lipschitz continuity for the univariate case. We have for $z, \delta \in \mathbb{R}$
\begin{align*}
|\phi(z) - \phi(z-\delta)| = |\phi'(z-\tilde{\delta}) | \cdot |\delta|,
\end{align*}
for some $\tilde{\delta}$ between $0$ and $|\delta|$ where $\phi'(z) = \phi(z)(e^{-z-\gamma}-1)$.
Note that if $z>-\gamma$, then $|\phi'(z) | \leq c_1 \phi(z)$ and if $z<-\gamma$, then $|\phi'(z) | \leq c_2 e^{-2z -2\gamma- c_3e^{-z-\gamma}}$
for some constants $c_1,c_2,c_3>0$. From this, it follows that for $\tilde{\delta}$ which is bounded, $|\phi'(z-\tilde{\delta}) | \leq L_\phi(z)$
 where 
\begin{align*}
L_\phi(z)=
\begin{cases}
   K_1e^{-(z+\gamma)} e^{-e^{-K_2 (z+\gamma)}},& \text{if } z > -\gamma\\
     K_3e^{-2(z+\gamma)} e^{-e^{-K_4  (z+\gamma)}},              & \text{otherwise,}
\end{cases}
\end{align*}
for some positive $K_i$'s that could depend on $\bar{\delta}$.
Now, to show the result for the multivariate case, let $z,\delta \in \mathbb{R}^d$. 
With $ ||\delta||_\infty = \max_{i=1,\ldots,d} |\delta_i | $,
we have
\begin{align*}
&\phi(z; 0_d,1) -\phi(z-\delta; 0_d,1) = \prod_{i=1}^d \phi(z_i) -  \prod_{i=1}^d \phi(z_i-\delta_i) 
\pm \phi(z_1-\delta_1)  \prod_{i=2}^d \phi(z_i) 
\\ &
\;\;\; \pm \phi(z_1-\delta_1)  \phi(z_2-\delta_2)  \prod_{i=3}^d \phi(z_i) \cdots \pm \phi(z_1-\delta_1)  \phi(z_2-\delta_2) \cdots  \phi(z_d-\delta_d) 
\\&
=[\phi(z_1) - \phi(z_1-\delta_1) ] \prod_{i=2}^d \phi(z_i) 
+  \phi(z_1-\delta_1)[\phi(z_2) - \phi(z_2-\delta_2) ] \prod_{i=3}^d \phi(z_i) 
\\&
\;\;\; + \cdots +  \prod_{i=1}^{d-1}  \phi(z_i-\delta_i)[\phi(z_d) - \phi(z_d-\delta_d) ]. 
\end{align*}
By the proven  locally Lipschitz continuity for the univariate case, this is bounded by 
\begin{align*}
&L_\phi(z_1) |\delta_1|  \prod_{i=2}^d \phi(z_i) + L_\phi(z_2) |\delta_2| \phi(z_1-\delta_1) \prod_{i=3}^d \phi(z_i) + \cdots + L_\phi(z_d) |\delta_d|  \prod_{i=1}^{d-1} \phi(z_i-\delta_i)\\
&\leq ||\delta||_\infty \bigg[L_\phi(z_1)\prod_{i=2}^d \phi(z_i) +L_\phi(z_2) \phi(z_1-\delta_1) \prod_{i=3}^d \phi(z_i) + \cdots + L_\phi(z_d)  \prod_{i=1}^{d-1} \phi(z_i-\delta_i)  \bigg].
\end{align*}
\end{proof}

\begin{lemma}
\label{lm:Lipschitz_Gumbel_scale}
The density $\phi(z; \mu,\sigma)  $
is locally Lipschitz continuous in the inverse of the scale parameter with envelope $\bar{M}_\phi$  in a sense that given two scale parameters 
$\sigma^1,\sigma^2>0$,
 \[
 \bigg|  \phi\left( z;0_d,\sigma^1 \right) - \phi\left( z;0_d,\sigma^2 \right) \bigg| \leq 
   \bigg|  \frac{1}{\sigma^1}-  \frac{1}{\sigma^2} \bigg|
\bar{M}_\phi(z), \mbox{ where}
 \]
 \[
\bar{M}_\phi(z) =
  M_\phi(z_1)  \prod_{i=2}^d   \phi\left(z_i;0,\sigma^1 \right)
  +
M_\phi(z_2)  
 \phi\left(z_1;0,\sigma^2 \right)
 \prod_{i=3}^d   \phi\left(z_i;0,\sigma^1 \right)  
 +\cdots+
 M_\phi(z_d)  
 \prod_{i=1}^{d-1}  \phi\left(z_i;0,\sigma^2 \right)
 \]
  and 
$
M_\phi(z_i)= \frac{1}{\sigma^1} 
 L_\phi\left( \frac{z_i}{\sigma^1} \right)  |z_i| 
 +
     \phi\left(\frac{z_i}{\sigma^2} \right), \, i=1,\ldots,d
$
for $L_\phi(\cdot)$ defined in Lemma \ref{lm:Lipschitz_Gumbel_location}.
\end{lemma}

\begin{proof}
We first show the result for the univariate case.
Let $\sigma^1,\sigma^2 >0$ and $z\in \mathbb{R}$.
\begin{align*}
& \bigg| \phi\left( z; 0, \sigma^1 \right) -\phi\left(z; 0, \sigma^2 \right) \bigg| 
 =\bigg| \frac{1}{\sigma^1}\phi\left(\frac{z}{\sigma^1} \right)-\frac{1}{\sigma^2}\phi\left(\frac{z}{\sigma^2} \right) \pm \frac{1}{\sigma^1}\phi\left(\frac{z}{\sigma^2} \right)  \bigg|\\
&
\leq  \bigg| \frac{1}{\sigma^1}\phi\left(\frac{z}{\sigma^1} \right)- \frac{1}{\sigma^1}\phi\left(\frac{z}{\sigma^2} \right) \bigg|
 +
  \bigg| \frac{1}{\sigma^1}\phi\left(\frac{z}{\sigma^2} \right)- \frac{1}{\sigma^2}\phi\left(\frac{z}{\sigma^2} \right) \bigg|\\
  &=
   \frac{1}{\sigma^1} \bigg| \phi\left(\frac{z}{\sigma^1} \right) - \phi\left(\frac{z}{\sigma^1} -\left( \frac{z}{\sigma^1}-\frac{z}{\sigma^2} \right)  \right)  \bigg|
   +
    \phi\left(\frac{z}{\sigma^2} \right) \bigg|  \frac{1}{\sigma^1}-  \frac{1}{\sigma^2} \bigg| \\
&\leq 
 \frac{1}{\sigma^1} 
 L_\phi\left( \frac{z}{\sigma^1} \right)  \bigg|  \frac{1}{\sigma^1}-  \frac{1}{\sigma^2} \bigg| |z| 
 +
     \phi\left(\frac{z}{\sigma^2} \right) \bigg|  \frac{1}{\sigma^1}-  \frac{1}{\sigma^2} \bigg|
     = M_\phi(z)  \bigg|  \frac{1}{\sigma^1}-  \frac{1}{\sigma^2} \bigg|, 
\end{align*} 
for $L_\phi(\cdot)$ defined in Lemma \ref{lm:Lipschitz_Gumbel_location}.
In the multivariate case, 
\begin{align*}
&  \phi\left( z; 0_d, \sigma^1 \right) -\phi\left(z; 0_d, \sigma^2 \right)  
= \prod_{i=1}^d  \phi\left(z_i; 0, \sigma^1 \right) -  \prod_{i=1}^d  \phi\left(z_i;0,\sigma^2 \right)\\
&
\pm
 \phi\left(z_1;0,\sigma^2 \right)
 \prod_{i=2}^d   \phi\left(z_i;0,\sigma^1 \right)
\pm
 \phi\left(z_1;0,\sigma^2 \right)
 \phi\left(z_2;0,\sigma^2 \right)
 \prod_{i=3}^d  \phi\left(z_i;0,\sigma^1 \right)\\
&
\pm
\cdots 
\pm 
 \prod_{i=1}^d   \phi\left(z_i;0,\sigma^2 \right)\\
 &=
\bigg[ \phi\left(z_1;0,\sigma^1 \right)- \phi\left(z_1;0,\sigma^2 \right) \bigg]  
\prod_{i=2}^d   \phi\left(z_i;0,\sigma^1 \right)\\
&+
\phi\left(z_1;0,\sigma^2 \right)
\bigg[  \phi\left(z_2;0,\sigma^1 \right)- \phi\left(z_2;0,\sigma^2 \right) \bigg]  
\prod_{i=3}^d  \phi\left(z_i;0,\sigma^1 \right)\\
&+\cdots +
\prod_{i=1}^{d-1} \phi\left(z_i;0,\sigma^2 \right) 
\bigg[ \phi\left(z_d;0,\sigma^1 \right)- \phi\left(z_d;0,\sigma^2 \right)  \bigg].
\end{align*}
With the result for the univariate case, the last expression is bounded by 
\begin{align*}
  & \bigg|  \frac{1}{\sigma^1}-  \frac{1}{\sigma^2} \bigg| 
  \bigg\{
  M_\phi(z_1)  \prod_{i=2}^d   \phi\left(z_i;0,\sigma^1 \right)
  +
M_\phi(z_2)  
 \phi\left(z_1;0,\sigma^2 \right)
 \prod_{i=3}^d   \phi\left(z_i;0,\sigma^1 \right)  \\
 &+\cdots+
 M_\phi(z_d)  
 \prod_{i=1}^{d-1}  \phi\left(z_i;0,\sigma^2 \right)
  \bigg\}.
\end{align*}

\end{proof}

\pagebreak

\section{Supplement. Auxiliary Results and Details}
\label{sec:auxiliary_res_det}

\subsection{Newton-Kantorovich algorithm for DP solution}
\label{sec:NK}

For a given value of the model parameter $\psi$, the \textit{Emax} function $Q_\psi$ is the fixed point of the operator $T_\psi(\cdot)$ defined by 
{\small
\begin{align*}
&T_{\psi}( \textbf{Q}) [x] = \sum_{k=1}^m \omega_k \sigma_k \bigg[A_{kx} + E1( e^{-a_{kx}} ) \bigg] \\
&=\sum_{k=1}^m \omega_k \sigma_k \bigg\{ \\
& \log \sum_{j=1}^J \exp \bigg(\frac{ u(x,j,;\psi) +\beta \sum_{y=1}^K G_{xy}^j \textbf{Q}(y) + \mu_{jk}    }{  \sigma_k  }\bigg) +\\
& E1\bigg(    \exp\bigg[ -\frac{  u(0,j,;\psi) + \beta \sum_{y=1}^K G_{xy}^j \textbf{Q}(y) }{\sigma_k  } - \gamma \\
&\;\;\;\;\;\;\;\;\;\;\;\;\;\;\;-\log\sum_{j=1}^J \exp\bigg\{  \frac{ u(x,j,;\psi) + \beta \sum_{y=1}^K G_{xy}^j \textbf{Q}(y) + \mu_{jk}    }{  \sigma_k  }  \bigg\} \bigg]   \bigg) \bigg\}.
\end{align*}
}

To find the fixed point of $T_\psi$, following \cite{Rust:87},  we use the Newton-Kantorovich algorithm, which is essentially a Newton method for solving the nonlinear system of equations, because it is more efficient than the iterations on $T_\psi$. 

\begin{lemma}\label{lemma:NK}

The Newton-Kantorovich algorithm has the update rule 
\begin{align}
Q^{(n+1)} = Q^{(n)} - \bigg[ I - T'_\psi( Q^{(n)}  ) \bigg]^{-1} \bigg[ I - T_\psi   \bigg] (Q^{n}), \label{eq: newton}
\end{align}
where 
\[
T'_\psi( Q  ) 
=\beta
\begin{pmatrix} 
\sum_{d=0}^J G_{11}^d P(d|x=1) & \hdots & \sum_{d=0}^J G_{1K}^d P(d|x=1) \\
 \vdots &  & \vdots \\ 
\sum_{d=0}^J G_{K1}^d P(d|x=K) & \hdots & \sum_{d=0}^J G_{KK}^d P(d|x=K)
\end{pmatrix}.
\]
\end{lemma}

\begin{proof}[proof of Lemma \ref{lemma:NK}]

Note that 
\[
T'_\psi( Q  ) 
= 
\begin{pmatrix} 
\frac{\partial}{\partial Q(1)}T_\psi (Q)(1) & \hdots & \frac{\partial}{\partial Q(K)}T_\psi (Q)(1) \\
 \vdots &  & \vdots \\ 
\frac{\partial}{\partial Q(1)}T_\psi (Q)(K) & \hdots & \frac{\partial}{\partial Q(K)}T_\psi (Q)(K)
\end{pmatrix}.
\]
Recall that $T_\psi(Q)(x)=\sum_{k=1}^m \omega_k \sigma_k \left[ A_{kx} +E1(e^{-a_{kx}})  \right]$.
Denoting $f( a_{kx}  ) = E1(e^{-a_{kx}})$, we have 
\begin{align*}
\frac{\partial}{\partial Q(1)}T_\psi (Q)(x)&=\sum_{k=1}^m \omega_k \sigma_k \left[ \frac{\partial}{\partial Q(1)} A_{kx} + \frac{\partial}{\partial Q(1)} f(a_{kx})  \right], \text{ for } x=1,\ldots,K.
\end{align*}

Recall that 
\begin{align*}
A_{kx} &= \log \sum_{j=1}^J \exp\left[ \frac{v(x,j)+\mu_{jk}}{\sigma_k} \right] 
 = \log \sum_{j=1}^J \exp\left[ \frac{u(x,j)+\beta\sum_{y=1}^K G_{xy}^jQ(y)+\mu_{jk}}{\sigma_k} \right], \\
a_{kx} &=  \frac{v_0}{\sigma_k} + \gamma - A_k 
=\frac{1}{\sigma_k} \left[ u(x,0) + \beta\sum_{y=1}^K G_{xy}^jQ(y)  \right] + \gamma -A_{kx}. 
\end{align*}
Hence, for example, 
\begin{align*}
\frac{\partial A_{k1} }{\partial Q(1)} 
&= \frac{1}{  \sum_{j=1}^J \exp\left[ \frac{v(1,j)+\mu_{jk}}{\sigma_k} \right]  } \sum_{d=1}^J\frac{\partial  \exp\left[ \frac{v(1,d)+\mu_{dk}}{\sigma_k} \right] }{\partial  \left[ \frac{v(1,d)+\mu_{dk}}{\sigma_k} \right] }\frac{\partial \left[ \frac{v(1,d)+\mu_{dk}}{\sigma_k} \right]  }{\partial Q(1)}\\
&= \frac{1}{  \sum_{j=1}^J \exp\left[ \frac{v(1,j)+\mu_{jk}}{\sigma_k} \right]  } \sum_{d=1}^J \exp\left[ \frac{v(1,d)+\mu_{dk}}{\sigma_k} \right] \frac{\beta}{ \sigma_k  } G_{11}^d\\
&=\frac{\beta}{ \sigma_k  } \sum_{d=1}^J \frac{  \exp\left[ \frac{v(1,d)+\mu_{dk}}{\sigma_k} \right]    }{  \sum_{j=1}^J \exp\left[ \frac{v(1,j)+\mu_{jk}}{\sigma_k} \right]  } G_{11}^d \\
&=\frac{\beta}{ \sigma_k  } \sum_{d=1}^J \exp\left[ \frac{v(1,d)+\mu_{dk}}{\sigma_k} - A_{k1}\right]   G_{11}^d. 
\end{align*}
Also,
\begin{align*}
\frac{\partial f(a_{k1}) }{\partial Q(1)} &= 
\frac{\partial g(e^{-a_{k1}})   }{\partial  e^{-a_{k1}}   } \frac{\partial  e^{-a_{k1}}  }{\partial  (-a_{k1})   }\frac{\partial  (-a_{k1})  }{\partial   a_{k1}  }\frac{\partial  a_{k1}  }{\partial  Q(1)   }=g'(e^{-a_{k1}}) e^{-a_{k1}}  (-1) \frac{\partial  a_{k1}  }{\partial  Q(1)   },
\end{align*}
where we let $g(y) = E1(y)$. We know that $g'(y) = -\frac{e^{-y}}{y}$. Note that $\frac{\partial  a_{k1}  }{\partial  Q(1)   }= \frac{\beta}{ \sigma_k} G_{11}^1 - \frac{\partial  A_{k1}  }{\partial  Q(1)   }$.
Hence 
\begin{align*}
\frac{\partial f(a_{k1}) }{\partial Q(1)} &= -\exp\left[ - e^{-a_{k1}}\right] e^{a_{k1}} e^{-a_{k1}} (-1) \frac{\partial  a_{k1}  }{\partial  Q(1)   }
=\exp\left[ - e^{-a_{k1}}\right]  \left[  \frac{\beta}{ \sigma_k} G_{11}^1 - \frac{\partial  A_{k1}  }{\partial  Q(1)   }  \right].
\end{align*}
Now, 
\begin{align*}
\frac{\partial}{\partial Q(1)}T_\psi (Q)(1)&=
\sum_{k=1}^m \omega_k \sigma_k \left[ \frac{\partial  A_{k1} }{\partial Q(1)}  + \exp\left[ - e^{-a_{k1}}\right]  \left[  \frac{\beta}{ \sigma_k} G_{11}^1 - \frac{\partial  A_{k1}  }{\partial  Q(1)   }  \right] \right]\\
&=
\sum_{k=1}^m \omega_k \sigma_k \left[ \frac{\partial  A_{k1} }{\partial Q(1)}  \bigg(1-  \exp\left[ - e^{-a_{k1}}\right]   \bigg)
+  \exp\left[ - e^{-a_{k1}}\right]   \frac{\beta}{ \sigma_k} G_{11}^1\right]. 
\end{align*}
Note that 
\begin{align*}
&\sum_{k=1}^m \omega_k \sigma_k \frac{\partial  A_{k1} }{\partial Q(1)}  \bigg(1-  \exp\left[ - e^{-a_{k1}}\right]   \bigg)\\
&=\sum_{k=1}^m \omega_k \sigma_k 
\underbrace{
 \frac{\beta}{ \sigma_k  } \sum_{d=1}^J \exp\left[ \frac{v(1,d)+\mu_{dk}}{\sigma_k} - A_{k1}\right]   G_{11}^d
}_\text{$\frac{\partial A_{k1} }{\partial Q(1)} $} 
\bigg(1-  \exp\left[ - e^{-a_{k1}}\right]   \bigg) \\
&=\beta \sum_{d=1}^J G_{11}^d \sum_{k=1}^m \omega_k
\exp\left[ \frac{v(1,d)+\mu_{dk}}{\sigma_k} - A_{k1}\right] \bigg(1-  \exp\left[ - e^{-a_{k1}}\right]   \bigg) =\beta \sum_{d=1}^J G_{11}^d P(d|x=1)\\
&\sum_{k=1}^m \omega_k \sigma_k \exp\left[ - e^{-a_{k1}}\right]  \frac{\beta}{ \sigma_k} G_{11}^1 = \beta \sum_{k=1}^m \omega_k  \exp\left[ - e^{-a_{k1}}\right]  G_{11}^1 =\beta G_{11}^1 P(0|x=1).
\end{align*}
Hence,
\begin{align*}
\frac{\partial}{\partial Q(1)}T_\psi (Q)(1)&=\beta \sum_{d=0}^J G_{11}^d P(d|x=1).
\end{align*}

We can similarly compute other elements. 
\end{proof}


\subsection{Properties of univariate extreme value distributions}\label{EV}
Let $Z \sim \phi(\cdot)$. Then it has zero mean and variance $\frac{\pi^2}{6}$. 
Its density is
$e^{-z-\gamma -e^{-z-\gamma}}$
and 
its cdf is $e^{ -e^{-z-\gamma}  }$, 
where  $\gamma$ is the Euler-Mascheroni constant.

Lemma \ref{lemma:median} and Lemma \ref{lemma:truncated_integral} show how to compute median and truncated integrals of a random variable that follows a mixture of extreme value distributions, which 
is helpful for imposing the location and scale normalizations after a MCMC run.

\begin{lemma}\label{lemma:median}
Median
\begin{enumerate}
	\item Let $X \sim \phi(\cdot;\mu,\sigma)$. Then its median is $\mu - \sigma \log\log 2 -\sigma\gamma$.
	\item Let $X \sim \sum_{k=1}^m \omega_k\phi(\cdot; \mu_k,\sigma_k)$. Then there is no closed form for its median. It has to be solved for via a root-finding algorithm.		
\end{enumerate}
\end{lemma}

\begin{proof}
Let $X \sim \phi(\cdot; \mu,\sigma)$. By definition of median, we want to find $M$ such that $0.5=Pr(X<M)$. Letting $t(x) = e^{-\left(\frac{x-\mu}{\sigma} \right) } e^{-\gamma}$, this is equivalent to  
\begin{align*}
0.5&=\exp\left[ - t(M) \right]  \iff t(M) = \log 2  \iff 	e^{-\left(\frac{M-\mu}{\sigma}\right) } e^{-\gamma} = \log 2 \\
&\iff -\left(\frac{M-\mu}{\sigma} \right) -\gamma = \log\log 2 \iff M = \mu - \sigma \log\log 2 -\sigma\gamma
\end{align*}
\end{proof}


\begin{lemma}\label{lemma:truncated_integral}
Truncated Integral
\begin{enumerate}
	\item Let $X \sim \phi(\cdot; \mu,\sigma)$. Then 
		\begin{align*}
		\int_{M}^{\infty} x \phi(x; \mu,\sigma)dx = \mu - M \exp\left[- e^{-b} \right] + \sigma E1(e^{-b})
		\end{align*}
		where $b= \frac{M-\mu}{\sigma}+\gamma$
	\item Let $X \sim p(\cdot \vert \psi, m)=\sum_{k=1}^m \omega_k\phi(\cdot; \mu_k,\sigma_k)$. Then 	 
		\begin{align*}
		\int_{M}^{\infty} x p(x \vert \psi, m)dx = \sum_{k=1}^m \omega_k   \bigg\{ \mu_k - M \exp\left[- e^{-b_k} \right] + \sigma_k E1(e^{-b_k}) \bigg\}
		\end{align*}
		where $b_k= \frac{M-\mu_k}{\sigma_k}+\gamma$
\end{enumerate}
\end{lemma}

\begin{proof}
Let $X \sim \phi(\cdot; \mu,\sigma)$. 
\begin{align*}
\int_{M}^{\infty} x \phi(x; \mu,\sigma)dx  &= \int_M^\infty x \frac{1}{\sigma}e^{-\left(\frac{x-\mu}{\sigma} \right)} e^{-\gamma} \exp\left[ -e^{-\left(\frac{x-\mu}{\sigma} \right)} e^{-\gamma}   \right] dx \\
&=\int_b^\infty \big[ \sigma z +( \mu - \sigma \gamma ) \big] \frac{1}{\sigma} e^{-z} \exp\left[ -e^{-z} \right] \sigma dz \\
&= \sigma \underbrace{  \int_b^\infty   z  e^{-z} \exp\left[ -e^{-z} \right] dz}_{\text{I}}
 + ( \mu - \sigma \gamma )   \underbrace{ \int_b^\infty    e^{-z} \exp\left[ -e^{-z} \right] dz}_{\text{II}}
\end{align*}
where we let $z = \left(\frac{x-\mu}{\sigma} \right) + \gamma$ which means that $x=\sigma z +( \mu - \sigma \gamma ), dx = \sigma dz$. $x=\infty \implies z = \infty, x = M \implies z =  \frac{M-\mu}{\sigma}+\gamma = b$. For computing the first integral, let $y = e^{-z}$ which means that $dz = -\frac{1}{y}dy, z=-\log y$. We have 
\begin{align*}
I = \int_{e^{-b}}^0 \left( -\log y \right) y e^{-y} (-\frac{1}{y}dy) = - \int_0^{e^{-b}} \log y e^{-y}dy
\end{align*}
Recall that $\gamma = -\int_0^\infty \log y e^{-y} dy = -\left( \int_0^{e^{-b}} \log y e^{-y} dy + \int^\infty_{e^{-b}} \log y e^{-y} dy   \right) $. Hence, $I = \gamma + \int^0_{e^{-b}} \log y e^{-y}dy$. By integration by parts, letting $u=\log y, dv=e^{-y} dy \implies du \frac{1}{y} dy, v=-e^{-y}$, 
\begin{align*}
I &= \gamma + \int^0_{e^{-b}} \log y e^{-y}dy = \gamma + \int udv -\int vdu \\
&= \gamma -\big[ \log y e^{-y} \big]^{\infty}_{e^{-b}} + \int^{\infty}_{e^{-b}} \frac{1}{y} e^{-y} dy =\gamma - \big[ 0 - (-b) \exp\left(-e^{-b}\right)\big] + E1\left( e^{-b}\right)\\
&=\gamma -b \exp\left(-e^{-b}\right) + E1\left( e^{-b}\right)
\end{align*}
where by definition $E1(z) = \int_z^\infty \frac{1}{t} e^{-t}dt$. For the second integral, letting $y = e^{-z}$ which means that $dz = -\frac{1}{y}dy$, we have  
\begin{align*}
II &=\int_b^\infty    e^{-z} \exp\left[ -e^{-z} \right] dz \\
& =   \int^0_{e^{-b}} y e^{-y} (-\frac{1}{y}dy)= \int_0^{e^{-b}} e^{-y} dy = -\big( \exp\left( -e^{-b}\right)-1 \big) \\
& = 1 - \exp\left( -e^{-b}\right)
\end{align*}
Finally,
\begin{align*}
\int_{M}^{\infty} x dF(x) &= \sigma I +( \mu - \sigma \gamma ) II =  \mu - M \exp\left[- e^{-b} \right] + \sigma E1(e^{-b})
\end{align*}

\end{proof}

\subsection{Derivatives of the Log-likelihood}
\label{sec:derivatives}

We utilize a built-in HMC package in Matlab. The package requires the parameters to be unbounded. 
Recall that the vector of parameters is $\psi_{1m} = \big(\theta,\sigma,\gamma_k,\mu_{k}, \tilde{\sigma}_k, \; k=1,\ldots,m\big)$, where the weights and the scales are not unbounded. 
For a given fixed $m$, we denote the transformed parameters by $\chi = (\theta, \log \sigma, \alpha_k,\mu_{k}, \log \tilde{\sigma}_k, \; k=1,\ldots,m )$,
where 
\[
\omega_k = \frac{\gamma_k}{\sum_{l=1}^m \gamma_l}=\frac{e^{\alpha_k}}{1+\sum_{\ell=1}^{m-1} e^{\alpha_\ell}}, k=1,\ldots,m 
\text{ and } \alpha_m = 0 
\]
so that all the components of $\chi$ are unbounded.

Let $L(\chi) = \log p(D^n|\chi,m) = \sum_{d,x} n_{dx} \times \log p(d|x;\chi,m)$ 
denote the log-likelihood as a function of $\chi$, where  $n_{dx}$ is the number of decision makers in the data set that chose $d$ at the observed state $x$.
Below, we present the derivatives of $L(\chi)$ with respect to $\chi$ for $J=1$. The computation for $J>1$ can be done similarly. 

\begin{lemma}\label{lemma:derivatives} 

Derivatives of $L(\chi)$ with respect to $\chi$ for $J=1$.

(1) For $k=1,\ldots,m-1$, 
	\begin{align*}
	\frac{\partial L(\chi)}{\partial \alpha_k} = \sum_{k=1}^{m-1} \frac{\partial L(\chi)}{\partial \omega_\ell} \frac{\partial \omega_\ell}{\partial \alpha_k} 
	\mbox{, where}
	\end{align*}
	\begin{align*}
	\frac{\partial L(\chi)}{\partial \omega_\ell}=&
			\sum_x n_{0x} \frac{p(0|x,\ell) - p(0|x,m) +\sum_k \omega_k \frac{\partial p(0|x,k)}{\partial \omega_\ell}}{\sum_{k'} \omega_{k'} p(0|x,k')} 
			\\&+ \sum_x n_{1x} \frac{p(1|x,\ell) - p(1|x,m)+\sum_k \omega_k \frac{\partial p(1|x,k)}{\partial \omega_\ell} }{\sum_{k'} \omega_{k'} p(1|x,k')},
			\end{align*}
			\[
			\frac{\partial p(0|x,k)}{\partial \omega_\ell} =
			\exp\big[-e^{-a_{xk}}-a_{xk} \big]\frac{\beta}{\sigma_k} \sum_{y=1}^K (G_{xy}^0-G_{xy}^1) \frac{\partial Q(y)}{\partial \omega_\ell},
			\]
			\[
			\frac{\partial \omega_\ell}{\partial \alpha_k} = \frac{e^{\alpha_\ell}}{(1+\sum_{s=1}^{m-1}e^{\alpha_s})^2}\big[ (1+\sum_{s\ne k}^{m-1} 	
			e^{\alpha_s}) 1(k=\ell )  -e^{\alpha_k} 1(k \ne \ell)  \big].
			\]

(2) For $k=1,\ldots,m$,
	\begin{align*}
	 \frac{\partial L(\chi)}{\partial \mu_k}=\sum_x n_{0x} \frac{\sum_\ell \omega_\ell \frac{\partial p(0|x,\ell)}{\partial \mu_k}}{\sum_{k'} \omega_{k'} p(0|x,k')} 
	 +\sum_x n_{1x} \frac{\sum_\ell \omega_\ell \frac{\partial p(1|x,\ell)}{\partial \mu_k}}{\sum_{k'} \omega_{k'} p(1|x,k')} 
	\mbox{, where}
	 \end{align*}
\[
\frac{\partial p(0|x,\ell)}{\partial \mu_k} =
				  \exp\big[-e^{-a_{x\ell}}-a_{x\ell} \big]\frac{\partial a_{x\ell}}{\partial \mu_k},
\]
\[
\frac{\partial a_{x\ell}}{\partial \mu_k} =\frac{\beta}{\sigma_\ell}\sum_{y=1}^K (G_{xy}^0-G_{xy}^1) \frac{\partial Q(y)}{\partial \mu_k}-
\frac{1}{\sigma_\ell}1(k=\ell).
\]	

(3) For $s_k=\log \tilde{\sigma}_k$, $k=1,\ldots,m$,
	\begin{align*}
	 \frac{\partial L(\chi)}{\partial s_k}=\frac{\partial L(\chi)}{\partial \sigma_k} \sigma_k
	\mbox{, where}
	 \end{align*}		 
\[
\frac{\partial L(\chi)}{\partial \sigma_k}=
				  \sum_x n_{0x} \frac{\sum_\ell \omega_\ell \frac{\partial p(0|x,\ell)}{\partial \sigma_k}}{\sum_{k'} \omega_{k'} p(0|x,k')} 
				+\sum_x n_{1x} \frac{\sum_\ell \omega_\ell \frac{\partial p(1|x,\ell)}{\partial \sigma_k}}{\sum_{k'} \omega_{k'} p(1|x,k')}, 
				\]
			\[
			\frac{\partial p(0|x,\ell)}{\partial \sigma_k} =
				  \exp\big[-e^{-a_{x\ell}}-a_{x\ell} \big]\frac{\partial a_{x\ell}}{\partial \sigma_k},
					\]
			\[
			\frac{\partial a_{x\ell}}{\partial \sigma_k} =	\frac{\beta}{\sigma_\ell} \sum_{y=1}^K (G_{xy}^0-G_{xy}^1) 
			\frac{\partial Q(y)}{\partial 					\mu_k}+\frac{1}{\sigma_\ell}[\gamma - a_{xk}]1(k=\ell).
			\]
		
(4) With respect to $s=\log(\sigma)$
	\begin{align*}
	 \frac{\partial L(\chi)}{\partial s}=\sum_{\ell=1}^m \frac{\partial L(\chi)}{\partial s_\ell}. 
	 \end{align*}		

(5) The derivatives of the \textit{Emax} function can be computed by the implicit function theorem: 
\begin{align*}
\frac{\partial Q(x)}{\partial \omega_k}=- \frac{\frac{\partial f_x}{\partial \omega_k}}{\frac{\partial f_x}{\partial Q(x)} },
\quad 
\frac{\partial Q(x)}{\partial \mu_k}=- \frac{\frac{\partial f_x}{\partial \mu_k}}{\frac{\partial f_x}{\partial Q(x)} },
\quad
\frac{\partial Q(x)}{\partial \sigma_k}=- \frac{\frac{\partial f_x}{\partial \sigma_k}}{\frac{\partial f_x}{\partial Q(x)} },
\end{align*}	
where $f_x = Q(x) - \sum_{k=1}^m \omega_k \sigma_k \big[ A_{xk} +E1(e^{-a_{xk}}) \big]$	and 
\[
\frac{\partial f_x}{\partial Q(x)}= 1-\beta G_{xx}^1 -\beta(G_{xx}^0 -G_{xx}^1) \sum_{k=1}^m\omega_k \exp\big[ -e^{-a_{xk}} \big],
\]
\[
\frac{\partial f_x}{\partial \omega_k} = -\sigma_k \big[A_{xk} +E1(e^{-a_{xk}}) \big] +  \sigma_m \big[A_{xm} +E1(e^{-a_{xm}}) \big], 
\]
\[ 
\frac{\partial f_x}{\partial \mu_k} = -\omega_k \big[1 - \exp(-e^{-a_xk}) \big],
\]
\[
 \frac{\partial f_x}{\partial \sigma_k} =-\omega_k \big[E1(e^{-a_{xk}}) + \exp(-e^{-a_xk}) (\gamma - a_{xk})\big].
\]
\end{lemma}

\begin{proof}
For (1), 
\begin{align*}
	\frac{\partial L(\chi)}{\partial \alpha_k} &= 
		\sum_{k=1}^{m-1} \frac{\partial L(\chi)}{\partial \omega_\ell} \frac{\partial \omega_\ell}{\partial \alpha_k}, \\
	\frac{\partial L(\chi)}{\partial \omega_\ell}&=
		\sum_x n_{0x} \frac{\frac{\partial }{\partial \omega_\ell}  \sum_k \omega_k p(0|x,k) }{\sum_{k'} \omega_{k'} p(0|x,k')} 
		+ \sum_x n_{1x} \frac{\frac{\partial }{\partial \omega_\ell}  \sum_k \omega_k p(1|x,k)  }{\sum_{k'} \omega_{k'} p(1|x,k')},\\	
 \sum_k^m \omega_k p(d|x,k) &= 
	 	\sum_k^{m-1} \omega_k p(d|x,k) +(1- \sum_k^{m-1} \omega_k )p(d|x,m), 	\\				
  \frac{\partial}{\partial \omega_\ell} \sum_k^m \omega_k p(d|x,k)  	 
	  	&=p(d|x,\ell) 	+\sum_k^{m-1} \omega_k \frac{\partial p(d|x,k)}{\partial \omega_\ell}
	  	-p(d|x,m) +(1- \sum_k^{m-1} \omega_k )\frac{\partial p(d|x,m)}{\partial \omega_\ell} 	\\
	  &=p(d|x,\ell)  -p(d|x,m)+\sum_k^{m} \omega_k \frac{\partial p(d|x,k)}{\partial \omega_\ell}	  \\	   	  
	 \implies \frac{\partial L(\chi)}{\partial \omega_\ell}&=	  	  
			\sum_x n_{0x} \frac{p(0|x,\ell) - p(0|x,m) +\sum_k \omega_k \frac{\partial p(0|x,k)}{\partial \omega_\ell}}{\sum_{k'} \omega_{k'} p(0|x,k')} \\
			&+ \sum_x n_{1x} \frac{p(1|x,\ell) - p(1|x,m)+\sum_k \omega_k \frac{\partial p(1|x,k)}{\partial \omega_\ell} }{\sum_{k'} \omega_{k'} p(1|x,k')},	 \\	  
	  \frac{\partial p(0|x,k)}{\partial \omega_\ell} &=
	  	\exp\big[-e^{-a_{xk}}-a_{xk} \big]\frac{\partial a_{xk}}{\partial \omega_\ell}.	  
\end{align*}
For $J=1$, we have, 
\begin{align*}		
	a_{x\ell}&=\frac{1}{\sigma_\ell}\big[v(x,0)-v(x,1)-\mu_\ell \big] +\gamma \\		
	&=\frac{1}{\sigma_\ell}\left[u(x,0)-u(x,1)+\beta\sum_{y=1}^K (G_{xy}^0-G_{xy}^1) Q(y) -\mu_\ell \right] +\gamma.		 	 			
\end{align*}
Hence,
\begin{align*}
 \frac{\partial p(0|x,k)}{\partial \omega_\ell} &=			
	  \exp\big[-e^{-a_{xk}}-a_{xk} \big]\frac{\beta}{\sigma_k} \sum_{y=1}^K (G_{xy}^0-G_{xy}^1) \frac{\partial Q(y)}{\partial \omega_\ell}. 
\end{align*}
Note that 
\begin{align*}	  
\frac{\partial p(1|x,k)}{\partial \omega_\ell} &= 
	 	1- \frac{\partial p(0|x,k)}{\partial \omega_\ell}. 
\end{align*}
Moreover, recall that $\omega_\ell = \frac{e^{\alpha_\ell}}{1+\sum_{s=1}^{m-1}e^{\alpha_s}}$.\\
If $k = \ell$, then $\frac{\partial \omega_\ell}{\partial \alpha_k} 
= \frac{e^{\alpha_\ell} (1+\sum_{s=1}^{m-1}e^{\alpha_s}) -e^{\alpha_\ell} e^{\alpha_\ell}      }{(1+\sum_{s=1}^{m-1}e^{\alpha_s})^2} 
= \frac{ e^{\alpha_\ell} (1+\sum_{s\ne \ell}^{m-1} e^{\alpha_s}   )}{(1+\sum_{s=1}^{m-1}e^{\alpha_s})^2}$. \\
If $k \ne \ell$, then $\frac{\partial \omega_\ell}{\partial \alpha_k} = \frac{-e^{\alpha_\ell} e^{\alpha_k}  }{(1+\sum_{s=1}^{m-1}e^{\alpha_s})^2} $.

For (2),
\begin{align*}
	\frac{\partial L(\chi)}{\partial \mu_k}&=
		\sum_x n_{0x} \frac{\frac{\partial }{\partial \mu_k}  \sum_k \omega_k p(0|x,k) }{\sum_{k'} \omega_{k'} p(0|x,k')} 
		+ \sum_x n_{1x} \frac{\frac{\partial }{\partial \mu_k}  \sum_k \omega_k p(1|x,k)  }{\sum_{k'} \omega_{k'} p(1|x,k')}\\
	&=\sum_x n_{0x} \frac{\sum_\ell \omega_\ell \frac{\partial p(0|x,\ell)}{\partial \mu_k}}{\sum_{k'} \omega_{k'} p(0|x,k')} 
	 	+\sum_x n_{1x} \frac{\sum_\ell \omega_\ell \frac{\partial p(1|x,\ell)}{\partial \mu_k}}{\sum_{k'} \omega_{k'} p(1|x,k')}, 	\\	
 	\frac{\partial p(0|x,\ell)}{\partial \mu_k} &=
	  	\exp\big[-e^{-a_{x\ell}}-a_{x\ell} \big]\frac{\partial a_{x\ell}}{\partial \mu_k},\\
    \frac{\partial a_{x\ell}}{\partial \mu_k}&=
    \begin{cases}
      \frac{\beta}{\sigma_\ell}\sum_{y=1}^K (G_{xy}^0-G_{xy}^1) \frac{\partial Q(y)}{\partial \mu_k}-\frac{1}{\sigma_\ell}, & \text{if}\ k=\ell \\
      \frac{\beta}{\sigma_\ell}\sum_{y=1}^K (G_{xy}^0-G_{xy}^1) \frac{\partial Q(y)}{\partial \mu_k}, & \text{otherwise.}
    \end{cases}
  \end{align*}

(3) Recall that $\sigma_k = \sigma \tilde{\sigma}_k$. $\tilde{\sigma}_k = e^{s_k}$, so $\frac{\partial \tilde{\sigma}_k}{\partial s_k}=\tilde{\sigma}_k$. We have
\begin{align*}
	\frac{\partial L(\chi)}{\partial s_k}&=\frac{\partial L(\chi)}{\partial \sigma_k}\frac{\partial \sigma \tilde{\sigma}_k}{\partial \tilde{\sigma}_k	} 
	\frac{\partial \tilde{\sigma}_k}{\partial s_k}=	\frac{\partial L(\chi)}{\partial \sigma_k}\sigma \tilde{\sigma}_k=	\frac{\partial L(\chi)}{\partial \sigma_k}\sigma_k,	\\
	\frac{\partial L(\chi)}{\partial \sigma_k}&=
		\sum_x n_{0x} \frac{\frac{\partial }{\partial \sigma_k}  \sum_k \omega_k p(0|x,k) }{\sum_{k'} \omega_{k'} p(0|x,k')} 
		+ \sum_x n_{1x} \frac{\frac{\partial }{\partial \sigma_k}  \sum_k \omega_k p(1|x,k)  }{\sum_{k'} \omega_{k'} p(1|x,k')}\\
	&=\sum_x n_{0x} \frac{\sum_\ell \omega_\ell \frac{\partial p(0|x,\ell)}{\partial \sigma_k}}{\sum_{k'} \omega_{k'} p(0|x,k')} 
	 	+\sum_x n_{1x} \frac{\sum_\ell \omega_\ell \frac{\partial p(1|x,\ell)}{\partial \sigma_k}}{\sum_{k'} \omega_{k'} p(1|x,k')}, 	\\
 	\frac{\partial p(0|x,\ell)}{\partial \sigma_k} &=
	  	\exp\big[-e^{-a_{x\ell}}-a_{x\ell} \big]\frac{\partial a_{x\ell}}{\partial \sigma_k},\\
    \frac{\partial a_{x\ell}}{\partial \sigma_k}&=
    \begin{cases}
      \frac{1}{\sigma_\ell}\left[  \gamma - a_{xk} + \beta\sum_{y=1}^K (G_{xy}^0-G_{xy}^1) \frac{\partial Q(y)}{\partial \sigma_k}  \right], & \text{if}\ k=\ell \\
      \frac{\beta}{\sigma_\ell} \sum_{y=1}^K (G_{xy}^0-G_{xy}^1) \frac{\partial Q(y)}{\partial \sigma_k} , & \text{otherwise.}
    \end{cases}
\end{align*}  

(4) Note that $\sigma = e^{s}$, so $\frac{\partial \sigma}{\partial s}=\sigma$
\begin{align*}
\frac{\partial L(\chi)}{\partial s} =\sum_{k=1}^m \frac{ \partial L(\chi) }{\partial \sigma_k} \frac{\partial \sigma \tilde{\sigma}_k}{\partial \sigma} \frac{\partial \sigma}{\partial s}=\sum_{k=1}^m \frac{ \partial L(\chi) }{\partial \sigma_k}\sigma_k =\sum_{k=1}^m \frac{ \partial L(\chi) }{\partial s_k}.
\end{align*}

(5) Recall that
\[Q(x) = \sum_{k=1}^m \omega_k \sigma_k \big[ A_{xk} +E1(e^{-a_{xk}}) \big],\]
\[f_x = Q(x)-\sum_{k=1}^m \omega_k \sigma_k \big[ A_{xk} +E1(e^{-a_{xk}}) \big],\]
\[a_{xk} =\frac{1}{\sigma_k}\left[u(x,0)-u(x,1)+\beta\sum_{y=1}^K (G_{xy}^0-G_{xy}^1) Q(y) -\mu_k \right] +\gamma, \]
\[A_{xk} =\frac{1}{\sigma_k}\left[u(x,1) + \beta \sum_{y=1}^K G_{xy}^1 Q(y) +\mu_k \right].\]
Hence, we have  
\begin{align*}
\frac{\partial f_x}{\partial Q(x)} &=1-\sum_{k=1}^m \omega_k \sigma_k \left[ \frac{\partial A_{xk}}{\partial Q(x)} +\frac{\partial E1(e^{-a_{xk}})}{\partial Q(x)} \right],\\
 \frac{\partial A_{xk}}{\partial Q(x)} &=\frac{\beta}{\sigma_k} G_{xx}^1,\\
 \frac{\partial E1(e^{-a_{xk}})}{\partial Q(x)}&=\exp\big[ -e^{-a_{xk}}\big] \frac{\partial a_{xk}}{\partial Q(x)}=\exp\big[ -e^{-a_{xk}}\big] \frac{\beta}{\sigma_k}(G_{xx}^0 - G_{xx}^1).\\
 \implies \frac{\partial f_x}{\partial Q(x)} &= 1-\beta G_{xx}^2 -\beta(G_{xx}^1 -G_{xx}^2) \sum_{k=1}^m\omega_k \exp\big[ -e^{-a_{xk}} \big].
  \end{align*}
It is easy to show that 
$\frac{\partial f_x}{\partial \omega_k} = -\sigma_k \big[A_{xk} +E1(e^{-a_{xk}}) \big] +  \sigma_m \big[A_{xm} +E1(e^{-a_{xm}}) \big]$ \\and  
$\frac{\partial f_x}{\partial \mu_k} =-\omega_k\sigma_k\left[ \frac{\partial A_{xk}}{\partial \mu_k} +\frac{\partial E1(e^{-a_{xk}})}{\partial \mu_k}    \right]=-\omega_k \big[1 - \exp(-e^{-a_xk}) \big]$.

Finally, 
\begin{align*}
\frac{\partial f_x}{\partial \sigma_k} &=-\omega_k\big[ A_{xk} + E1(e^{-a_{xk}})   \big] -\omega_k\sigma_k\left[ \frac{\partial A_{xk}}{\partial \sigma_k} +\frac{\partial E1(e^{-a_{xk}})}{\partial \sigma_k}    \right],\\
\frac{\partial A_{xk}}{\partial \sigma_k} &=\frac{1}{\sigma_k^2}\big[v(x,1) +\mu_k \big]= \frac{1}{\sigma_k}A_{xk},\\
 \frac{\partial E1(e^{-a_{xk}})}{\partial Q(x)}&=\exp\big[ -e^{-a_{xk}}\big] \frac{\partial a_{xk}}{\partial \sigma_k}=\exp\left[ -e^{-a_{xk}}\right]\left[ -\frac{1}{\sigma_k^2}(v(x,0)-v(x,1) -\mu_k)  \right] \\
 &=\exp\big[ -e^{-a_{xk}}\big] \frac{1}{\sigma_k}\big[\gamma- a_{xk} \big], \\
 \implies \frac{\partial f_x}{\partial \sigma_k} &=-\omega_k \big[E1(e^{-a_{xk}}) + \exp(-e^{-a_xk}) (\gamma - a_{xk})\big].
 \end{align*}

\end{proof}

\subsection{The Prior for the Transformed Mixing Weights}
\label{prior_weight}

We want to define a prior on $\big( \alpha_1,\dots, \alpha_{m-1} \big)$ which implies the prior $\big( \omega_1,\ldots,\omega_m\big) \sim Dir(a_1,\ldots,a_m)$ or equivalently $\gamma_\ell \sim Gamma(a_\ell,1)$ for $\ell=1,\ldots,m$. Note that the inverse map $g(\alpha)= \omega$ is defined by $\omega_\ell = \frac{e^{\alpha_\ell}}{1+\sum_{s=1}^{m-1}e^{\alpha_s}}, \text{ for } \ell = 1,\ldots, m-1$. Denote the Jacobian $A(\alpha) = \frac{dg(\alpha)}{d\alpha}$. By the change of variable formula, we have $f_\alpha(\alpha) 
= f_\omega(g(\alpha)) det \frac{dg(\alpha)}{d\alpha}
= f_\omega(g(\alpha)) det A(\alpha)
$. 

In order to have $f_\omega(\omega)=Dir(\omega_1,\ldots,\omega_m; \bar{\omega}_1,\ldots,\bar{\omega}_m)$, we can use the above expression for the density $f_\alpha(\alpha)$ to determine the desired prior for $\alpha_1,\ldots,\alpha_m$.

In addition, for using HMC, we need the first order derivative $\frac{\partial f_\alpha(\alpha) }{\partial \alpha}$. Note that in general, $\frac{\partial}{\partial \alpha_j} det A(\alpha) = det A(\alpha) tr\bigg(A(\alpha)^{-1}\frac{\partial}{\partial \alpha_j} A(\alpha)  \bigg)$. Hence, $\frac{\partial f_\alpha(\alpha) }{\partial \alpha_j}=\bigg[ \frac{\partial}{\partial \alpha_j} f_\omega(g(\alpha)) \bigg]det A(\alpha) + f_\omega(g(\alpha)) \bigg[ det A(\alpha) tr\bigg(A(\alpha)^{-1}\frac{\partial}{\partial \alpha_j} A(\alpha)  \bigg)\bigg]$.


\section{Supplement. Implementation Details for Rust's model}
\label{sec:rust_implement_details}

\subsection{\citet{NoretsTang2013}}
\citet{NoretsTang2013} showed that if we assume that the distribution is unknown, then the preference-parameter $\theta = (\theta_0,\theta_1)$ is only set-identified and proposed an algorithm to compute the identified set. In Section 5.1 of their paper, they applied the method to the Rust's model. In their paper, the following setting of the utility function is used:
\begin{align*}
u(x,0) &= \theta_0 + \theta_1 x +\Delta \epsilon, \\
u(x,1) &=0,
\end{align*}
where $\theta_0$ and $\theta_1$ are the data generating values of preference parameters. The term $ \Delta \epsilon = \epsilon_1 - \epsilon_2$ follows some unknown distribution $F$, which is assumed to have the same location-scale normalization as the logistic distribution. Specifically, $\int z dF(z) = 0$ and $\int_{M_F}^{\infty} z dF(z) = \log 2$, where $M_F$ is the median of $F$ or $F(M_F) = 0.5$. 
Figures \ref{fig_Rust_id} and \ref{fig_Rust_id3D} show the identified set of preference-parameters computed by their algorithm. For each $\theta$ inside of the black lines, there is a corresponding unknown distribution $F$ such that the pair $(\theta,F)$ implies the true vector of CCPs.

\subsection{MCMC convergence} 
Let $N$ be the number of samples per state. We generate the data by multiplying $N$ by the true CCP. We run MCMC with variable $m$ with $N \in \{ 3, 10 \}$. 
After each jump proposal block, we run 10 iterations of the HMC block. We obtained 500,000 draws. 
When checking for convergence of the chain, we have to be careful because we are using a mixture model. 
It is well-known that if the likelihood of a mixture model with $m$ components has one mode for a fixed labeling of the components, then it can have $m!$ modes because the likelihood is invariant to a re-labeling of the components and there are $m!$ ways to label components.

Although it is not possible to empirically detect label-switching(s) of our chain, if there was label-switching(s), it would be misleading, for example, to focus our attention on $\mu_1$ when checking convergence of the chain. 
\citet{Geweke:06} points out that this is not a problem as long as both the object of interest and the prior distribution are permutation-invariant. 
Note that we are using exchangeable priors and the object of our interest, the density estimate, is also permutation-invariant. 
We conduct a convergence diagnosis on the mean of $\epsilon$ or $\sum_{k=1}^m \omega_k \mu_k$ which is a permutation-invariant object. 
To check the convergence we perform the mean equality test for the first 10\% and the last 50\% of the samples. 
We conclude that the HMC samples for both $N=3$ and $N=10$ come from stationary distributions after a burn-in period of 10,000 draws.
 
\FloatBarrier


\begin{figure}[ht] 
  \begin{subfigure}[b]{\linewidth}
    \centering
    \includegraphics[width=0.55\linewidth]{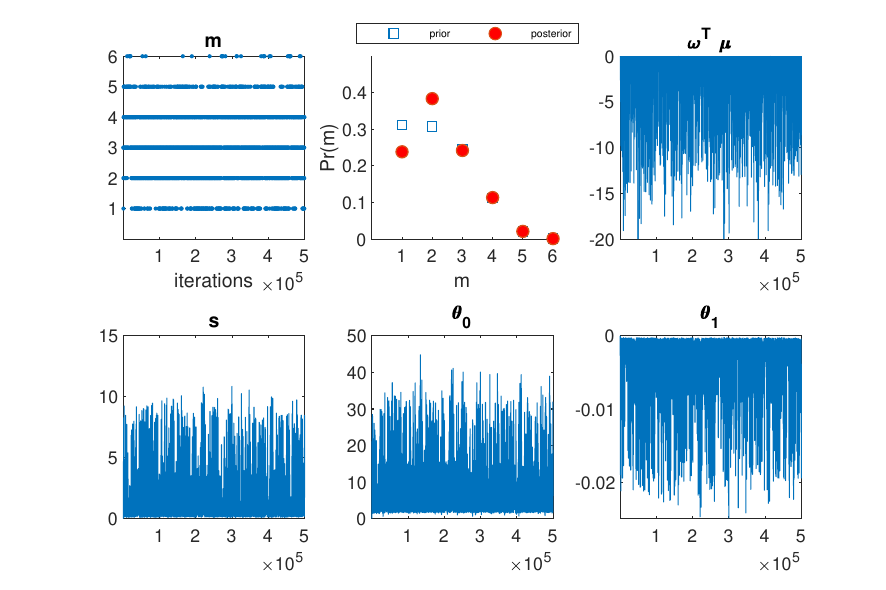} 
    \caption{$N=3$} 
    \vspace{4ex}
  \end{subfigure}
         \hfill
  \begin{subfigure}[b]{\linewidth}
    \centering
    \includegraphics[width=0.55\linewidth]{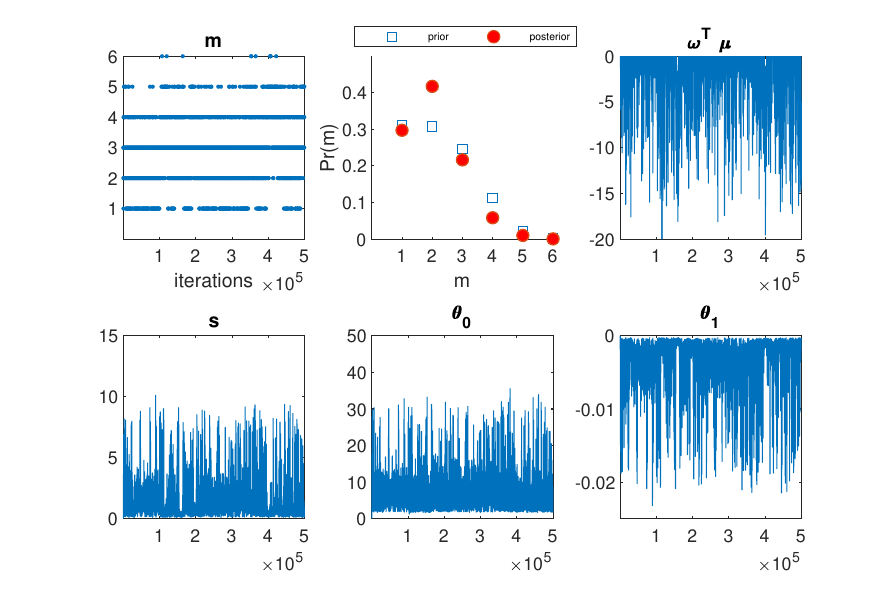} 
    \caption{$N=10$} 
  \end{subfigure}
  \caption{Posterior draws from Rust example. 
  Trace plot of $m$ (upper-left), p.m.f. of $m$ (upper-middle), trace plots of $\sum_{k=1}^m \omega_k \mu_k$ (upper-right), $s$ (bottom-left), $\theta_0$ (bottom-middle), and $\theta_1$ (bottom-right). }
  \label{fig_Rust_trace} 
\end{figure}

\FloatBarrier

\section{Supplement. Implementation Details for Gilleskie's Model}
\subsection{More on State Transitions}\label{state}

The individual contracts an illness and moves to the state $x=(1,0,0)$ with probability 
$\pi^S(H) = 1/[1+\exp( \delta'H)]$.
In each illness period $t \in \{1,\ldots,T\}$, the individual recovers and returns to the state of being well with probability 
$\pi^W(x_t,d_t) = \exp( \eta^T E(x_t,d_t) )/[1+\exp( \eta^T E(x_t,d_t) )]$, 
where $\eta^T E(x_t,d_t)
=
\eta_0 
+ \eta_1v_{t+1} +\eta_2 (v_{t+1} )^2
+ \eta_3a_{t+1} + \eta_4 (a_{t+1})^2 
+ \eta_5 v_{t+1}a_{t+1}
+\eta_6 t +\eta_7 t^2 +\eta_8 t^3
+\xi' H$.

As shown in Table \ref{num_states}, at the beginning of time $t$, there are $t^2$ possible values of the current state. 
\begin{table}[h]
\begin{center}
\begin{tabular}{ c c c }
\hline
 time & value & \# possible values \\ 
\hline 
 t=0 & $(0,0,0)$ & 1 \\ 
 t=1 & $(1,0,0)$ & 1 \\  
 t=2 & $(2,v,a)$ & $2^2=4$ ($v=0,1,a=0,1$) \\  
 t=3 & $(3,v,a)$ & $3^2=9$ ($v=0,1,2,a=0,1,2$) \\  
 \vdots & \vdots & \vdots \\
 t & $(t,v,a)$ & $t^2$ ($v=0,1,\ldots,t-1,a=0,1,\ldots,t-1$) \\  
 \vdots & \vdots & \vdots \\
 T & $(T,v,a)$ & $T^2$ ($v=0,1,\ldots,T-1,a=0,1,\ldots,T-1$) \\   
 \hline 
 Total &  & $K =1+1+2^2+3^2+\cdots T^2=1+\sum_{j=1}^Tj^2=1+\frac{T(T+1)(2T+1)}{6}$ \\ 
\end{tabular}
\end{center}

\caption{State transitions}
\label{num_states}
\end{table}

\FloatBarrier
We summarize the state transitions below. 

\begin{enumerate}
\item The initial state is $(0,0,0)$.
\item If the state is currently at $(0,0,0)$,
\begin{itemize}
\item wp $\pi^S$, $(0,0,0) \to (1,0,0),$
\item wp $1-\pi^S$, $(0,0,0) \to (0,0,0)$.
\end{itemize}

\item In general, for $t=1,\ldots,T-1$, if currently at $(t,v_t,a_t)$,

\begin{itemize}
\item wp $1-\pi^W(x_t,d_t)$,
\begin{equation*}
    (t,v_t,a_t) \to
    \begin{cases}
      (t+1,v_t,a_t), & \text{if}\ \ d_t=0 \\
      (t+1,v_t+1,a_t), & \text{if}\ \ d_t=1 \\
      (t+1,v_t,a_t+1), & \text{if}\ \ d_t=2 \\
      (t+1,v_t+1,a_t+1), & \text{if}\ \ d_t=3 
    \end{cases}
\end{equation*}	
\item wp $\pi^W(x_t,d_t)$, $(t,v_t,a_t) \to (0,0,0)$.
\end{itemize}
 
\item The end of an episode
\begin{align*}
(T,v,a) &\to (0,0,0) \quad \forall v, a,\\
(t,T-1,a) &\to (0,0,0) \quad \forall t, a,\\
(t,v,T-1) &\to (0,0,0) \quad \forall t,v.
\end{align*}

\end{enumerate}

Figure \ref{G_example} shows the state transition matrices when $T=2$. 
In this case, the possible states are $x=(t,v,a) \in \{ 000, 100,   200, 201,210,211 \}$. The number of states is $K=6$. 

\begin{figure}
\begin{center}

\resizebox{0.9\columnwidth}{!}{%
\begin{tabular}{ c | c c c c c c}
  \multicolumn{7}{c}{\bm{$d=0$}} \\
      & 000 & 100 &200 &201&210&211 \\ 
 \hline
 000 & $1-\pi^S$ & $\pi^S$  &0 &0&0&0\\  
 100 & $\pi^W(100,0)$ & 0 &$1-\pi^W(100,0)$  &0&0&0\\  
 200 & 1 & 0  &0 &0&0 &0\\  
 201 & 1 & 0  &0 &0&0&0 \\
  210 & 1 & 0  &0 &0&0&0 \\
    211 & 1 & 0  &0 &0&0&0 
 \end{tabular}
\quad 
\begin{tabular}{ c | c c c c c c}
  \multicolumn{7}{c}{\bm{$d=1$}} \\
      & 000 & 100 &200 &201&210&211 \\ 
 \hline
 000 & $1-\pi^S$ & $\pi^S$  &0 &0&0&0\\  
 100 & $\pi^W(100,1)$ & 0 &0  &0&$1-\pi^W(100,1)$&0\\  
 200 & 1 & 0  &0 &0&0 &0\\  
 201 & 1 & 0  &0 &0&0&0 \\
  210 & 1 & 0  &0 &0&0&0 \\
      211 & 1 & 0  &0 &0&0&0 
 \end{tabular}
 }
\quad 
\resizebox{0.9\columnwidth}{!}{%
\begin{tabular}{ c | c c c c c c}
  \multicolumn{7}{c}{\bm{$d=2$}} \\
      & 000 & 100 &200 &201&210&211 \\ 
 \hline
 000 & $1-\pi^S$ & $\pi^S$  &0 &0&0&0\\  
 100 & $\pi^W(100,2)$ & 0 &0  &$1-\pi^W(100,2)$&0&0\\  
 200 & 1 & 0  &0 &0&0 &0\\  
 201 & 1 & 0  &0 &0&0&0 \\
  210 & 1 & 0  &0 &0&0&0 \\
      211 & 1 & 0  &0 &0&0&0   
 \end{tabular}
\quad 
\begin{tabular}{ c | c c c c c c}
  \multicolumn{7}{c}{\bm{$d=3$}} \\
      & 000 & 100 &200 &201&210&211 \\ 
 \hline
 000 & $1-\pi^S$ & $\pi^S$  &0 &0&0&0\\  
 100 & $\pi^W(100,3)$ & 0 &0  &0&0&$1-\pi^W(100,3)$\\  
 200 & 1 & 0  &0 &0&0 &0\\  
 201 & 1 & 0  &0 &0&0&0 \\
  210 & 1 & 0  &0 &0&0&0 \\
      211 & 1 & 0  &0 &0&0&0 
 \end{tabular}
 }
\end{center}
\caption{Example of transition matrices with $T=2$.}
\label{G_example}
\end{figure}

\FloatBarrier
\subsection{More on Utility Function}\label{sec:utility_derivation}
The system of utility functions in \citet{Gilleskie_Eca:98} can be written as follows
\begin{alignat*}{4}
u(d_t=1,x_t,\epsilon_t) 
&=  \theta_1  +&&\theta_41(t=0)  &&+\theta_6 C(x_t,1)1(t>0)   &&+ \epsilon_{t1},  \\
u(d_t=2,x_t,\epsilon_t) 
&=  \theta_2   +&&\theta_41(t=0)  &&+\theta_6 C(x_t,2)1(t>0)   &&+\epsilon_{t2},\\
u(d_t=3,x_t,\epsilon_t) 
&=  \theta_3   +&&\theta_41(t=0)    &&+\theta_6 C(x_t,3)1(t>0)   &&+\epsilon_{t3},  \\
u(d_t=0,x_t,\epsilon_t) 
&=                &&\theta_51(t=0)                       &&+ \theta_6 C(x_t,0)1(t>0).                    
\end{alignat*} 
Since we do not restrict
the location and scale of $\epsilon_{tj}$, $j=1,2,3$, the values of the intercepts $\theta_j$, $j=1,2,3$ can be set to arbitrary values. 
The coefficient $\theta_4$ is set to  a large negative value so that $d_t=0$ is always chosen when $t=0$ (i.e.,\ the individual is well) and 
$\theta_5$ can be set to an arbitrary positive value.
Define $\bm{\theta} = \left( \theta_1,\theta_2, \theta_3,\theta_4,\theta_5, \theta_6 \right)^\prime$. Then 
\begin{alignat*}{14}
u(d_t=1,x_t,\epsilon_t) 
&= \bigg( &&1, &&0, &&0,    &&1(t=0),  &&0, &&C(x_t,1)1(t>0) \bigg) \bm{\theta}  &&+\epsilon_{t1} &&=\bm{Z_1}(x_t)\bm{\theta}+\epsilon_{t1}, \\
u(d_t=2,x_t,\epsilon_t) 
&= \bigg(&&0, &&1, &&0,     &&1(t=0),  &&0,  &&C(x_t,2)1(t>0)  \bigg) \bm{\theta}  &&+\epsilon_{t2} &&=\bm{Z_2}(x_t)\bm{\theta}+\epsilon_{t2},\\
u(d_t=3,x_t,\epsilon_t) 
&= \bigg(&&0, &&0 ,&&1,     && 1(t=0), &&0, &&C(x_t,3)1(t>0)   \bigg) \bm{\theta}  &&+\epsilon_{t3} &&=\bm{Z_3}(x_t)\bm{\theta}+\epsilon_{t3},\\
u(d_t=0,x_t,\epsilon_t) 
&= \bigg(&&0,&&0, &&0,      &&0,           &&1(t=0), &&C(x_t,0)1(t>0) \bigg) \bm{\theta}  &&   &&=\bm{Z_0}(x_t)\bm{\theta},
\end{alignat*}
where
\small
\begin{align*}
\bm{Z_1}(x)=\big(Z_{11}(x),Z_{12}(x),Z_{13}(x),Z_{14}(x),Z_{15}(x) ,Z_{16}(x)  \big) 
= \bigg( 1,0, 0, 1(t=0),0, C(x,1)1(t>0) \bigg), \\
\bm{Z_2}(x)=\big(Z_{21}(x),Z_{22}(x),Z_{23}(x),Z_{24}(x),Z_{25}(x),Z_{26}(x)  \big) 
= \bigg(0,1,0,1(t=0),0,C(x,2)1(t>0)  \bigg), \\
\bm{Z_3}(x)=\big(Z_{31}(x),Z_{32}(x),Z_{33}(x),Z_{34}(x),Z_{35}(x),Z_{36}(x) \big) 
=\bigg(0,0,1, 1(t=0),0,C(x,3)1(t>0) \bigg),\\
\bm{Z_0}(x)=\big(Z_{01}(x),Z_{02}(x),Z_{03}(x),Z_{04}(x),Z_{05}(x),Z_{06}(x)  \big) 
=\bigg(0,0,0,0, 1(t=0),C(x,0)1(t>0)  \bigg).
\end{align*}
\normalsize
Letting $\bm{Z}_{j,1:5}(x)$ to be the first 5 columns of $\bm{Z_j}(x)$, $j=0,1,2,3$, we can also write
\begin{alignat*}{3}
u(1,x_t,\epsilon_t)&=\bm{Z}_{1,1:5}(x_t)\bm{\hat{\theta}}_{1:5}  +\theta_6 Z_{16}(x_t) +\epsilon_{t1}, \\
u(2,x_t,\epsilon_t)&=\bm{Z}_{2,1:5}(x_t)\bm{\hat{\theta}}_{1:5}  +\theta_6 Z_{26}(x_t) +\epsilon_{t2}, \\
u(3,x_t,\epsilon_t)&=\bm{Z}_{3,1:5}(x_t)\bm{\hat{\theta}}_{1:5}  +\theta_6 Z_{36}(x_t) +\epsilon_{t3}, \\
u(0,x_t,\epsilon_t)&=\bm{Z}_{0,1:5}(x_t)\bm{\hat{\theta}}_{1:5}  +\theta_6 Z_{06}(x_t),  &
\end{alignat*}
where $\bm{\hat{\theta}}_{1:5}$ is the vector of pre-specified coefficients and $\theta_{6}$ is a parameter to be estimated. 

\subsection{Data Generating Parameters}
\label{sec:giil_data_gen}
We compute the data generating CCPs based on the following values, mostly based on estimates in  \citet{Gilleskie_Eca:98}
for Type 2 illness with some adjustments so that the expected number of doctor visits and work absences roughly match with Gilleskie's sample. 
We use the following data-generating parameter values: 
$T=8$, 
$\beta=0.9997$, 
$\theta_1=-1.25$, 
$\theta_2= -0.83$, 
$\theta_3= -2.08$,
$\theta_4=-10000$,
$\theta_5=1$,
$\theta_6=0.0469$,
$Y=100$, 
$PC=15$,
$L=0.7$, 
$\phi_1=5.6$, and $\phi_2=-1.75$. 
For the state transition parameters
$\eta$,  
$\delta$, and 
$\xi$, we use the estimated values in Gilleskie (1998). 
We compute the data generating CCPs based on the parameter values mentioned above assuming that the data generating distribution of the utility shock is the following two-component mixture of extreme value distributions:
$
\sum_{k=1}^2
\omega_k
\phi \left( \epsilon; \mu_k,\sigma_k \right)
$
where 
$\omega_1=0.5568,\omega_2=0.4432$,
$\mu_1=[-0.4683,3.4628 ,-0.0914 ]',\mu_2=[0.9798,-2.2437 , 1.3496 ]' $,
$\sigma_1=3.7045$, and $\sigma_2=0.6378$. 
We then use the CCPs and the state-transition probabilities to sequentially generate 100 illness episodes.

\subsection{MCMC Plots from Gilleskie's Application}
\label{appsec:traceplots}

\FloatBarrier

\begin{figure}[ht] 
  \begin{subfigure}[b]{0.6\linewidth}
    \centering
    \includegraphics[width=\linewidth]{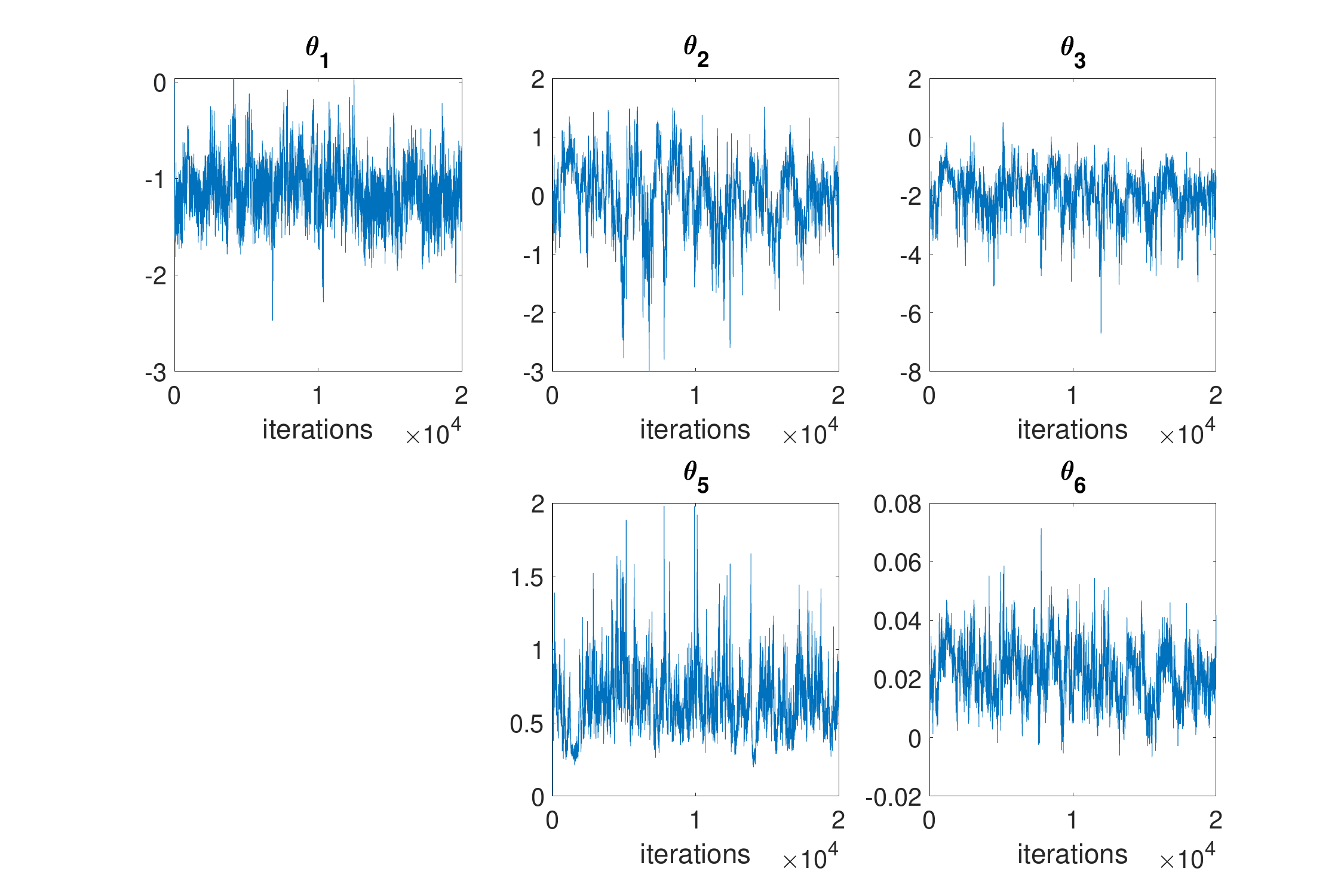} 
\caption{\small{Trace plots of utility parameters}}
  \end{subfigure}
  \hfill
  \begin{subfigure}[b]{0.6\linewidth}
    \centering
    \includegraphics[width=\linewidth]{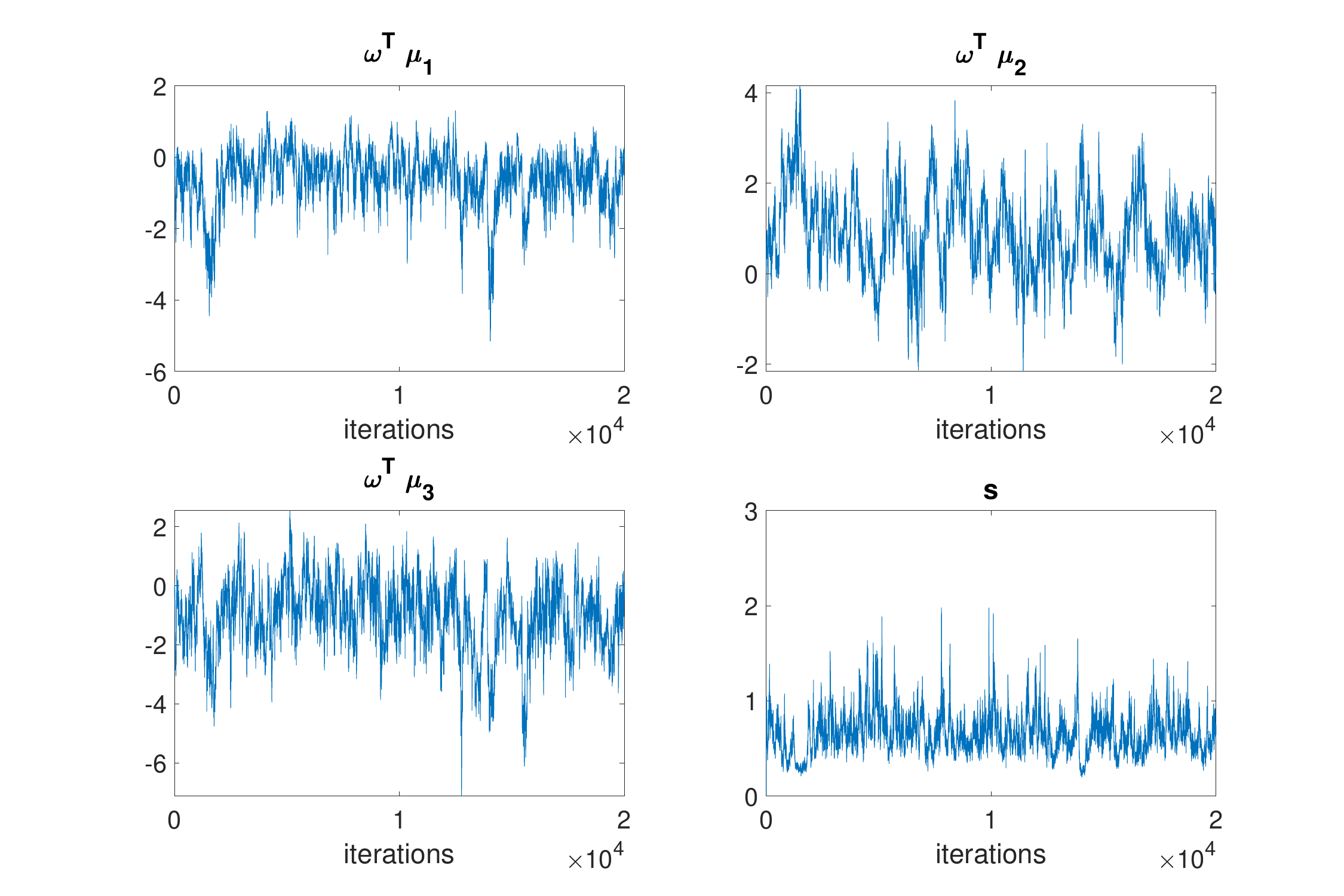} 
\caption{\small{Trace plots of 
  $\sum_{k=1}^m \omega_k \mu_{1k}$ (upper-left), 
    $\sum_{k=1}^m \omega_k \mu_{2k}$ (upper-right), 
  $\sum_{k=1}^m \omega_k \mu_{3k}$ (bottom-left), 
  $s$ (bottom-right), which is defined in Section \ref{subsection:normalizations}.
   }}
  \end{subfigure}
\caption{\small{Trace plots in Gilleskie model }}
\label{fig_Gilleskie_trace}
\end{figure}

\FloatBarrier


\subsection{Prior sensitivity check for Gilleskie's model}\label{sensitivity}
We consider two additional sets of priors for estimating  Gilleskie's model in Section \ref{section:gilleskie} and present results of estimation and counterfactual analysis. 
Despite slight differences, the overall findings remain the same.  
That is,  the credible intervals for the identified sets introduced in Section \ref{sec:inference_IS} include the true counterfactual value of mean doctor visits, and the HPD intervals get very close to it. 
In contrast,  
the confidence intervals computed via MLE assuming the dynamic logit misses the true value. 

\subsubsection{Prior sensitivity  check 1}
First, we consider increasing the prior variances for the component specific scale parameters $\sigma_k$'s and the preference parameter $\theta_6$.
The priors are now specified as follows, 
$\underbar{a}=10$, $A_m=0.05$, and $\tau=5$,
$\mu_{jk} \sim N(0, 3^2)$, 
$\log \sigma_{k}  \sim N(0, 1)$, 
$\log \sigma \sim N(0,0.01^2)$, and 
$\theta_6 \sim N(0,5^2)$. 
\FloatBarrier
\begin{table}[h]
  \centering
  \def\sym#1{\ifmmode^{#1}\else\(^{#1}\)\fi}%
  \begin{tabular}{l*{6}{c}}
    \toprule
    & \multicolumn{3}{c}{MLE under logit}  & \multicolumn{3}{c}{Semiparametric Bayes} \\
    \cmidrule(lr){3-4}\cmidrule(lr){5-7}
    & \multicolumn{1}{c}{True} 
    & \multicolumn{1}{c}{Est.} & \multicolumn{1}{c}{95\%CI (length)}  
    & \multicolumn{1}{c}{Est.} & \multicolumn{1}{c}{95\%HPD (length)}  
                                   & \multicolumn{1}{c}{$\hat{B}^{I_{E(v)}}_{0.95}$ (length)}      \\    
    \midrule
        \addlinespace
     $E(v)$  & 1.49  &  1.35    &  [1.15,    1.55]  (0.40)  & 1.37 &  [1.16,    1.58]    (0.42)   &  [1.02,    1.71]    (0.68)   \\
        \addlinespace
    $c.f.E(v)$ & 2.31  &  1.70      &  [1.46,    1.94]   (0.48) &  1.72   &  [1.18,    2.28]      (1.09)    &  [1.13,    4.97]    (3.84) \\
        \addlinespace        
    \bottomrule
  \end{tabular}
  
\caption{Estimated (Est.)  $E(v)$ and  $c.f.E(v)$: 
the MLE with its 95\% confidence interval and
the posterior mean with the 95\% HPD interval and $\hat{B}^{I_{E(v)}}_{0.95}$. }
\label{table_estimation_check1}  
\end{table}
\FloatBarrier

\begin{figure}[ht] 
  \begin{subfigure}[b]{0.4\linewidth}
    \centering
    \includegraphics[width=\linewidth]{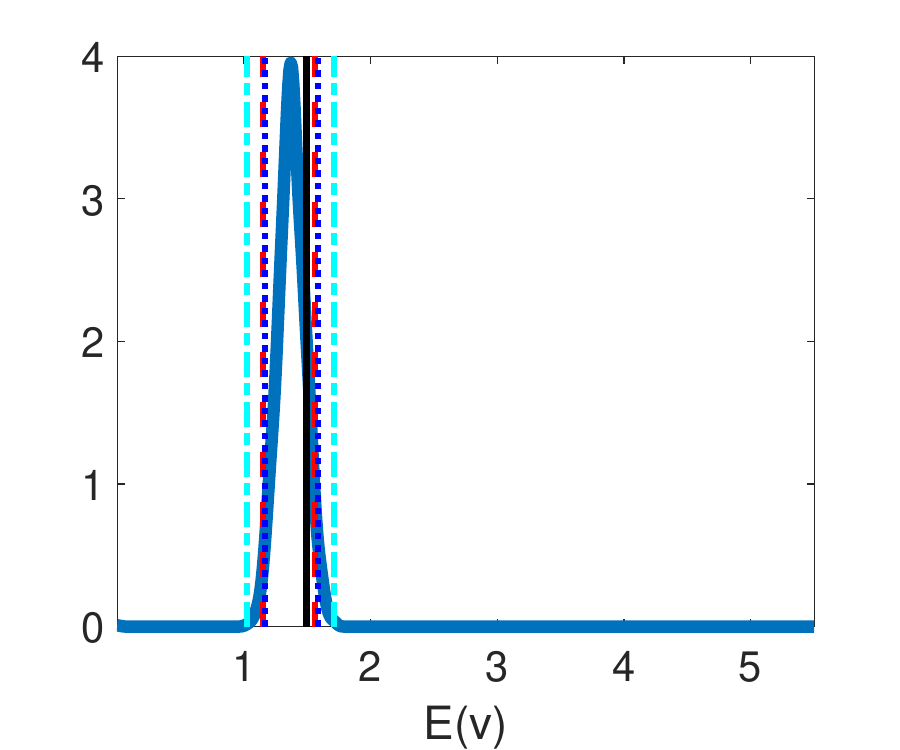} 
\caption{\small{Mean doctor visits }}
  \end{subfigure}
  \begin{subfigure}[b]{0.4\linewidth}
    \centering
    \includegraphics[width=\linewidth]{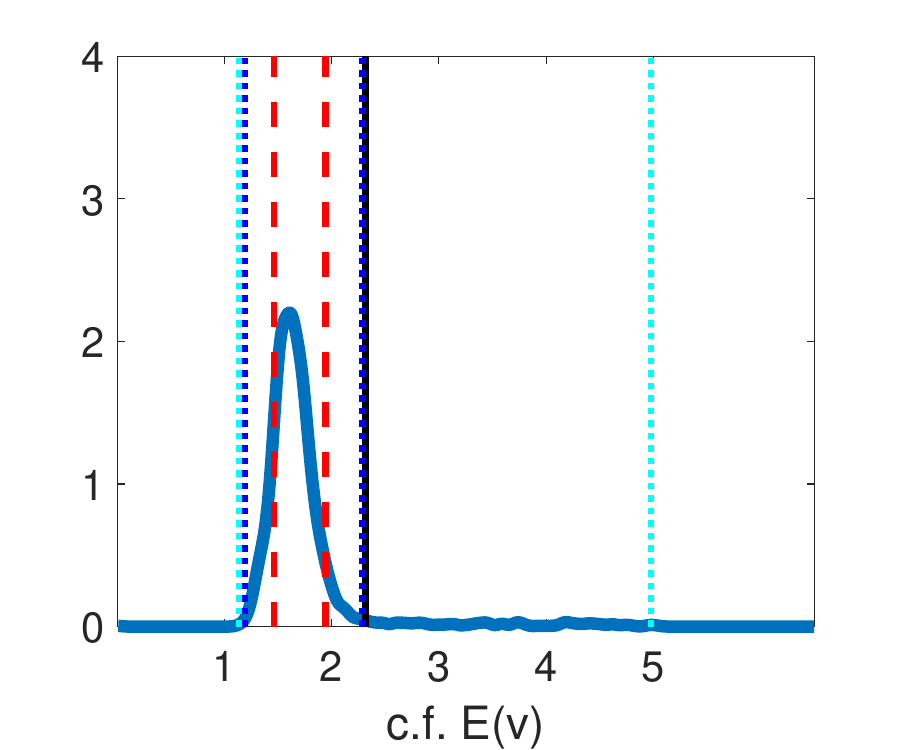} 
\caption{\small{ Mean c.f. doctor visits }}
  \end{subfigure}

\caption{\small{Actual and counterfactual (c.f.) expected number of doctor visits. 
Posterior density (blue solid),
95\% confidence interval (red dashed), 
95\% HPD interval (blue dotted),
$\hat{B}^{I_{E(v)}}_{0.95}$ (light blue dash-dotted),
and data generating value (black solid). }}
\label{fig_Gilleskie_EvEa_check1}
\end{figure}
\FloatBarrier

\subsubsection{Prior sensitivity  check 2}
Second, we consider using a set of priors similar to the one used to estimate  Rust model in Section \ref{section:rust}. 
The priors are now specified as follows, 
$\underbar{a}=10$, $A_m=0.05$, and $\tau=5$,
$\mu_{jk} \sim 0.5N(2,2^2) + 0.5N(-3,2^2)$, 
$\log \sigma_{k}  \sim 0.4N(0,1) + 0.6N(-1,1)$, 
$\log \sigma \sim N(0,0.01^2)$, and 
$\theta_6 \sim N(0,3^2)$.

\FloatBarrier

\begin{table}[h]
  \centering
  \def\sym#1{\ifmmode^{#1}\else\(^{#1}\)\fi}%
  \begin{tabular}{l*{6}{c}}
    \toprule
    & \multicolumn{3}{c}{MLE under logit}  & \multicolumn{3}{c}{Semiparametric Bayes} \\
    \cmidrule(lr){3-4}\cmidrule(lr){5-7}
    & \multicolumn{1}{c}{True} 
    & \multicolumn{1}{c}{Est.} & \multicolumn{1}{c}{95\%CI (length)}  
    & \multicolumn{1}{c}{Est.} & \multicolumn{1}{c}{95\%HPD (length)}  
                                   & \multicolumn{1}{c}{$\hat{B}^{I_{E(v)}}_{0.95}$ (length)}      \\   
    \midrule
        \addlinespace
     $E(v)$  & 1.49  &  1.35    &  [1.15,    1.55]  (0.40)  & 1.37 &  [1.16,    1.57]    (0.41)   &  [1.02,    1.77]    (0.75)   \\
        \addlinespace
    $c.f.E(v)$ & 2.31  &  1.70      &  [1.46,    1.94]   (0.48) &  1.80   &  [1.16,    3.04]      (1.87)    &  [1.08,    5.16]    (4.07) \\        
        \addlinespace        
    \bottomrule
  \end{tabular}
  
\caption{Estimated (Est.)  $E(v)$ and  $c.f.E(v)$: 
the MLE with its 95\% confidence interval and
the posterior mean with the 95\% HPD interval and $\hat{B}^{I_{E(v)}}_{0.95}$. }
\label{table_estimation_check2}  
\end{table}
\FloatBarrier
\begin{figure}[ht] 
  \begin{subfigure}[b]{0.4\linewidth}
    \centering
    \includegraphics[width=\linewidth]{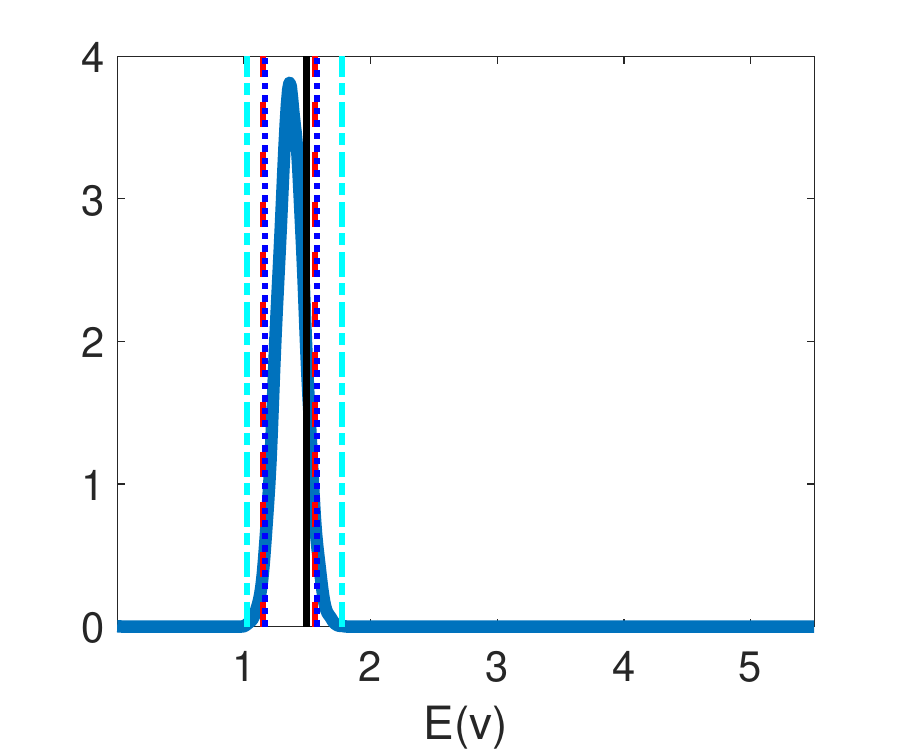} 
\caption{\small{Mean doctor visits }}
  \end{subfigure}
  \begin{subfigure}[b]{0.4\linewidth}
    \centering
    \includegraphics[width=\linewidth]{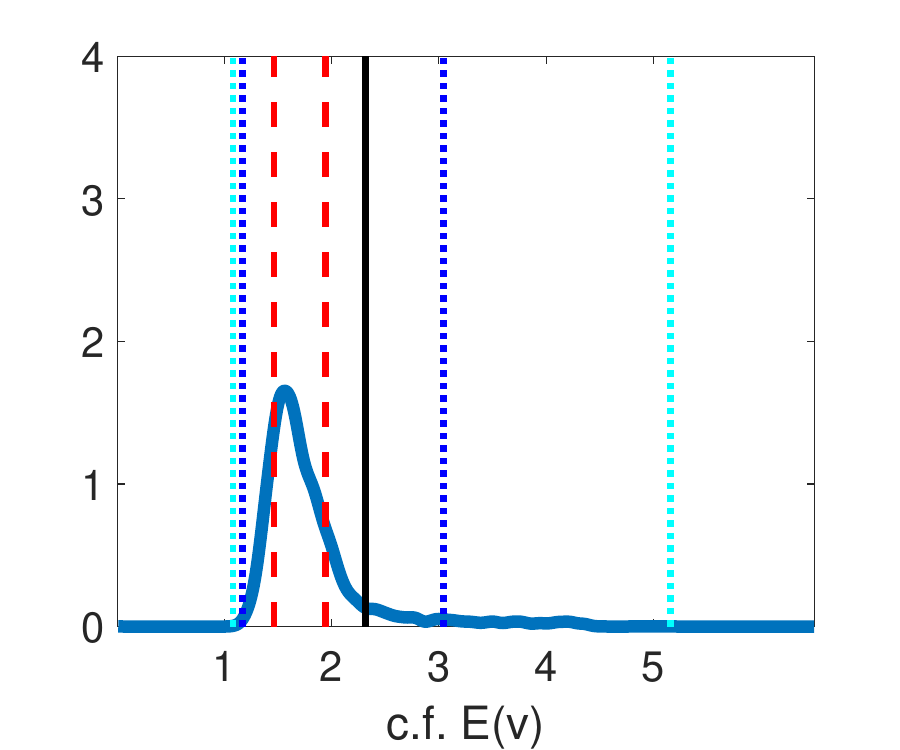} 
\caption{\small{ Mean c.f. doctor visits }}
  \end{subfigure}

\caption{\small{Actual and counterfactual (c.f.) expected number of doctor visits. 
Posterior density (blue solid),
95\% confidence interval (red dashed), 
95\% HPD interval (blue dotted),
$\hat{B}^{I_{E(v)}}_{0.95}$ (light blue dash-dotted),
and data generating value (black solid). }}
\label{fig_Gilleskie_EvEa_check2}
\end{figure}
\FloatBarrier



\subsection{Results of Gilleskie's model when the data generating process is dynamic logit}\label{gilleskie_logit}
Recall that in Section \ref{section:gilleskie}, we let the data generating distribution of the utility shocks to be a two-component mixture of extreme value distributions. Here we show results when it is logistic, and therefore the dynamic logit model is correctly specified.
Table \ref{table_estimation_logit} shows estimation results.
The overall findings are similar to the case with a mixture data generating distribution: 
the credible intervals of our semiparametric approach are much wider than the confidence intervals based on the dynamic logit assumption.
However, now the confidence interval includes the true counterfactual value of doctor visits $E(v)$ (See Figure \ref{fig_Gilleskie_EvEa_logit}) while it did not under the mixture data generating distribution in Section \ref{section:gilleskie}. This is probably a consequence of the dynamic logit model being correctly specified in this particular example. 
We conducted some prior sensitivity checks in Section \ref{sensitivity_logit} and confirm that the general findings are the same.

\FloatBarrier
\begin{table}[h]
  \centering
  \def\sym#1{\ifmmode^{#1}\else\(^{#1}\)\fi}%
  \begin{tabular}{l*{6}{c}}
    \toprule
    & \multicolumn{3}{c}{MLE under logit}  & \multicolumn{3}{c}{Semiparametric Bayes} \\
    \cmidrule(lr){3-4}\cmidrule(lr){5-7}
    & \multicolumn{1}{c}{True} 
    & \multicolumn{1}{c}{Est.} & \multicolumn{1}{c}{95\%CI (length)}  
    & \multicolumn{1}{c}{Est.} & \multicolumn{1}{c}{95\%HPD (length)}  
                                   & \multicolumn{1}{c}{$\hat{B}^{I_{E(v)}}_{0.95}$ (length)}      \\      
    \midrule
        \addlinespace
     $E(v)$  & 1.53 &  1.43     &  [1.22,    1.65]  (0.42)  & 1.45 &  [1.24,    1.66]    (0.41)   &  [1.01,    1.82]    (0.80)   \\
        \addlinespace
    $c.f.E(v)$  & 1.64 &  1.67      &  [1.43,    1.91]   (0.47) &  1.74   &  [1.35,    2.12]      (0.76)    &  [1.18,    3.29]    (2.11) \\
        \addlinespace        
    \bottomrule
  \end{tabular}
  
\caption{Estimated (Est.)  $E(v)$ and  $c.f.E(v)$: 
the MLE with its 95\% confidence interval and
the posterior mean with the 95\% HPD interval and $\hat{B}^{I_{E(v)}}_{0.95}$. }
\label{table_estimation_logit}  
\end{table}

\FloatBarrier

\begin{figure}[ht] 
  \begin{subfigure}[b]{0.4\linewidth}
    \centering
    \includegraphics[width=\linewidth]{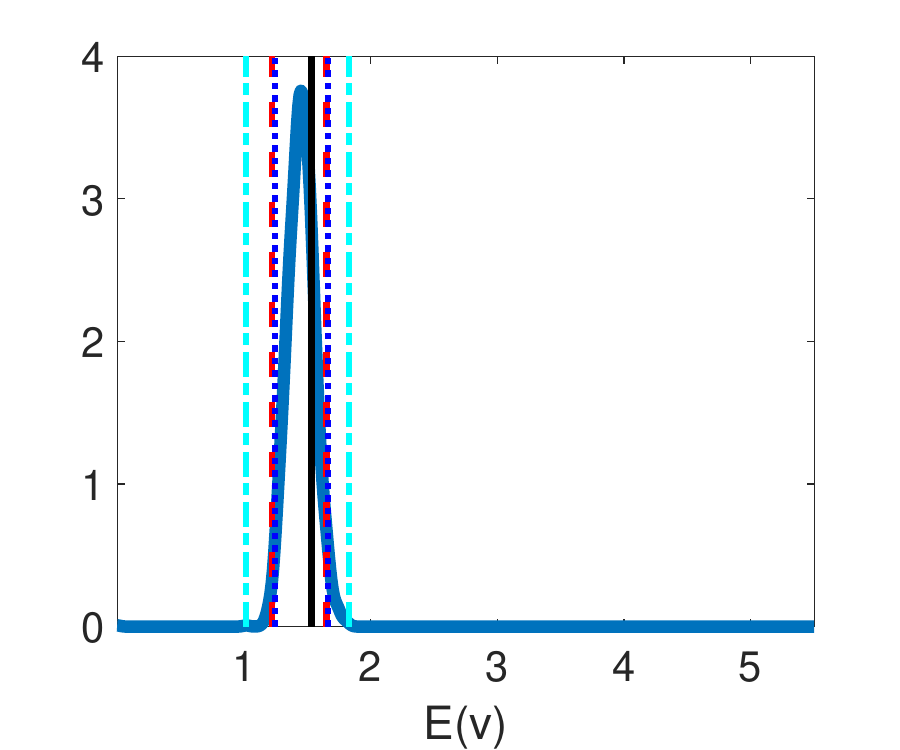} 
\caption{\small{Mean doctor visits }}
  \end{subfigure}
  \begin{subfigure}[b]{0.4\linewidth}
    \centering
    \includegraphics[width=\linewidth]{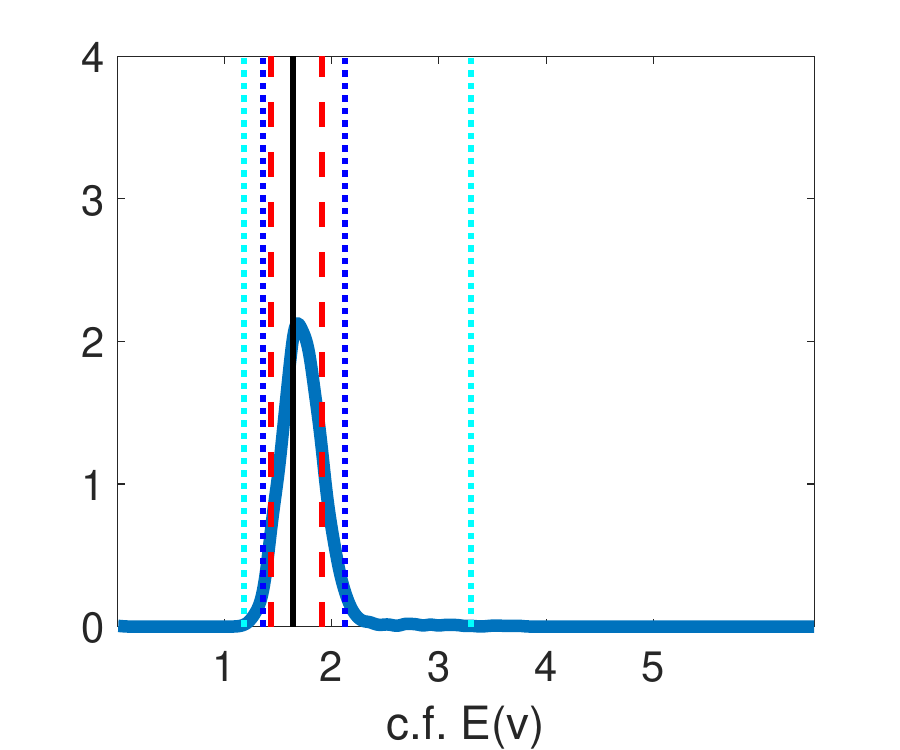} 
\caption{\small{ Mean c.f. doctor visits }}
  \end{subfigure}

\caption{\small{Actual and counterfactual (c.f.) expected number of doctor visits. 
Posterior density (blue solid),
95\% confidence interval (red dashed), 
95\% HPD interval (blue dotted),
$\hat{B}^{I_{E(v)}}_{0.95}$ (light blue dash-dotted),
and data generating value (black solid). }}

\label{fig_Gilleskie_EvEa_logit}
\end{figure}
\FloatBarrier

\subsubsection{Prior sensitivity  checks}\label{sensitivity_logit}

\paragraph{Prior sensitivity  check 1}
First, we consider increasing the prior variances for the component specific scale parameters $\sigma_k$'s and the preference parameter $\theta_6$.
The priors are now specified as follows, 
$\underbar{a}=10$, $A_m=0.05$, and $\tau=5$,
$\mu_{jk} \sim N(0, 3^2)$, 
$\log \sigma_{k}  \sim N(0, 1)$, 
$\log \sigma \sim N(0,0.01^2)$, and 
$\theta_6 \sim N(0,5^2)$. 
The results are presented below.
\FloatBarrier
\begin{table}[h]
  \centering
  \def\sym#1{\ifmmode^{#1}\else\(^{#1}\)\fi}%
  \begin{tabular}{l*{6}{c}}
    \toprule
    & \multicolumn{3}{c}{MLE under logit}  & \multicolumn{3}{c}{Semiparametric Bayes} \\
    \cmidrule(lr){3-4}\cmidrule(lr){5-7}
    & \multicolumn{1}{c}{True} 
    & \multicolumn{1}{c}{Est.} & \multicolumn{1}{c}{95\%CI (length)}  
    & \multicolumn{1}{c}{Est.} & \multicolumn{1}{c}{95\%HPD (length)}  
                                   & \multicolumn{1}{c}{$\hat{B}^{I_{E(v)}}_{0.95}$ (length)}      \\     
    \midrule
        \addlinespace
     $E(v)$  & 1.53 &  1.43     &  [1.22,    1.65]  (0.42)  & 1.45 &  [1.25,    1.66]    (0.41)   &  [1.11,    1.82]    (0.71)   \\
        \addlinespace
    $c.f.E(v)$  & 1.64 &  1.67      &  [1.43,    1.91]   (0.47) &  1.79   &  [1.24,    2.48]      (1.24)    &  [1.20,    3.77]    (2.56) \\    
        \addlinespace        
    \bottomrule
  \end{tabular}
  
\caption{Estimated (Est.)  $E(v)$ and  $c.f.E(v)$: 
the MLE with its 95\% confidence interval and
the posterior mean with the 95\% HPD interval and $\hat{B}^{I_{E(v)}}_{0.95}$. }
\label{table_estimation_logit_check1}  
\end{table}

\FloatBarrier

\begin{figure}[ht] 
  \begin{subfigure}[b]{0.4\linewidth}
    \centering
    \includegraphics[width=\linewidth]{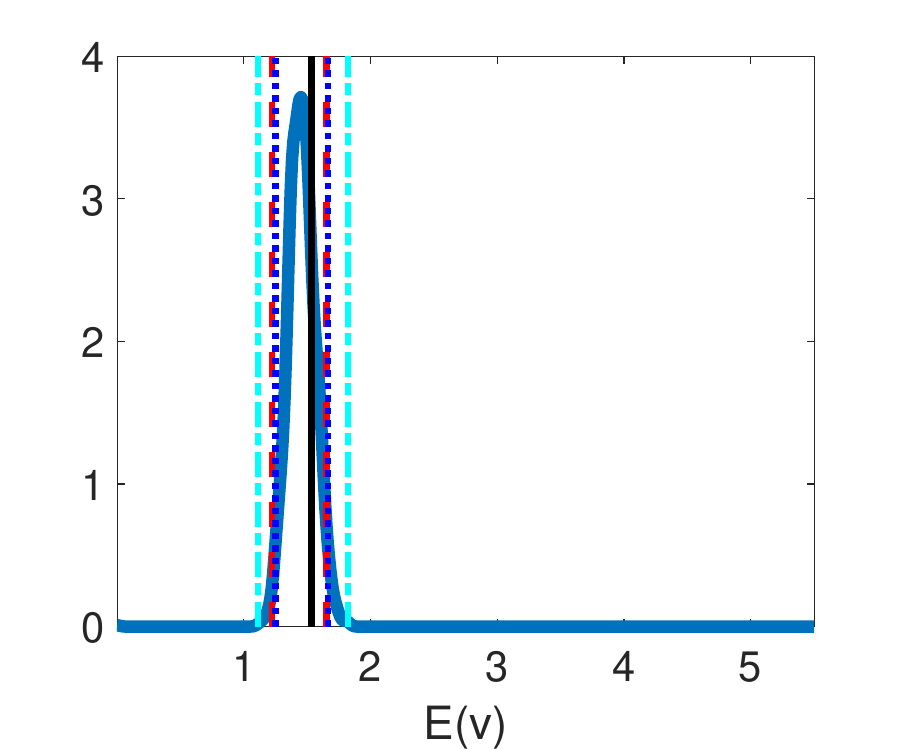} 
\caption{\small{Mean doctor visits }}
  \end{subfigure}
  \begin{subfigure}[b]{0.4\linewidth}
    \centering
    \includegraphics[width=\linewidth]{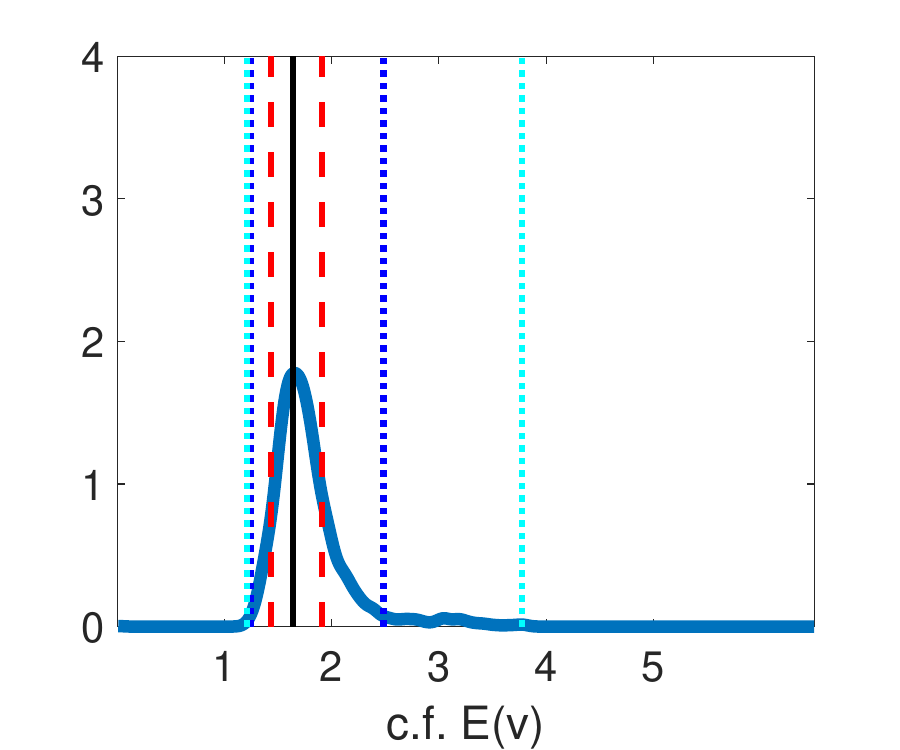} 
\caption{\small{ Mean c.f. doctor visits }}
  \end{subfigure}

\caption{\small{Actual and counterfactual (c.f.) expected number of doctor visits. 
Posterior density (blue solid),
95\% confidence interval (red dashed), 
95\% HPD interval (blue dotted),
$\hat{B}^{I_{E(v)}}_{0.95}$ (light blue dash-dotted),
and data generating value (black solid). }}

\label{fig_Gilleskie_EvEa_logit_check1}
\end{figure}
\FloatBarrier

\paragraph{Prior sensitivity  check 2}
Second, we consider using a set of priors similar to the one used to estimate  Rust model in Section \ref{section:rust}. 
The priors are now specified as follows, 
$\underbar{a}=10$, $A_m=0.05$, and $\tau=5$,
$\mu_{jk} \sim 0.5N(2,2^2) + 0.5N(-3,2^2)$, 
$\log \sigma_{k}  \sim 0.4N(0,1) + 0.6N(-1,1)$, 
$\log \sigma \sim N(0,0.01^2)$, and 
$\theta_6 \sim N(0,3^2)$.
The results are shown below.
\FloatBarrier
\begin{table}[h]
  \centering

  \def\sym#1{\ifmmode^{#1}\else\(^{#1}\)\fi}%
  \begin{tabular}{l*{6}{c}}
    \toprule
    & \multicolumn{3}{c}{MLE under logit}  & \multicolumn{3}{c}{Semiparametric Bayes} \\
    \cmidrule(lr){3-4}\cmidrule(lr){5-7}
    & \multicolumn{1}{c}{True} 
    & \multicolumn{1}{c}{Est.} & \multicolumn{1}{c}{95\%CI (length)}  
    & \multicolumn{1}{c}{Est.} & \multicolumn{1}{c}{95\%HPD (length)}  
                                   & \multicolumn{1}{c}{$\hat{B}^{I_{E(v)}}_{0.95}$ (length)}      \\        
    \midrule
        \addlinespace
     $E(v)$  & 1.53 &  1.43     &  [1.22,    1.65]  (0.42)  & 1.44 &  [1.23,    1.65]    (0.41)   &  [1.12,    1.79]    (0.66)   \\
        \addlinespace
    $c.f.E(v)$  & 1.64 &  1.67      &  [1.43,    1.91]   (0.47) &  1.61   &  [1.29,    1.95]      (0.66)    &  [1.16,    2.17]    (1.00) \\           
        \addlinespace        
    \bottomrule
  \end{tabular}
  
\caption{Estimated (Est.)  $E(v)$ and  $c.f.E(v)$: 
the MLE with its 95\% confidence interval and
the posterior mean with the 95\% HPD interval and $\hat{B}^{I_{E(v)}}_{0.95}$. 
}
\label{table_estimation_logit_check2}

\end{table}

\begin{figure}
  \begin{subfigure}{0.4\linewidth}
    \centering
    \includegraphics[width=\linewidth]{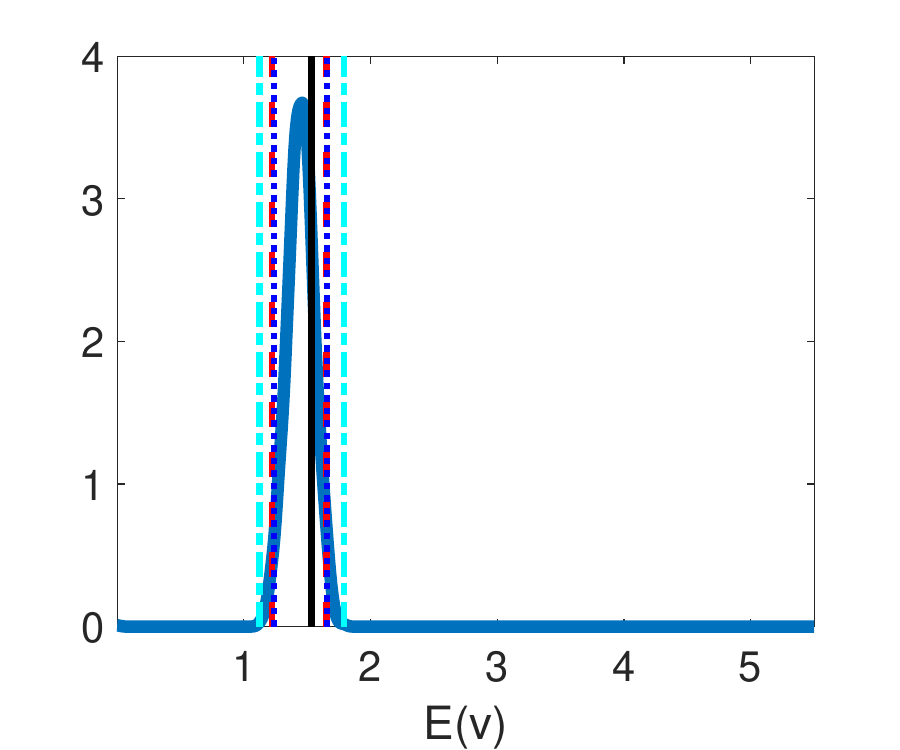} 
\caption{\small{Mean doctor visits }}
  \end{subfigure}
  \begin{subfigure}{0.4\linewidth}
    \centering
    \includegraphics[width=\linewidth]{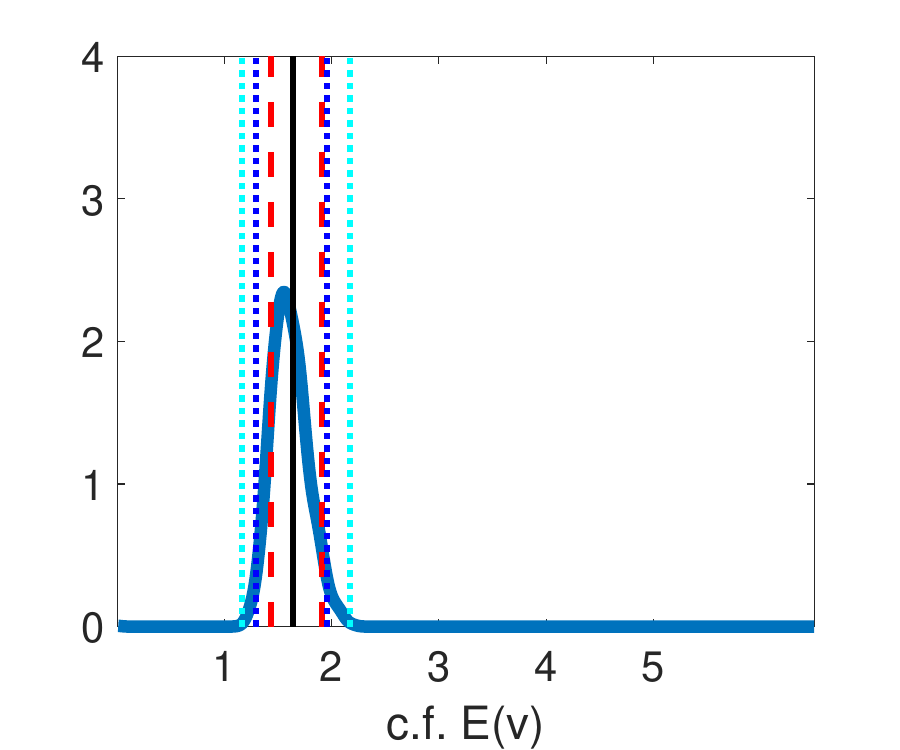} 
\caption{\small{ Mean c.f. doctor visits }}
  \end{subfigure}

\caption{\small{Actual and counterfactual (c.f.) expected number of doctor visits. 
95\% confidence interval (red dashed), 
95\% HPD interval (blue dotted),
$\hat{B}^{I_{E(v)}}_{0.95}$ (light blue dash-dotted),
and data generating value (black solid).
 }}

\label{fig_Gilleskie_EvEa_logit_check2}
\end{figure}
\FloatBarrier

\pagebreak

\end{document}